\PassOptionsToPackage{table}{xcolor}

\documentclass[sigconf]{acmart}

\AtBeginDocument{%
  }

\setcopyright{acmlicensed}
\copyrightyear{2018}
\acmYear{2018}
\acmDOI{XXXXXXX.XXXXXXX}

\acmConference[SIGMOD'25]{SIGMOD International Conference on Management of Data}{June 22-27,
  2025}{Berlin, Germany}
\acmISBN{978-1-4503-XXXX-X/18/06}

\usepackage{multirow}
\usepackage{listings}
\usepackage{amsmath}

\usepackage[capitalise]{cleveref}
\Crefname{algocf}{Algorithm}{Algorithms}
\crefalias{algocfline}{line}

\usepackage{url}
\usepackage[skip=0.5pt]{subcaption}
\usepackage[skip=0pt]{caption}
\usepackage[linesnumbered,ruled,vlined]{algorithm2e}
\usepackage{tikz}
\usetikzlibrary{matrix}

\usepackage[inline]{enumitem}

\renewcommand{\epsilon}{\ve}
\def\ve{\varepsilon}

\usetikzlibrary{trees,arrows.meta,chains,fit,shapes,calc}

\lstnewenvironment{VerbatimText}[1][]{
    
    \lstset{fancyvrb=true,basicstyle=\footnotesize,captionpos=b,xleftmargin=2em,#1}
}{}
\newcommand{\paratitle}[1]{\vspace{1mm}\noindent\textbf{{#1}.}}

\definecolor{mygreen}{rgb}{0,0.6,0}
\definecolor{mygray}{rgb}{0.5,0.5,0.5}
\definecolor{mymauve}{rgb}{0.58,0,0.82}

\newcommand{\ignore}[1]{}

\newcommand*{\rom}[1]{\expandafter\@slowromancap\romannumeral #1@}

  \newcommand{\amir}[1]{}

\definecolor{applegreen}{rgb}{0.55, 0.71, 0.0}

\newcommand{\reva}[1]{{{#1}}}

\newcommand{\revc}[1]{{{#1}}}
\newcommand{\common}[1]{{{#1}}}

\newcommand{\drastic}{\ensuremath{\mathcal{I_\mathsf{D}}}}
\newcommand{\problematic}{\ensuremath{\mathcal{I_\mathsf{P}}}}
\newcommand{\mininconsistency}{\ensuremath{\mathcal{I_{\mathsf{MI}}}}}
\newcommand{\maxconsistency}{\ensuremath{\mathcal{I_\mathsf{MC}}}}

\newcommand{\repair}{\ensuremath{\mathcal{I_\mathsf{R}}}}

\newcommand{\constraintset}{\ensuremath{{\Sigma}}}

\newcommand{\graph}{\ensuremath{{\mathcal{G}_\constraintset^{D}}}}

\newcommand{\node}{\ensuremath{{u}}}

\newcommand{\attrset}{\mathcal{A}}

\usepackage{enumitem}

\newcommand{\opt}{\normalfont{\text{opt}}}
\newcommand{\degreebound}{d_{\normalfont{\text{bound}}}}
\newcommand{\noisydegreebound}{\tilde{d}_{\normalfont{\text{bound}}}}
\newcommand{\degree}{d}
\newcommand{\graphsimple}{\mathcal{G}}

\newcommand{\RNum}[1]{\uppercase\expandafter{\romannumeral #1\relax}}

\newtheorem{theorem}{Theorem}

\newtheorem{proposition}{Proposition}
\newtheorem{observation}{Observation}
\newtheorem{example}{Example}

\newtheorem{definition}{Definition}
\newtheorem{lemma}{Lemma}

\newcommand{\proj}[1]{{\Pi}}
\newcommand{\sel}[1]{{\sigma}}

\newcommand{\cut}[1]{}
\newcommand{\eat}[1]{}

\newcommand{\dom}{{\tt dom}}

\newcommand{\neighbor}{\approx}

\newcommand{\xh}[1]{\textcolor{purple}{(XH: #1)}}

\newcommand{\commenttext}[1]{{\color{gray}{#1}}}

\newcommand{\squishlist}{
	\begin{list}{$\bullet$}
		{
			\setlength{\itemsep}{0pt}
			\setlength{\parsep}{0pt}
			\setlength{\topsep}{3pt}
			\setlength{\partopsep}{0pt}
			\setlength{\leftmargin}{1.5em}
			\setlength{\labelwidth}{1em}
			\setlength{\labelsep}{0.5em} } }

	\newcommand{\squishend}{
\end{list}  }

\newif\ifpaper
\paperfalse   

\begin{document}

\title{Computing Inconsistency Measures Under Differential Privacy}


\author{Shubhankar Mohapatra}
\email{shubhankar.mohapatra@uwaterloo.ca}
\affiliation{%
  \institution{University of Waterloo}
  \country{Canada}}

\author{Amir Gilad}
\authornote{Authors AG, XH, BK have
equal contribution and are listed in alphabetical order}
\email{amirg@cs.huji.ac.il}
\affiliation{
  \institution{Hebrew University}
  \country{Israel}
}

\author{Xi He}
\email{xi.he@uwaterloo.ca}
\authornotemark[1]
\affiliation{%
  \institution{University of Waterloo}
  \country{Canada}
  }

\author{Benny Kimelfeld}
\authornotemark[1]
\email{bennyk@cs.technion.ac.il}
\affiliation{%
  \institution{Technion}
  \country{Israel}
}

\renewcommand{\shortauthors}{Mohapatra et al.}

\begin{abstract}
  Assessing data quality is crucial to knowing whether and how to use the data for different purposes. Specifically, given a collection of integrity constraints, various ways have been proposed to quantify the inconsistency of a database. Inconsistency measures are particularly important when we wish to assess the quality of private data without revealing sensitive information. We study the estimation of inconsistency measures for a database protected under Differential Privacy (DP). Such estimation is nontrivial since some measures intrinsically query sensitive information, and the computation of others involves functions on underlying sensitive data. Among five inconsistency measures that have been proposed in recent work, we identify that two are intractable in the DP setting. The major challenge for the other three is high sensitivity: adding or removing one tuple from the dataset may significantly affect the outcome. To mitigate that, we model the dataset using a conflict graph and investigate private graph statistics to estimate these measures. The proposed machinery includes adapting graph-projection techniques with parameter selection optimizations on the conflict graph and a DP variant of approximate vertex cover size. We experimentally show that we can effectively compute DP estimates of the three measures on five real-world datasets with denial constraints, where the density of the conflict graphs highly varies.
\end{abstract}


\eat{

\ccsdesc[500]{Do Not Use This Code~Generate the Correct Terms for Your Paper}
\ccsdesc[300]{Do Not Use This Code~Generate the Correct Terms for Your Paper}
\ccsdesc{Do Not Use This Code~Generate the Correct Terms for Your Paper}
\ccsdesc[100]{Do Not Use This Code~Generate the Correct Terms for Your Paper}
}

\keywords{Differential privacy, Inconsistency measures, Integrity constraints}

\maketitle

\section{Introduction}\label{sec:intro}

Differential Privacy (DP)~\cite{dwork2006calibrating} has become the de facto standard for querying sensitive databases and has been adopted by various industry and government bodies~\cite{abowd2018us,erlingsson2014rappor,ding2017collecting}. 
DP offers high utility for aggregate data releases while ensuring strong guarantees on individuals' sensitive data. 
The laudable progress in DP study, as demonstrated by multiple recent works~\cite{zhang2015private,zhang2017privbayes,kotsogiannis2019privatesql,abs-1802-06739,gupta2010differentially,ZhangXX16}, has made it approachable and useful in many common scenarios. 
A standard DP mechanism adds noise to the query output, constrained by a privacy budget that quantifies the permitted privacy leakage. Once the privacy budget is exhausted, no more queries can be answered directly using the database. 
However, while DP ensures data privacy, it limits users' ability to directly observe or assess data quality, leaving them to rely on the data without direct validation.

\common{
The utility of such sensitive data primarily depends on its quality. Therefore, organizations that build these applications spend vast amounts of money on purchasing data from private data marketplaces~\cite{liu2021dealer, DBLP:conf/infocom/SunCLH22, DBLP:journals/corr/abs-2210-08723,DBLP:journals/tdsc/XiaoLZ23}. These marketplaces build relationships and manage monetary transactions between data owners and buyers. These buyers are often organizations that want to develop applications such as machine learning models or personalized assistants. Before the buyer purchases a dataset at a specific cost, they may want to ensure the data is suitable for their use case, adhere to particular data quality constraints, and be able to profile its quality to know if the cost reflects the quality.

{\em To address such scenarios, we consider the problem of assessing the quality of databases protected by DP.} 
Such quality assessment will allow users to decide whether they can rely on the conclusions drawn from the data or whether the suggested data is suitable for them. 
To solve this problem, we must tackle several challenges. First, since DP protects the database, users can only observe noisy aggregate statistics, which can be challenging to summarize into a quality score. Second, if the number of constraints is large (e.g., if they were generated with an automatic system~\cite{BleifussKN17,DBLP:journals/pvldb/LivshitsHIK20,PenaAN21}), translating each constraint to an SQL {\tt COUNT} query and evaluating it over the database with a DP mechanism may lead to low utility since the number of queries is large, allowing for only a tiny portion of the privacy budget to be allocated to each query. }

Hence, our proposed solution employs \emph{inconsistency measures}~\cite{thimm2017compliance, parisi2019inconsistency, LivshitsKTIKR21,DBLP:conf/sum/Bertossi18,DBLP:conf/ecsqaru/GrantH13,DBLP:journals/ijar/GrantH23,LivshitsK22,LivshitsBKS20} 
that quantify data quality with a single number for all constraints, essentially yielding {\em a data quality score.} This approach aligns well with DP, as such measures give a single aggregated numerical value representing data quality, regardless of the given number of constraints. 
As inconsistency measures, we adopt the ones studied by \citet{LivshitsKTIKR21} following earlier work on the topic~\cite{thimm2017compliance, parisi2019inconsistency, 
DBLP:conf/sum/Bertossi18,DBLP:conf/ecsqaru/GrantH13}. 
This work discusses and studies five measures, including (1) the {\em drastic measure}, a binary indicator for whether the database contains constraint violations, the (2) \emph{maximal consistency measure}, counting the number of maximal tuple sets for which addition of a single tuple will cause a violation, the (3) {\em minimal inconsistency measure}, counting the number of minimal tuple sets that violate a constraint, 
the (4) \emph{problematic measure}, counting the number of constraint violations, the (5) \emph{minimal repair measure}, counting the minimal tuple deletions needed to achieve consistency. 
These measures apply to various inconsistency measures that have been studied in the literature of data quality management, including functional dependencies, the more general conditional functional dependencies~\cite{bohannon2007conditional}, and the more general denial constraints~\cite{ChomickiM05}.
We show that the first two measures are incompatible for computation in the DP setting (\Cref{sec:hardness}), focusing throughout the paper on the latter three. 

An approach that one may suggest to computing the inconsistency measures in a DP manner is to translate the measure into an SQL query and then compute the query using an SQL engine that respects DP~\cite{tao2020computing,dong2022r2t,kotsogiannis2019privatesql,johnson2018towards}. Specifically relevant is R2T~\cite{dong2022r2t}, the state-of-the-art DP mechanism for SPJA queries, including self-joins. Nevertheless, when considering the three measures of inconsistency we focus on, this approach has several drawbacks. One of these measures (number of problematic tuples) requires the SQL {\tt DISTINCT} operator that R2T cannot handle. In contrast, another measure (minimal repair) cannot be expressed at all in SQL, making such engines irrelevant. 

Contrasting the first approach, the approach we propose and investigate here models the violations of the integrity constraints as a \emph{conflict graph} and applies DP techniques for graph statistics. In the conflict graph, nodes are tuple identifiers, and there is an edge between a pair of tuples if this pair violates a constraint. 
Then, each inconsistency measure can be mapped to a specific graph statistic. 
Using this view of the problem allows us to leverage prior work on releasing graph statistics with DP~\cite{hay2009accurate,KasiviswanathanNRS13,day2016publishing} and develop tailored mechanisms for computing inconsistency measures with DP. 

To this end, we harness graph projection techniques from the state-of-the-art DP algorithms~\cite{day2016publishing} that truncate the graph to achieve DP. While these algorithms have proven effective in prior studies on social network graphs, they may encounter challenges with conflict graphs arising from their unique properties. To overcome this, we devise a novel optimization for choosing the truncation threshold. 
We further provide a DP mechanism for the minimal repair measure that augments the classic 2-approximation of the vertex cover algorithm~\cite{vazirani1997approximation} to restrict its sensitivity and allow effective DP guarantees with high utility. 
Our experimental study shows that our novel algorithms prove efficacious for different datasets with various conflict graph sizes and sparsity levels.

\begin{table}
\centering
\begin{tabular}{|l|l|l|l|}
\hline
\textbf{DP Algorithms} & \textbf{Adult~\cite{misc_adult_2}}                  & \textbf{Flight~\cite{flight}}                 & \textbf{Stock~\cite{oleh_onyshchak_2020}}                  \\ \hline
\textbf{R2T~\cite{dong2022r2t}}           & $0.17\pm 0.01$ & $0.12\pm0.03$ & $123.19\pm 276.73$    \\ \hline
\textbf{This work}  & $0.10 \pm 0.05$ & $0.10 \pm 0.20$  & $0.07\pm 0.08$ \\ \hline
\end{tabular}
  \caption{Relative errors for a SQL approach vs our approach to compute the minimal inconsistency measure at $\epsilon=1$} \label{tab:intro_comparison}
\end{table}

Beyond handling the two inconsistency measures that R2T cannot handle, our approach provides considerable advantages even for the one R2T can handle (number of conflicts). 
For illustration, \Cref{tab:intro_comparison} shows the results of evaluating R2T~\cite{dong2022r2t} on three datasets with the same privacy budget of $1$ for this measure. Though R2T performed well for the Adult and Flight datasets, it reports more than 120\% relative errors for the Stock dataset with very few violations. 
On the other hand, our approach demonstrates strong performance across all three datasets.

The main contributions of this paper are as follows.
First, we formulate the novel problem of computing inconsistency measurements with DP for private datasets and discuss the associated challenges, including a thorough analysis of the sensitivity of each measure.
Second, we devise several algorithms that leverage the conflict graph and algorithms for releasing graph statistics under DP to estimate the measures that we have determined are suitable. Specifically, we propose a new optimization for choosing graph truncation threshold that is tailored to conflict graphs and augment the classic vertex cover approximation algorithms to bound its sensitivity to $2$ to obtain accurate estimates of the measures. 
Third, we present experiments on five real-world datasets with varying sizes and densities to show that the proposed DP algorithms are efficient in practice. Our average error across these datasets is 1.3\%-67.9\% compared to the non-private measure.

\section{Preliminaries}\label{sec:prelim}
We begin with some background that we need to describe the concept of inconsistency measures for private databases. 

\def\tids{\mathit{tids}}

\subsection{Database and Constraints}\label{sec:prelim-integrity-constraints}
We consider a single-relation schema $\attrset = (A_1, \ldots, A_m)$, which is a vector of distinct attribute names $A_i$, each associated with a domain $\dom(A_i)$ of values. A database $D$ over $\attrset$ is a associated with a set $\tids(D)$ of \emph{tuple identifiers}, and it maps every identifier $i\in\tids(D)$ to a tuple $D[i]=(a_1,\dots,a_m)$ in $A_1\times\dots\times A_m$. A database $D'$ is a \emph{subset} of $D$, denoted $D'\subseteq D$, if  $D'$ is obtained from $D$ by deleting zero or more tuples, that is, $\tids(D')\subseteq\tids(D)$ and $D'[i]=D[i]$ for all $i\in\tids(D')$.

Following previous work on related topics~\cite{DBLP:journals/pvldb/GeMHI21,DBLP:journals/pvldb/LivshitsHIK20}, we focus on Denial Constraints (DCs) on pairs of tuples. Using the formalism of Tuple Relational Calculus, such a DC is of the form
$\forall t,t' \neg \big(\varphi_1\land\dots\land\varphi_k\big)$ where each $\varphi_j$ is a comparison $\sigma_1\circ\sigma_2$ so that: (a) each of $\sigma_1$ and $\sigma_2$ is either $t[A_i]$, or $t'[A_i]$, or $a$, where $A_i$ is some attribute and $a$ is a constant value, and (b) the operator $\circ$ belongs to set $\{<, >, \leq, \geq, =, \neq\}$ of comparisons. This DC states that there cannot be two tuples $t$ and $t'$ such that all comparisons $\varphi_j$ hold true (i.e., at least one $\varphi_j$ should be violated).

Note that the class of DCs of the form that we consider generalizes the class of Functional Dependencies (FDs). An FD has the form $X\rightarrow Y$ where $X,Y\subseteq\{A_1,\dots,A_m\}$, and it states that every two tuples that agree on (i.e., have the same value in each attribute of) $X$ must also agree on $Y$.

In the remainder of the paper, we denote by $\constraintset$ the given set of DCs. A database $D$ \emph{satisfies} $\constraintset$, denoted $D \models \constraintset$, if $D$
satisfies every DC in $\constraintset$; otherwise, $D$ \emph{violates} $\constraintset$, denoted $D \not\models \constraintset$.

\def\set#1{\mathord{\{#1\}}}
A common way of capturing the violations of $\constraintset$ in $D$ is through the \emph{conflict graph} $\graph$, which is the graph $(V, E)$, where $V=\tids(D)$ an edge $e = \set{i,j} \in E$ occurs whenever the tuples $D[i]$ and jointly $D[j]$ violate $\constraintset$. 
To simplify the notation, we may write simply $\mathcal{G}$ instead of $\graph$ when there is no risk of ambiguity.

\begin{example}
     Consider a dataset that stores information about capital and country as shown in Figure~\ref{fig:db_to_graph}. Assume an FD constraint $\sigma: \mbox{Capital} \rightarrow \mbox{Country}$ between attributes capital and country that says that the country of two tuples must be the same if their capital is the same. Assume the dataset has 3 rows (white color) and a neighboring dataset has an extra row (grey color) with the typo in its country attribute. As shown in the right side of Figure~\ref{fig:db_to_graph}, the dataset with 4 rows can be converted to a conflict graph with the nodes corresponding to each tuple and edges referring to conflicts between them. The $\mininconsistency$ measure computes the size of the set of all minimally inconsistent subsets $|MI_\constraintset(D)|$ (the number of edges in the graph) for this dataset. 
     \label{example:running_example}
\end{example}

\begin{figure}
    \centering
    \begin{tikzpicture}
    \node[draw, align=center] (table) at (1.5,0) {
        \begin{tabular}{|c|c|c|}
            \hline
            \textbf{ID} & \textbf{Capital} & \textbf{Country} \\
            \hline
            1 & Ottawa & Canada \\
            \hline
            2 & Ottawa & Canada \\
            \hline
            3 & Ottawa & Canada \\
            \hline
             \rowcolor{lightgray}
            4 & Ottawa & Kanada \\
            \hline
        \end{tabular}
    };

    
    \node[draw, circle] (node1) at (4.5,0.5) {1};
    \node[draw, circle] (node2) at (5.5,0.5) {2};
    \node[draw, circle] (node3) at (5,-0.7) {3};
    \node[draw, circle] (node4) at (5,0) {4};
    \draw (node1) -- (node4);
    \draw (node2) -- (node4);
    \draw (node3) -- (node4);
    
    \draw[->] (3.7, 0) -- (4.2, 0);
\end{tikzpicture}
   \vspace{0.2cm} 
    \caption{Toy example dataset to show a worst-case analysis. An additional row may violate all other rows in the dataset (left). Easier analysis can be done by instead converting the dataset into its corresponding conflict graph (right).
    }
    \label{fig:db_to_graph}
\end{figure}
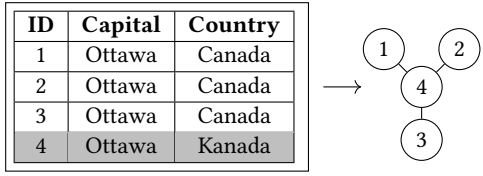

\subsection{Inconsistency Measures}\label{sec:prelim-inconsistency}

\def\MI{\mathord{\mathsf{MI}}}
\def\MC{\mathord{\mathsf{MC}}}

Inconsistency measures have been studied in previous work~\cite{DBLP:conf/sum/Bertossi18,DBLP:conf/ecsqaru/GrantH13,DBLP:journals/ijar/GrantH23,LivshitsK22,LivshitsBKS20} as a means of measuring database quality for a set of DCs. 
We adopt the measures and notation of \citet{LivshitsKTIKR21}. Specifically, they consider five inconsistency measures that capture different aspects of the dataset quality. To define these concepts, we need some notation. Given a database $D$ and a set $\constraintset$ of anti-monotonic integrity constraints, we denote by $\MI_\constraintset(D)$ the set of all \emph{minimally inconsistent subsets}, that is, the sets $E \subseteq D$ such that $E \not\models \constraintset$ but $E' \models \constraintset$ for all $E' \subsetneq E$. We also denote by $\MC_\constraintset(D)$ the set of all \emph{maximal consistent subsets} of $D$; that is, the sets $E \subseteq D$ such that $E \models \constraintset$ and $E' \not\models \constraintset$ whenever $E \subsetneq E' \subseteq D$.

\begin{definition}[Inconsistency measures~\cite{LivshitsKTIKR21}]\label{def:inconsistencymeasure}
Given a database $D$ and a set of DCs $\constraintset$, the inconsistency measures are defined as follows:
\begin{itemize}
    \item Drastic measure: $\drastic(D, \constraintset) = 1$ if $D \models \constraintset$ and 0 otherwise.
    \item Minimal inconsistency measure:  $\mininconsistency(D,\constraintset) = |\MI_\constraintset(D)|$.
    \item Problematic measure: $\problematic(D,\constraintset) = |\cup \MI_\constraintset(D)|$.
    \item Maximal consistency measure: $\maxconsistency(D,\constraintset) = |\MC_\constraintset(D)|$ \footnote{We drop "-1" from the original definition~\cite{LivshitsKTIKR21} for simplicity.}.
    \item Optimal repair measure: $\repair(D,\constraintset) = |D| - |D_R|$, where $|D_R|$ is the largest subset $D_R \subseteq D$ such that $D_R \models \constraintset$.
\end{itemize}
\end{definition}

Observe that inconsistency measures also have a graphical interpretation for the conflict graph \graph. For instance, the drastic measure $\drastic(D,\constraintset)$ corresponds to a binary indicator for whether there exists an edge in $\graph$. We summarize the graph interpretation of these inconsistency measures in Table~\ref{tab:summary}.

\eat{
We can also define these inconsistency measures in terms of a conflict graph \graph as follows~\cite{LivshitsKTIKR21}:

\xh{may drop this observation chunk, and points to Table~\ref{tab:summary}}.
\begin{observation}\label{prop:graph-defs}
Given a database $D$, a set of DCs \constraintset, and their corresponding conflict graph \graph, the following holds: 
\begin{enumerate}
    \item $\drastic(D,\constraintset)$ is 1 if there exists at least an edge in \graph, 0 otherwise. 
    \item\label{itm:mininc} $\mininconsistency(D,\constraintset)$ is the number of edges in \graph.
    \item $\problematic(D,\constraintset)$ is the number of nodes in \graph\ that have a positive degree.
    \item $\maxconsistency(D,\constraintset)$ is the number of maximal independent sets in \graph.
    \item $\repair(D,\constraintset)$ is the minimum vertex cover size of \graph. 
\end{enumerate}
\end{observation}
}

\subsection{Differential Privacy}~\label{sec:prelim-dp}
Differential privacy (DP) \cite{dwork2006calibrating} aims to protect private information in the data. In this work, we consider the unbounded DP setting where we define two neighboring datasets, $D$ and $D'$ (denoted by $D \neighbor D'$) if $D'$ can be transformed from $D$ by adding or removing one tuple in $D$.

\def\prob{\mathrm{Pr}}

\begin{definition}[Differential Privacy \cite{dwork2006calibrating}]
    An algorithm $\mathcal{M}$ is said to satisfy $\epsilon$-DP if for all $S\subseteq \mbox{Range}(\mathcal{M})$ and for all $D\neighbor D'$,
        $$\prob[\mathcal{M}(D)\in S] \leq e^\epsilon \prob[\mathcal{M}(D')\in S]\,.$$
\end{definition}

The privacy cost is measured by the parameters $\epsilon$, often called the \emph{privacy budget}. The smaller $\epsilon$ is, the stronger the privacy is.  
Complex DP algorithms can be built from the basic algorithms following two essential properties of differential privacy:

\begin{proposition}[DP Properties~\cite{Dwork06,DworkKMMN06}]\label{prop:DP-comp-post}  
    The following hold.
    \begin{enumerate}
        \item {\bf (Sequential composition)} 
        If $\mathcal{M}_i$ satisfies $\epsilon_i$-DP, then the sequential application of $\mathcal{M}_1$, $\mathcal{M}_2, \cdots$,  satisfies $(\sum_{i} \epsilon_i)$-DP.
        
        \item {\bf (Parallel composition)} 
        If each $\mathcal{M}_i$ accesses disjoint sets of tuples, they satisfy $(\max_i \epsilon_i)$-DP together.
        
        \item {\bf (Post-processing)} 
        Any function applied to the output of an $\epsilon$-DP mechanism $\mathcal{M}$ also satisfies $\epsilon$-DP.
    \end{enumerate}
\end{proposition}

Many applications in DP require measuring the change in a particular function's result over two neighboring databases. The supremum over all pairs of neighboring databases is called the \emph{sensitivity} of the function. 

\begin{definition}[Global sensitivity \cite{DworkMNS16}]\label{def:sensitivity}
Given a function $f: \mathcal{D} \rightarrow \mathbb{R}$,  the  sensitivity of $f$ is 
\begin{equation}
\Delta_f = \max\limits_{D'\approx D}|f(D) - f(D')|.
\end{equation}
\end{definition}

\paratitle{Laplace mechanism}
The Laplace mechanism~\cite{DworkMNS16} is a common building block in DP mechanisms and is used to get a noisy estimate for queries with numeric answers. The noise injected is calibrated to the query's global sensitivity.

\begin{definition}[Laplace Mechanism~\cite{DworkMNS16}]\label{def:LM}
Given a database $D$, a function $f$ : $\mathcal{D} \rightarrow \mathbb{R}$, and a privacy budget $\epsilon$, the Laplace mechanism $\mathcal{M}_L$ returns $f(D) + \nu_q$, where $\nu_q\sim Lap(\Delta_f/\epsilon)$. 
\end{definition}

The Laplace mechanism can answer many numerical queries, but the exponential mechanism can be used in many natural situations requiring a non-numerical output. 

\paratitle{Exponential mechanism} The exponential mechanism~\cite{mcsherry2007mechanism} expands the application of DP by allowing a non-numerical output. 

\begin{definition}[Exponential Mechanism~\cite{mcsherry2007mechanism}]\label{def:EM}
Given a dataset $D$, a privacy budget $\epsilon$, a set $\Theta$ of output candidates, a quality function $q(D, \theta_i) \in \mathbb{R}$, the exponential mechanism $\mathcal{M}_{EM}$ outputs a candidate $\theta_i \in \Theta$ with probability proportional to $\exp \left(\frac{\epsilon q(D, \theta_i)}{2\Delta_q}\right)$, where $\Delta_q$ is the sensitivity of the quality function $q$.
\end{definition}

\eat{
The exponential mechanism offers strong utility guarantees proportional to the number of candidates $|\Theta|$ and the sensitivity of the quality function $\Delta q$ as formalized by Theorem~\ref{thm:utility_expo}.

\xh{we can remove this theorem}
\begin{theorem}[\cite{mcsherry2007mechanism}]\label{thm:utility_expo}
Let $D$ be a private dataset, and $OPT(D)=\max _{\theta \in \Theta} q(D, \theta)$ be the score attained by the best object $\theta$ with respect to the dataset $D$ due to the exponential mechanism $M_{EM}(D)$. If the set of objects that achieve the $OPT(D)$, $\Theta^*=\{\theta \in \Theta: q(D, \theta)=OPT(D)\}$ has size $|\Theta^*| \geq 1$. Then
$$ \Pr[q(D, M_E(D)) \leq OPT (D) - \frac{2\Delta_q}{\epsilon_{em}} (\ln |\Theta| + t) ] \leq \exp(-t)$$
where $\epsilon_{em}$ is the privacy budget for the exponential mechanism and $\Delta_q$ is the sensitivity of the quality function $q$. 
\end{theorem}
}

\begin{table*}[h]
\centering
\footnotesize
\begin{tabular}{|l|ll|lll|}
\hline
                                      & \multicolumn{2}{c|}{\textbf{Non-private analysis \cite{LivshitsKTIKR21}}}                           & \multicolumn{3}{c|}{\textbf{DP analysis (this work)}}                                                                         \\ \hline
\textbf{Inconsistency Measures for $D$} & \multicolumn{1}{l|}{\textbf{Graph Interpretations in $\graph(V,E)$}}     & \textbf{Computation cost} & \multicolumn{1}{l|}{\textbf{Sensitivity}}      & \multicolumn{1}{l|}{\textbf{Computation cost}} & \textbf{Utility} \\ \hline
Drastic measure     $\drastic$                       & \multicolumn{1}{l|}{if exists an edge}                        & $O(|\Sigma|n^2)$    & \multicolumn{1}{l|}{1}                         & \multicolumn{1}{l|}{N.A.}                                    &    N.A.            \\ \hline
Minimal inconsistency measure    $\mininconsistency$       & \multicolumn{1}{l|}{the no. of edges}                      & $O(|\Sigma|n^2)$  & \multicolumn{1}{l|}{$n$}                    & \multicolumn{1}{l|}{$O(|\Sigma|n^2 + |\Theta|m)$}                                 &       $-\tilde{q}_{\opt}(D,\epsilon_2)+O(\frac{\theta_{\max}\ln|\Theta|}{\epsilon_1})$           \\ \hline
Problematic measure         $\problematic$               & \multicolumn{1}{l|}{the no. of nodes with positive degrees} & $O(|\Sigma|n^2)$  & \multicolumn{1}{l|}{$n$}                    & \multicolumn{1}{l|}{$O(|\Sigma|n^2+|\Theta|m)$}                                 &         $-\tilde{q}_{\opt}(D,\epsilon_2)+O(\frac{\theta_{\max}\ln|\Theta|}{\epsilon_1})$               \\ \hline
Maximal consistency measure    $\maxconsistency$           & \multicolumn{1}{l|}{the no. of maximal independent sets}   & \#P-complete              & \multicolumn{1}{l|}{$O(3^n)$} & \multicolumn{1}{l|}{N.A.}                                    &    N.A.              \\ \hline
Optimal repair measure        $\repair$               & \multicolumn{1}{l|}{the minimum vertex cover size}            & NP-hard                   & \multicolumn{1}{l|}{1}                    & \multicolumn{1}{l|}{$O(|\Sigma|n^2+ m)$}                                 &             $\mathcal{I}_R(D,\Sigma) + O(1/\epsilon) $     \\ \hline
\end{tabular}
\caption{Summary of Inconsistency Measures, $n=|D|=|\graph.V|$, $m=|\graph.E|$, $\Theta$ is the candidate set.}\label{tab:summary} 
\end{table*}

\paratitle{DP for graphs}
When the dataset is a graph $\mathcal{G} = (V, E)$, the standard definition can be translated to two variants of DP~\cite{hay2009accurate}. The first is \emph{edge-DP} where two graphs are neighboring if they differ on one edge, and the second is \emph{node-DP}, when two graphs are neighboring if one is obtained from the other by removing a node (and its incident edges). The two definitions offer different kinds of privacy protection.  In our work, as we deal with databases and their corresponding conflict graphs, adding or removing a tuple of the dataset translates to node-DP. The corresponding definition of neighboring datasets changes to neighboring graphs where two graphs $\mathcal{G}$ and $\mathcal{G}'$ are called neighboring $\mathcal{G} \neighbor \mathcal{G}^\prime$ if $\mathcal{G}'$ can be transformed from $\mathcal{G}$ by adding or removing one node along with all its edges in $\mathcal{G}$. Node-DP provides a stronger privacy guarantee than edge-DP since it protects an individual's privacy and all its connections, whereas edge-DP concerns only one such connection.  

\begin{definition}[Node sensitivity]\label{def:node_sensitivity}
Given a function $f$ over a graph $\mathcal{G}$, the sensitivity of $f$ is 
$\Delta_f = \max\limits_{\mathcal{G}'\approx \mathcal{G}}|f(\mathcal{G}) - f(\mathcal{G}')|$.
\end{definition}

The building blocks of DP, such as the Laplace and Exponential mechanisms, also work on graphs by simply substituting the input to a graph and the sensitivity to the corresponding node sensitivity. 

\paratitle{Graph projection} Graph projection algorithms refer to a family of algorithms that help reduce the node sensitivity of a graph by 
truncating the edges and, hence, bounding the maximum degree of the graph. Several graph projection algorithms exist~\cite{KasiviswanathanNRS13, blocki2013differentially}, among which the \textit{edge addition} algorithm~\cite{day2016publishing} stands out for its effectiveness in preserving most of the underlying graph structure. The edge addition algorithm denoted by $\pi_\theta^\Lambda$, takes as input the graph $\mathcal{G}= \graph = (V, E)$, a bound on the maximum degree of each vertex ($\theta$), and a stable ordering of the edges ($\Lambda$) to output a projected $\theta$-bounded graph denoted by $\mathcal{G}_\theta = \pi_\theta^{\Lambda}(\mathcal{G})$. 
\begin{definition}[Stable ordering~\cite{day2016publishing}]~\label{def:stable_ordering} A graph edge ordering $\Lambda$ is stable if and only if given two neighboring graphs $\mathcal{G} = (V, E)$ and $\mathcal{G}' = (V', E')$ that differ by only a node, $\Lambda(\mathcal{G})$
and $\Lambda(\mathcal{G}')$ are consistent in the sense that if two edges appear
both in $\mathcal{G}$ and $\mathcal{G}'$, their relative ordering are the same in  $\Lambda(\mathcal{G})$ and
$\Lambda(\mathcal{G}')$.
\end{definition}
The stable ordering of edges, $\Lambda(\mathcal{G})$, can be any deterministic ordering of all the edges $E$ in the $\mathcal{G}$. Such stabling edge ordering can be easily obtained in practice. For example, it could be an ordering (e.g. alphabetical ordering) based on the node IDs of the graph such that in the neighboring dataset $\mathcal{G}'$, the edges occur in the same ordering as $\mathcal{G}$. The edge addition algorithm starts with an empty set of edges and operates by adding edges in the same order as $\Lambda$ so that each node has a maximum degree of $\theta$. 
To simplify the notation, in the remainder of the paper, we drop $\Lambda$ and denote the edge addition algorithm $\pi_\theta^{\Lambda}(\mathcal{G})$ as $\pi_\theta(\mathcal{G})$.

\ifpaper
\else
\begin{algorithm}
\caption{Edge addition algorithm~\cite{day2016publishing}}\label{algo:day_edgeadd}
    \KwData{Graph $\mathcal{G} (V,E)$, Bound $\theta$, Stable ordering $\Lambda$}
    \KwResult{$\theta$-bounded graph $\pi_\theta(\mathcal{G})$}
    $E^\theta \leftarrow \emptyset ; d(v) \leftarrow 0$ for each $v \in V$ \;
    \For{$e=(u,v) \in \Lambda$} {
        \If{$d(u)<\theta \& d(v)<\theta$}{
            $E^\theta \leftarrow E^\theta \cup\{e\} \; d(u) \leftarrow d(u)+1 , d(v) \leftarrow d(v)+1$ \;
        }
    }
    return $G^\theta=(V, E^\theta)$\;
\end{algorithm}
\fi
\section{Inconsistency Measures under DP}\label{sec:problem}

\paratitle{Problem Setup} Consider a private dataset $D$, a set of DCs $\constraintset$, and a privacy budget $\epsilon$. For an inconsistency measure $\mathcal{I}$ from the set $\{\drastic,\mininconsistency,\problematic,\maxconsistency,\repair\}$ (\Cref{def:inconsistencymeasure}), we would like to design an $\epsilon$-DP algorithm $\mathcal{M}(D,\Sigma,\epsilon)$ such that with high probability, $|\mathcal{M}(D, \constraintset, \epsilon) - \mathcal{I}(D, \constraintset)|$ is bounded with a small error.

\paratitle{Sensitivity Analysis} 
We first analyze the sensitivity of the five inconsistency measures and discuss the challenges to achieving DP.  
\begin{proposition}\label{prop:sensitivity}
Given a database $D$ and a set of DCs $\constraintset$, where $|D|=n$, the following holds:
\begin{enumerate*}
    \item The global sensitivity of \drastic\ is 1. 
    \item The global sensitivity of 
    \mininconsistency\ is $n$.       
    \item The global sensitivity of \problematic\ is $n$.
    \item The global sensitivity of \maxconsistency\ is exponential in $n$.
    \item The global sensitivity of \repair\ is 1.
\end{enumerate*}
\end{proposition}
\ifpaper
\else
The proof can be found in Appendix~\ref{app:proof_sensitivity}.
\fi

\paratitle{Inadequacy of $\drastic$ and $\maxconsistency$}\label{sec:hardness}
We note that two inconsistency measures are less suitable for DP. First, the drastic measure \drastic\ is a binary measure that outputs $1$ if at least one conflict exists in the dataset and $0$ otherwise. Due to its binary nature, the measure's sensitivity is 1, meaning adding or removing a single row can significantly alter the result. Adding DP noise to such a binary measure can render it meaningless.

One way to compute the \drastic\ measure could be to consider a proxy of \drastic\ by employing a threshold-based approach that relies on \problematic\ or \mininconsistency. For example, if these measures are below a certain given number, we return $0$ and, otherwise, $1$. A recent work~\cite{PatwaSGMR23} addresses similar problems for synthetic data by employing the exponential mechanism. However, since we focus on directly computing the measures under DP, we leave this intriguing subject for future work. 

\revc{The \maxconsistency\ measure that computes the total number of independent sets in the conflict graph has both computational and high sensitivity issues. First, prior work~\cite{LivshitsBKS20} showed that computing \maxconsistency\ is \#P-complete and even approximating it is an NP-hard problem~\cite{DBLP:conf/ijcai/Roth93}. Even for special cases where \maxconsistency\ can be polynomially computed (when $\graph$ is $P_4$-free~\cite{KimelfeldLP20}), we show in Proposition~\ref{prop:sensitivity} that its sensitivity is exponential in the number of nodes of $\graph$. This significantly diminishes the utility of its DP estimate. Due to these challenges, we defer the study of $\drastic$ and $\maxconsistency$ to future work.}

\paratitle{Challenges for \repair, \mininconsistency, and \problematic}
Although the $\repair$ measure has a low sensitivity of 1 for its output range $[0,n]$, it is an NP-hard problem, and the common non-private solution
is to solve a linear approximation that requires solving a linear program~\cite{LivshitsBKS20}. However, in the worst case, this linear program again has sensitivity equal to $n$ (number of rows in the dataset) and may have up to $\binom n2$ number of constraints (all rows violating each other). Existing state-of-the-art DP linear solvers~\cite{hsu2014privately} are slow and fail for such a challenging task. Our preliminary experiments to solve such a linear program timed out after 24 hours with $n=1000$.
For \mininconsistency\ and \problematic\, they have polynomial computation costs and reasonable output ranges. However, they still have high sensitivity $n$. In the upcoming sections~\ref{sec:graph-algorithms-graphproj} and \ref{sec:vertex_cover}, we show that these problems can be alleviated by pre-processing the input dataset as a conflict graph and computing these inconsistency measures as private graph statistics.

\section{DP Graph Projection for $\mininconsistency$ and $\problematic$}\label{sec:graph-algorithms-graphproj}

Computing graph statistics such as edge count and degree distribution while preserving node-differential privacy (node-DP) is a well-explored area~\cite{day2016publishing, KasiviswanathanNRS13, blocki2013differentially}. 
Hence, in this section, we leverage the state-of-the-art node-DP approach for graph statistics to analyze the inconsistency measures $\mininconsistency$ and $\problematic$ as graph statistics on the conflict graph $\graph$. However, the effectiveness of this approach hinges on carefully chosen parameters.  We introduce two optimization techniques that consider the integrity constraints to optimize parameter selection and enhance the algorithm's utility.

\subsection{Graph Projection Approach for \mininconsistency\ and \problematic} 
A primary utility challenge in achieving node-DP for graph statistics is their high sensitivity. In the worst case, removing a single node from a graph of $n$ nodes can result in removing $(n-1)$ edges. To mitigate this issue, the state-of-the-art approach~\cite{day2016publishing} first projects the graph $\mathcal{G}$ onto a $\theta$-bounded graph $\mathcal{G}_{\theta}$, where the maximum degree is no more than $\theta$. Subsequently, the edge count of the transformed graph is perturbed by the Laplace mechanism with a sensitivity value of less than $n$. However, the choice of $\theta$ is critical for accurate estimation. 
A small $\theta$ reduces Laplace noise due to lower sensitivity, but results in significant edge loss during projection. Conversely, a $\theta$ close to $n$ preserves more edges but increases the Laplace noise. Prior work addresses this balance using the exponential mechanism (EM) to prefer a $\theta$ that minimizes the combined errors arising from graph projection and the Laplace noise.

\begin{algorithm}[t]
\caption{Graph projection approach for $\mininconsistency$ and $\problematic$}
\label{algo:graph_general}
    \KwData{Dataset $D$, constraint set $\constraintset$, candidate set $\Theta$, privacy budgets $\epsilon_1$ and $\epsilon_2$}
    \KwResult{DP inconsistency measure for $\mininconsistency$ or $\problematic$}
    
    Construct the conflict graph $\graph$\\
    
    Sample $\theta^*$ from $\Theta$  with a $\epsilon_1$-DP mechanism \commenttext{// Basic EM (Algorithm~\ref{algo:expo_mech_basic}); Optimized EM (Algorithm~\ref{algo:em_opt})}\
    
    Compute $\theta^*$-bounded graph $\mathcal{G}_{\theta^*} \gets \pi_{\theta^*}(\graph)$  
     \commenttext{// Edge addition algorithm~\cite{day2016publishing}}\\
    
    {\bf Return} $f(\mathcal{G}_{\theta^*})+\text{Lap}(\frac{\theta^*}{\epsilon_2})$
    \commenttext{//
    $f(\cdot)$ returns edge count for $\mininconsistency$ and the number of nodes with positive degrees for $\problematic$} \end{algorithm}

We outline this general approach in Algorithm~\ref{algo:graph_general}. This algorithm takes in the dataset $D$, the constraint set $\constraintset$, a candidate set $\Theta$ for degree bounds, and privacy budgets $\epsilon_1$ and $\epsilon_2$. These privacy budgets are later composed to get a final guarantee of $\epsilon$-DP.
We start by constructing the conflict graph $\graph$ generated from the input dataset $D$ and constraint set $\constraintset$ (line 1), as defined in \Cref{sec:prelim-integrity-constraints}. 
Next, we sample in a DP manner a value of $\theta^*$ from the candidate set $\Theta$ with the privacy budget $\epsilon_1$ (line 2). A baseline choice is an exponential mechanism detailed in Algorithm~\ref{algo:expo_mech_basic} to output a degree that minimizes the edge loss in a graph and the Laplace noise.
In line 3, we compute a bounded graph $\mathcal{G}_{\theta^*}$ using the edge addition algorithm~\cite{day2016publishing}, we compute a $\theta^*$-bounded graph $\mathcal{G}_{\theta^*}$ (detailed in Section~\ref{sec:prelim}). Finally, we perturb the true measure (either the number of edges for $\mininconsistency$ or the number of positive degree nodes for $\problematic$) on the projected graph, denoted by $f(\mathcal{G}_{\theta^*})$, by adding Laplace noise using the other privacy budget $\epsilon_2$ (line 4). 

The returned noisy measure at the last step has two sources of errors: (i) the bias incurred in the projected graph, i.e.,  $f(\mathcal{G})-f(\mathcal{G}_{\theta^*})$, and (ii) the noise from the Laplace mechanism with an expected square root error ${\sqrt{2}\theta^*}/{\epsilon_2}$. Both errors depend on the selected parameter $\theta^*$, and it is vital to select an optimal $\theta^*$ that minimizes the combined errors. Next, we describe a DP mechanism that helps select this parameter.

\begin{algorithm}[b]
\caption{EM-based first try for parameter selection}
\label{algo:expo_mech_basic}
    \KwData{Graph $\mathcal{G}$, candidate set $\Theta$, quality function $q$, privacy budget $\epsilon_1, \epsilon_2$ }
    \KwResult{Candidate $\theta^*$}
    Find the maximum value in $\Theta$ as $\theta_{\max}$ \\
    For each $\theta_i \in \Theta$, compute $q_{\epsilon_2}(\mathcal{G}, \theta_i)$ 
    \commenttext{// See Equation~\eqref{eq:quality_function}}
    \\
    
    Sample $\theta^*$ with prob $\propto \exp( \frac{\epsilon_1 q_{\epsilon_2}(\mathcal{G}, \theta_i)}{2\theta_{\max}})$ \\
    {\bf Return} $\theta^*$
\end{algorithm}

\paratitle{EM-based first try for parameter selection} 
The EM (Definition~\ref{def:EM}) specifies a quality function $q(\cdot,\cdot)$ that maps a pair of a database $D$ and a candidate degree $\theta$ to a numerical value. The optimal $\theta$ value for a given database $D$ should have the largest possible quality value and, hence, the highest probability of being sampled. We also denote $\theta_{\max}$ the largest degree candidate in $\Theta$ and use it as part of the quality function to limit its sensitivity.

The quality function we choose to compute the inconsistency measures includes two terms: for each $\theta\in \Theta$,
\begin{equation}~\label{eq:quality_function}
    q_{\epsilon_2}(\mathcal{G}, \theta) = - e_{\text{bias}}(\mathcal{G}, \theta) - {\sqrt{2}\theta}/{\epsilon_2}
\end{equation}
where the first term $e_{\text{bias}}$  captures the bias in the projected graph, and the second term ${\sqrt{2}\theta}/{\epsilon_2}$ captures the error from the Laplace noise at budget $\epsilon_2$. For the minimum inconsistency measure $\mininconsistency$, we define the bias term as
\begin{equation}
e_{\text{bias}}(\mathcal{G},\theta) = |\mathcal{G}_{\theta_{\max}}.E| - |\mathcal{G}_\theta.E|    
\end{equation}
i.e., the number of edges truncated at degree $\theta$ as compared to that at degree $\theta_{\max}$. For the problematic measure $\problematic$, we have 
\begin{equation}
e_{\text{bias}}(\mathcal{G}, \theta) = 
|\mathcal{G}_{\theta_{\max}}.V_{>0}| - |\mathcal{G}_{\theta}.V_{>0}|    
\end{equation} 
where $\mathcal{G}_\theta.V_{>0}$ denote the nodes with positive degrees. 

\begin{example}~\label{example:quality_function}
    Consider the same graph as Example~\ref{example:running_example} and a candidate set $\Theta = [1, 2, 3]$ to compute the $\mininconsistency$ measure (number of edges) with $\epsilon_2=1$. For the first candidate $\theta = 1$, as node 4 has degree 3, the edge addition algorithm would truncate 2 edges, for $\theta = 2$, 1 edge would be truncated and for $\theta = 3$, no edges would be truncated. We can, therefore, compute each term of the quality function for each $\theta$ given in Table~\ref{tab:example_quality_function}.  
    \begin{table}[]
        \centering
        \begin{tabular}{|c|c|c|c|}
             \hline
             $\theta$ & $e_{\text{bias}}$ & ${\sqrt{2}\theta}/{\epsilon_2}$ & q  \\
             \hline
             1 & 2 & $\sqrt{2}$ & $-2 - \sqrt{2}$\\
             2 & 1 & $2\sqrt{2}$ &  $-1 - 2\sqrt{2}$\\
             3 & 0 & $3\sqrt{2}$ & $-3\sqrt{2}$\\
             \hline
        \end{tabular}
        \caption{Quality function computation for $\mininconsistency$ for the conflict graph in Figure~\ref{fig:db_to_graph} when $\epsilon_2=1$}
        \label{tab:example_quality_function}
    \end{table}
    For this example, we see that $\theta=1$ has the best quality even if it truncates the most number of edges as the error from Laplace noise overwhelms the bias error.
\end{example}

We summarize the basic EM for the selection of the bounded degree in Algorithm~\ref{algo:expo_mech_basic}.
This algorithm has a complexity of $O(|\Theta|m)$, where $m$ is the edge size of the graph, as computing the quality function for each $\theta$ candidate requires running the edge addition algorithm once. The overall Algorithm~\ref{algo:graph_general} has a complexity of $O(|\Sigma|n^2+|\Theta|m)$, where the first term is due to the construction of the graph.

\paratitle{Privacy analysis}
The privacy guarantee of \cref{algo:graph_general} depends on the budget spent for the exponential mechanism and the Laplace mechanism, as summarized below. 
\begin{theorem}\label{thm:privacy_proof_dc_oblivious}
    ~\cref{algo:graph_general} satisfies $(\epsilon_1 + \epsilon_2)$-node DP for $\graph$ and $(\epsilon_1 + \epsilon_2)$-DP for the input database $D$.
\end{theorem}

\reva{
\begin{proof}[Proof sketch]
The proof is based on the sequential composition of two DP mechanisms as stated in Proposition~\ref{prop:DP-comp-post}.
\end{proof}
}

 As stated below, we just need to analyze the sensitivity of the quality function in the exponential mechanism and the sensitivity of the measure over the projected graph.

\begin{lemma}\label{lemma:sensitivity}
    The sensitivity of $f\circ\pi_\theta(\cdot)$ in Algorithm~\ref{algo:graph_general} is $\theta$, where $\pi_\theta$ is the edge addition algorithm with the input $\theta$ and $f(\cdot)$ counts edges for $\mininconsistency$ and nodes with a positive degrees for $\problematic$.
\label{lemma:sens_lap}
\end{lemma}

\reva{
\begin{proof}[Proof sketch]
For $\problematic$, we can analyze a worst-case scenario where the graph is a star with $n$ nodes such that there is an internal node connected to all other $n-1$ nodes, and the threshold $\theta$ for edge addition is $n$. The edge addition algorithm would play a minimal role, and no edges would be truncated. For a neighboring graph that differs on the internal node, all edges of the graph are removed (connected to the internal node), and the $\problematic = 0$ (no problematic nodes), making the sensitivity for $\problematic$ in this worst-case $=n$.

For $\mininconsistency$, the proof is similar to prior work~\cite{day2016publishing} for publishing degree distribution that uses stable ordering to keep track of the edges for two neighboring graphs. We need to analyze the changes made to the degree of each node by adding one edge at a time for two graphs $\mathcal{G}$ and its neighboring graph $\mathcal{G}'$ with an additional node $v^+$. The graphs have the stable ordering of edges (\cref{def:stable_ordering}) $\Lambda$ and $\Lambda'$, respectively.  Assuming the edge addition algorithm adds a set of $t$ extra edges incident to $v^+$ for $\mathcal{G}'$, we can create $t$ intermediate graphs and their respective stable ordering of edges that can be obtained by removing from the stable ordering $\Lambda'$ each edge $t$ and others that come after $t$ in the same sequence as they occur in $\Lambda'$. We analyze consecutive intermediate graphs, their stable orderings, and the edges actually that end up being added by the edge addition algorithm. As the edge addition algorithm removes all edges of a node once an edge incident is added, we observe that only one of these $t$ edges is added. All other edges incident to $v^+$ are removed. We prove this extra edge leads to decisions in the edge addition algorithm that always restricts such consecutive intermediate graphs to differ by at most $1$ edge. This proves the lemma for $\mininconsistency$ as at most $t$ (upper bounded by $\theta$) edges can differ between two neighboring graphs. 
\end{proof}
}

We now analyze the sensitivity of the quality function using both measures' sensitivity analysis.

\begin{lemma} \label{lemma:sens_quality}
The sensitivity of the quality function $q_{\epsilon_2}(\mathcal{G}, \theta_i)$ in Algorithm~\ref{algo:expo_mech_basic} defined in Equation~\eqref{eq:quality_function} is $\theta_{\max}=\max(\Theta)$. 
\end{lemma}

\reva{
\begin{proof}[Proof sketch]
We prove the theorem for the $\mininconsistency$ measure and show that it is similar for $\problematic$. The sensitivity of the quality function is computed by comparing the respective quality functions of two neighboring graphs $\mathcal{G}$ and $\mathcal{G}'$ with an extra node. It is upper bound by the difference of two terms $\left(|\mathcal{G}'_{\theta_{\max}}.E| - |\mathcal{G}_{\theta_{\max}}.E|\right) - \left(|\mathcal{G}'_\theta.E| - |\mathcal{G}_\theta.E|\right)$. The first term $\left(|\mathcal{G}'_{\theta_{\max}}.E| - |\mathcal{G}_{\theta_{\max}}.E|\right)$ is the sensitivity of the measures, as already proved by Lemma~\ref{lemma:sens_lap} is equal to $\theta_{max}$. The second term $\left(|\mathcal{G}'_\theta.E| - |\mathcal{G}_\theta.E|\right)$ is always $\geq 0$ as  $|\mathcal{G}'_\theta.E| \geq |\mathcal{G}_\theta.E|$ as discussed in the proof for Lemma~\ref{lemma:sens_lap}.
\end{proof}
}

\ifpaper
\else
Proofs for \cref{thm:privacy_proof_dc_oblivious}, \cref{lemma:sens_quality}, and \cref{lemma:sens_lap} can be found in \cref{app:graph_general}.
\fi

\paratitle{Utility analysis}
The utility of Algorithm~\ref{algo:graph_general} is directly encoded by the quality function of the exponential mechanism in Algorithm~\ref{algo:expo_mech_basic}. 
We first define the best possible quality function value for a given database and its respective graph as 
\begin{equation}
    q_{\opt}(D,\epsilon_2) = \max_{\theta\in \Theta} q_{\epsilon_2}(\graph,\theta)
\end{equation}
and the set of degree values that obtain the optimal quality value as 
\begin{equation}
 \Theta_{\opt} = \{\theta\in \Theta: q_{\epsilon_2}(\graph,\theta) = q_{\opt}(D,\epsilon_2) \}.   \end{equation} 
However, we define $e_{\text{bias}}$ as the difference in the number of edges or nodes in the projected graph $\mathcal{G}_{\theta}$ compared to that of $\mathcal{G}_{\theta_{\max}}$, instead of $\mathcal{G}$. This is to limit the sensitivity of the quality function. To compute the utility, we slightly modify the quality function without affecting the output of the exponential mechanism. 
\begin{equation}
    \tilde{q}_{\epsilon_2}(\mathcal{G},\theta) = q_{\epsilon_2}(\mathcal{G},\theta) + f(\mathcal{G}_{\theta_{\max}}) - f(\graph),
\end{equation}
where $f(\cdot)$ returns edge count for $\mininconsistency$ and the number of nodes with positive degrees for $\problematic$.
This modified quality function should give the same set of degrees $\Theta_{\opt}$ with optimal values equal to 
\begin{equation}
    \tilde{q}_{\opt}(D,\epsilon_2) = \max_{\theta\in \Theta} q_{\epsilon_2}(\graph,\theta) + f(\graph_{\theta_{\max}}) - f(\graph).
\end{equation}

Then, we derive the utility bound for Algorithm~\ref{algo:graph_general} based on the property of the exponential mechanism as follows.

\begin{theorem}\label{thm:graph_general_utility} On any database instance $D$ and its respective conflict graph $\graph$, let $o$ be the output of Algorithm~\ref{algo:graph_general} with Algorithm~\ref{algo:expo_mech_basic} over $D$.  
Then,  with a probability of at least $1-\beta$, we have 
\begin{equation}
|o-a| \leq -\tilde{q}_{\opt}(D,\epsilon_2) + \frac{2 \theta_{\max}}{\epsilon_1} (\ln \frac{2|\Theta|}{|\Theta_{\opt}|\cdot \beta}) 
\end{equation}
where $a$ is the true inconsistency measure over $D$ and $\beta\leq \frac{1}{e^{\sqrt{2}}}$.
\end{theorem}

\ifpaper
\else
The proof can be found in ~\cref{app:graph_general_utility}.
\fi

This theorem indicates that the error incurred by Algorithm~\ref{algo:graph_general} with Algorithm~\ref{algo:expo_mech_basic} is directly proportional to the log of the candidate size $|\Theta|$ and the sensitivity of the quality function. \reva{The $\beta$ parameter in the theorem is a controllable probability parameter. According to the accuracy requirements of a user's analysis, one may set $\beta$ as any value less than this upper bound. For example, if we set $\beta=0.01$, then our theoretical analysis of Algorithm 2 that says the algorithm's output being close to the true answer will hold with a probability of $1-\beta = 0.99$. We also show a plot to show the trend of the utility analysis as a function of $\beta$ in Appendix A.5~\cite{full_paper}.} Without prior knowledge about the graph, $\theta_{\max}$ is usually set as the number of nodes $n$, and $\Theta$ includes all possible degree values up to $n$, resulting in poor utility. 
Fortunately, for our use case, the edges in the graph arise from the DCs that are available to us. In the next section, we show how we can leverage these constraints to improve the utility of our algorithm by truncating candidates in the set $\Theta$.

\subsection{Optimized Parameter Selection}\label{sec:dc_aware}

Our developed strategy to improve the parameter selection includes two optimization techniques. The overarching idea behind these optimizations is to gradually truncate large candidates from the candidate set $\Theta$ based on the density of the graph. For example, we observe that the Stock dataset~\cite{oleh_onyshchak_2020} has a sparse conflict graph, and its optimum degree for graph projection is in the range of $10^0-10^1$. In contrast, the graph for the Adult dataset sample~\cite{misc_adult_2} is extraordinarily dense and has an optimum degree $\theta$ greater than $10^3$, close to the sampled data size. 
Removing unneeded large candidates, especially those greater than the true maximum degree of the graph, can help the high sensitivity issue of the quality function and improve our chances of choosing a better bound. 

Our first optimization estimates an upper bound for the true maximum degree of the conflict graph and removes candidates larger than this upper bound from the initial candidate set. The second optimization is a hierarchical exponential mechanism that utilizes two steps of exponential mechanisms. The first output, $\theta^1$, is used to truncate $\Theta$ further by removing candidates larger than $\theta^1$ from the set, and the second output is chosen as the final candidate $\theta^*$. In the rest of this section, we dive deeper into the details of these optimizations and discuss their privacy analysis.

\paratitle{Estimating the degree upper bound using FDs} \label{sec:dc_aware_ub}
Given a conflict graph $\graphsimple(V,E)$, we use $\degree(\graphsimple, v)$ to denote the degree of the node $v\in V$ in $\graphsimple$ and $\degree_{\max}(\graphsimple) = \max_{v\in V} \degree(\graphsimple, v)$ to denote the maximum degree in $\graphsimple$. We estimate $\degree_{\max}$ by leveraging how conflicts were formed for its corresponding dataset $D$ under $\constraintset$. 

The degree for each vertex in $\graphsimple$ can be found by going through each tuple $t$ in the database $D$ and counting the tuples that violate the $\constraintset$ jointly with $t$.
However, computing this value for each tuple is computationally expensive and highly sensitive, making it impossible to learn directly with differential privacy.  
We observe that the conflicts that arise due to functionality dependencies (FDs)  depend on the values of the left attributes in the FD. 
\begin{example}
    Consider the same setup as  Example~\ref{example:running_example}  and an FD 
$\sigma: \text{Capital}\rightarrow \text{Country}$. 
We can see that the number of violations added due to the erroneous grey row is 3. This number is also one smaller than the maximum frequency of values occurring in the Capital attribute, and the most frequent value is ``Ottawa''.
\end{example}

Based on this observation, we can derive an upper bound for the maximum degree of a conflict graph if it involves only FDs, and this upper bound has a lower sensitivity. We show the upper bound in Lemma~\ref{lemma:fd_theta} for one FD first and later extend for multiple FDs.  

\begin{lemma}\label{lemma:fd_theta}
Given a database $D$ and 
a FD $\sigma: X\rightarrow Y$ as the single constraint,
where $X = \{A_1,\dots,A_k\}$ and $Y$ is a single attribute. For its respective conflict graph $\graphsimple^D_{\Sigma = \{\sigma\}}$, simplified as $\graphsimple^D_{\sigma}$, we have the maximum degree of the graph $\degree_{\max}(\graphsimple^D_{\sigma})$ upper bounded by
    \begin{equation}
     \degreebound(D,X)
    =
    \max_{\vec{a_X} \in \dom(A_1) \times \ldots \times \dom(A_k) }  \normalfont{\text{freq}}(D, \vec{a_X})-1, 
    \end{equation}    
where $\normalfont{\text{freq}}(D, \vec{a_X})$ is the frequency of values $\vec{a_X}$ occurring for the attributes $X$ in the database $D$. 
The sensitivity for $\degreebound(D,X)$ is 1. 
\end{lemma}

\begin{proof}
An FD violation can only happen to a tuple $t$ with other tuples $t'$ that share the same values for the attributes $X$. 
Let $\vec{a_X}^*$ be the most frequent value for $X$ in $D$, i.e., 
$$\vec{a_X}^*=\text{argmax}_{\vec{a_X} \in \dom(A_1) \times \ldots \times \dom(A_k) } \normalfont{\text{freq}}(D, \vec{a_X}).$$
In the worst case, a tuple $t$ has the most frequent value $\vec{a_X}^*$ for $X$ but has a different value in $Y$ with all the other tuples with $X=\vec{a_X}^*$. Then the number of violations involved by $t$ is $\text{freq}(D,\vec{a_X}^*)-1$.

Adding a tuple or removing a tuple to a database will change, at most, one of the frequency values by 1. Hence, the sensitivity of  the maximum frequency values is 1.
\end{proof}

Now, we will extend the analysis to multiple FDs. 
\begin{theorem}\label{theorem:fd_theta}
Given a database $D$ and a set of FDs 
$\constraintset=\{\sigma_1,\ldots,\sigma_l \}$, for its respective conflict graph $\graph$, we have the maximum degree of the graph $\degree_{\max}(\graph)$ upper bounded by
\begin{eqnarray}
    \degreebound(D,\Sigma)
     =      \sum_{(\sigma: X\rightarrow Y)\in \Sigma} 
         \degreebound(D,X)
\end{eqnarray}
\end{theorem}

\proof
By Lemma~\ref{lemma:fd_theta}, 
for each FD $\sigma: X\rightarrow Y$, a tuple may violate at most $\degreebound(D,X)$ number of tuples. In the worst case, the same tuple may violate all FDs. \qed

We will spend some privacy budget $\epsilon_0$ to perturb the upper bound $\degreebound(D,X)$ for all FDs with LM and add them together. Each FD is assigned with  $\epsilon_0/|\Sigma_{\text{FD}}|$, where
$\Sigma_{\text{FD}}$ is the set of FDs in $\Sigma$. We denote this perturbed upper bound as $\noisydegreebound$ and add it to the candidate set $\Theta$ if absent.

\paratitle{Extension to general DCs}
The upper bound derived in Theorem~\ref{theorem:fd_theta}
only works for FDs but fails for general DCs. General DCs have more complex operators, such as ``greater/smaller than,'' in their formulas. Such inequalities require the computation of tuple-specific information, which is hard with DP. For example, consider the DC $\sigma: \neg(t_i[gain] > t_j[gain] \land t_i[loss] < t_j[loss])$ saying that if the gain for tuple $t_i$ is greater than the gain for tuple $t_j$, then the loss for $t_i$ should also be greater than $t_j$. We can observe that similar analyses for FDs do not work here as the frequency of a particular domain value in $D$ does not bound the number of conflicts related to a tuple. 
Instead, we have to iterate each tuple $t$'s gain value and find how many other tuples $t'$s violate this gain value. In the worst case, such a computation may have a sensitivity equal to the data size. Therefore, estimation using DCs may result in much noise, especially when the dataset has fewer conflicts, and the noise is added to correspond to the large sensitivity.

Our experimental study (\Cref{sec:experiments}) shows that datasets with general DCs have dense conflict graphs, which favors larger $\theta$s for graph projection. Hence, if we learn a small noisy upper bound  $\noisydegreebound$ based on the FDs with LM, we will first prune all degree candidates smaller than $\noisydegreebound$, but then include $|V|$, which corresponds to the case when no edges are truncated, and Laplace mechanism is applied with the largest possible sensitivity $|V|$, i.e., 
\begin{equation}\label{eq:fd-bound}
    \Theta' = \{\theta\in \Theta~|~\theta \leq \noisydegreebound\} \cup \{|V|\}.
\end{equation}
Though the maximum value in $\Theta'$ is $|V|$, the sensitivity of the quality function over the candidate set $\Theta'$ remains $\noisydegreebound$. For the $|V|$ candidate, the quality function only depends on the Laplace error $\frac{\sqrt{2}|V|}{\epsilon_2}$ and has no error from $e_{\text{bias}}$ as no edges will be truncated. 
\revc{Despite being tailored for FDs, we show that, in practice, our approach is cheap and performs well for DCs. In \Cref{sec:experiments}, we show that this approach works well for the dense Adult~\cite{adult} dataset where we compute the $\problematic$ using this strategy in Figure~\ref{fig:comparing_strategies}. Developing a specific strategy for DCs is an important direction of future work.}
In practice, one may skip this upper bound calculation process and skip directly to the two-step exponential mechanism if it is known that the graph is too dense or contains few FDs and more general DCs. We discuss this in detail in the experiments section.

\paratitle{Hierarchical EM}\label{sec:dc_aware_hier_expo_mech}
The upper bound $\degreebound$ may not be tight as it estimates the maximum degree in the worst case. The graph would be sparse with low degree values, and there is still room for pruning. 
To further prune candidate values in the set $\Theta$, we use a hierarchical EM that first samples a degree value $\theta^*$ to prune values in $\Theta$ and then sample again another value $\theta^*$ from the remaining candidates as the final degree the graph projection. 
Our work uses a two-step hierarchical EM by splitting the privacy budget equally into halves. One may extend this EM to more steps at the cost of breaking their privacy budget more times, but in practice, we notice that a two-step is enough for a reasonable estimate. 

\begin{example}\label{example:parameter_selection} 
    Consider the same setup as Example~\ref{example:running_example}. For this dataset, we start with $\Theta = [1, 2, 3]$ and 
    the $\theta_{\max}$ for this setup is 3. Assume no values are pruned in the first optimization phase. 
    We compare a single versus a two-step hierarchical EM for the second optimization step. From Table~\ref{tab:example_expo_prob} in Example~\ref{example:quality_function}, we know that the $\theta_1$ has the best quality. However, as the quality values are close, the probability of choosing the best candidate is similar, as shown in Table~\ref{tab:example_expo_prob} with $\epsilon=1$.
    \begin{table}[]
        \centering
        \begin{tabular}{|c|c|c|c|c|}
             \hline
             $\theta$ & $q$ & EM & 2-EM ($\theta^*_1 = \theta_3$) & 2-EM ($\theta^*_1 = \theta_2$)\\
             \hline
             1 & $-3.41$ & $0.35$ & $0.51$ & $1$\\
             2 & $-3.82$ & $0.33$ & $0.49$ & -\\
             3 & $-4.24$ & $0.31$ & - & -\\
             \hline
        \end{tabular}
        \caption{Probabilities of candidates with the exponential mechanism (EM) vs.~the two-step hierarchical exponential mechanism (2-EM). $\theta^*_1$ refers to the first-step output of 2-EM.}
        \label{tab:example_expo_prob}
    \end{table}
    The exponential mechanism will likely choose a suboptimal candidate in such a scenario as the probabilities are close. But if a two-step exponential mechanism is used even with half budget $\epsilon = 0.5$, the likelihood of choosing the best candidate $\theta_1$ goes up to $0.51$ if the first step chose $\theta_3$ or $1$ if the first step chosen $\theta_2$.
\end{example}

\begin{algorithm}[t]
\caption{Optimized EM for parameter selection}
\label{algo:em_opt}
    \KwData{Graph $\mathcal{G} (V,E)$, candidate set $\Theta=\{1,\ldots,|V|\}$, quality function $q$, privacy budget $\epsilon_1, \epsilon_2$ }
    \KwResult{Candidate $\theta^*$}
    
    \If{$\Sigma$ mainly consists of FDs} {$\epsilon_0\gets\epsilon_1/4$, $\epsilon_1 \gets \epsilon_1-\epsilon_0$\\
    Compute noisy upper bound $\noisydegreebound\gets \sum_{\sigma:X\rightarrow Y} (\degreebound(D,X)+\text{Lap}(|\Sigma_{\text{FD}}|/\epsilon_0))$\\
     }{}
    Prune candidates $\Theta \gets \{\theta\in \Theta~|~\theta \leq \noisydegreebound\} \cup \{\noisydegreebound, |V|\}$ \\
    Set $\theta_{\max}\gets \min(\noisydegreebound,|V|)$ \\ 
    \For{$s\in \{1,2\}$ }{
   For each $\theta_i \in \Theta$, compute $q_{\epsilon_2}(\mathcal{G}, \theta_i)$     \commenttext{// See Equation~\eqref{eq:quality_function}}
    \\
    Sample $\theta^*$ with prob $\propto \exp( \frac{\frac{\epsilon_1}{2} q_{\epsilon_2}(\mathcal{G}, \theta_i)}{2\theta_{\max}})$ \\
    Prune candidates $\Theta \gets \{\theta\in \Theta~|~\theta \leq \theta^*\}$ \\
    Set $\theta_{\max}\gets \theta^*$  
    }
    {\bf Return} $\theta^*$\\
\end{algorithm}

\paratitle{Incorporating the optimizations into the algorithm}
Algorithm~\ref{algo:em_opt} outlines the two optimization techniques. First, we decide when to use the estimated upper bound for the maximum degrees, for example, when the constraint set $\Sigma$ mainly consists of FDs. We will spend part of the budget $\epsilon_0$ from $\epsilon_1$ to perturb the upper bounds $\degreebound(D,X)$ for all FDs with Laplace mechanism and add them together (lines 1-3). The noisy upper bound $\noisydegreebound$ prunes the candidate set (line 4). We also add $|V|$ to the candidate set if there are general DCs in $\Sigma$, and then set the sensitivity of the quality function $\theta_{\max}$ to be the minimum of the noisy upper bound or $|V|$ (line 5).
Then, we conduct the two-step hierarchical exponential mechanism for parameter selection (lines 6-10). Lines 7-8 work similarly to the previous exponential mechanism algorithm with half of the remaining $\epsilon_1$, where we choose a $\theta^*$ based on the quality function. However, instead of using it as the final candidate, we use it to prune values in $\Theta$ and improve the sensitivity $\theta_{\max}$ for the second exponential mechanism (lines 9-10). Then, we repeat the exponential mechanism and output the sampled $\theta^*$ (line 11).
Algorithm~\ref{algo:em_opt} has a similar complexity of $O(|\Theta|m)$ as Algorithm~\ref{algo:expo_mech_basic}, where $|E|$ is the edge size of the graph. The overall Algorithm~\ref{algo:graph_general} has a complexity of $O(|\Sigma|n^2+|\Theta|m)$.

\paratitle{Privacy and utility analysis}
The privacy analysis of the optimizations depends on the analysis of three major steps: $\degreebound$ computation with the Laplace mechanism, the two-step exponential mechanism, and the final measure calculation with the Laplace mechanism. By sequential composition, we have Theorem~\ref{thm:privacy_proof_dc_aware}.

\begin{theorem}\label{thm:privacy_proof_dc_aware}
    Algorithm~\ref{algo:graph_general} with the optimized EM in  Algorithm~\ref{algo:em_opt} satisfies $(\epsilon_1 + \epsilon_2)$-DP.
\end{theorem}
\reva{
\begin{proof}[Proof sketch]
The proof is similar to Theorem~\ref{thm:privacy_proof_dc_oblivious} and is due to the composition property of DP as stated in Proposition~\ref{prop:DP-comp-post}.
\end{proof}
}

We show a tighter sensitivity analysis for the quality function in EM over the pruned candidate set. The sensitivity analysis is given by Lemma~\ref{lemma:sens_quality_2stepEM} and is used for $\theta_{\text{max}}$ in line 8 of Algorithm~\ref{algo:em_opt}.

\begin{lemma} \label{lemma:sens_quality_2stepEM}
The sensitivity of $q_{\epsilon_2}(\mathcal{G}, \theta_i)$ in the 2-step EM (Algorithm~\ref{algo:em_opt})
defined in Equation~\eqref{eq:quality_function} is $\theta_{\max}= \min(\noisydegreebound,|V|)$ for 1st EM step and $\theta_{\max}=\theta^*$ for the 2nd EM step.
\end{lemma}

\reva{
\begin{proof}[Proof sketch]
The proof follows from \Cref{lemma:sens_quality}, substituting the $\theta_{max}$ with the appropriate threshold values for each EM step.
\end{proof}
}

\ifpaper
\else
The proofs for the theorem and lemma are at \cref{app:privacy_proof_dc_aware} and \cref{app:sens_quality_2stepEM}.
\fi 

The utility analysis in Theorem~\ref{thm:graph_general_utility} for Algorithm~\ref{algo:graph_general} with the basic EM (Algorithm~\ref{algo:expo_mech_basic}) still applies to the optimized EM (Algorithm~\ref{algo:em_opt}). The basic EM usually has $\theta_{\max}=|D|=|V|$ and the full budget $\epsilon_1$, while the optimized EM has a much smaller $\theta_{\max}$ and slightly lower privacy budget when the graph is sparse. 

In practice (\Cref{sec:experiments}), we see significant utility improvements by the optimized EM for sparse graphs. When the graph is dense, we see the utility degrade slightly due to a smaller budget for each EM. However, the degradation is negligible with respect to the true inconsistency measure.

\section{DP Minimum Vertex Cover for \repair\ }\label{sec:vertex_cover}

This section details our approach for computing the optimal repair measure, $\repair$, using the conflict graph. \repair\ is defined as the minimum number of vertices that must be removed to eliminate all conflicts within the dataset. For the conflict graph $\graph$, this corresponds to finding the minimum vertex cover -- an NP-hard problem. To address this, we apply a well-known polynomial-time algorithm that provides a 2-approximation for vertex cover~\cite{vazirani1997approximation}. This randomized algorithm iterates through a random ordering of edges, adding both nodes of each edge to the vertex cover if they haven't been encountered, then removes all incident edges. The process repeats until the edge list is exhausted. 
In our setting, we aim to compute the minimum vertex cover size while satisfying DP. A straightforward approach would be to analyze the sensitivity of the 2-approximation algorithm and add the appropriate DP noise. However, determining the sensitivity of this naive approximation is challenging, as the algorithm's output can fluctuate significantly depending on the order of selected edges. This variability is illustrated in \Cref{example:naive_vertexcover}.

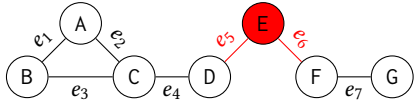
\begin{figure}
    \begin{center}
    \scalebox{1.0}{
    \begin{tikzpicture}[main/.style = {draw, circle}]
        \node[main] (1) {{\tt B}}; 
        \node[main] (2) [above right of=1] {{\tt A}};
        \node[main] (3) [below right of=2] {{\tt C}};
        \node[main] (4) [right of=3] {{\tt D}};
        \node[main, fill=red] (5) [above right of=4] {{\tt E}};
        \node[main] (6) [below right of=5] {{\tt F}};
        \node[main] (7) [right of=6] {{\tt G}};
        \draw (1) -- node[midway, above, sloped]{$e_1$} (2);
        \draw (2) -- node[midway, above, sloped]{$e_2$} (3);
        \draw (1) -- node[midway, below, sloped]{$e_3$} (3);
        \draw (3) -- node[midway, below, sloped]{$e_4$}(4);
        \draw[red] (5) -- node[midway, above, sloped]{$e_6$}(6);
        \draw[red] (4) --  node[midway, above, sloped]{$e_5$}(5);
        \draw (6) -- node[midway, below, sloped]{$e_7$}(7);
    \end{tikzpicture} 
    }
\end{center}
\caption{Toy graph example $\mathcal{G}$ with seven nodes ({\tt A} to {\tt G}) and seven edges. Consider a neighboring graph $\mathcal{G}'$ with the differing node {\tt E} (red) and its two edges. }
\label{fig:examplegraph}
\end{figure}

\begin{example}
Let us consider a graph $\mathcal{G}$ with 7 vertices {\tt A} to {\tt G} and 7 edges $e_1$ to $e_7$ as shown in Figure~\ref{fig:examplegraph}. We can have a neighboring graph $\mathcal{G}'$ by considering the vertex {\tt E} as the differing vertex and two of its edges $e_5$ and $e_6$ as the differing edges. This example shows that according to the vanilla 2-approximate algorithm, the output for the graphs $\mathcal{G}$ and $\mathcal{G}'$ may vary drastically. For $\mathcal{G}$, if $e_2$ is selected followed by $e_7$, then the vertex cover size is 4. However, for graph $\mathcal{G}'$, if $e_1$ or $e_4$ is selected first and subsequently after the other one $e_6$ is selected, then the output is 6. Moreover, this difference may get significantly large if the above graph is stacked multiple times and the corresponding vertex that creates this difference is chosen every time. 
\label{example:naive_vertexcover}
\end{example}

\begin{algorithm}[t]
\caption{DP approximation of minimum vertex cover size for $\repair$}
\label{algo:dp_vertexcover}
    \KwData{Graph $\mathcal{G} (V,E)$, stable global ordering $\Lambda$, privacy parameter $\epsilon$}
    \KwResult{DP minimum vertex cover size for $\repair$}
    Initialize vertex cover set $C = \emptyset$ and size $c = 0$ \\
    Initialize edge list $E_0 = E$ \\
    \For{$i \in \{1 \dots |\Lambda|\}$}{
        pop edge $e_i = \{u, v\}$ in order from $\Lambda$\\
        add $u$ and $v$ to $C$ and $c = c + 2$\\
        $E_{i+1}$ = remove all edges incident to $u$ or $v$ from $E_i$\\    
    }
    {\bf Return} $c$\ + Lap($\epsilon$/2)
\end{algorithm}

To solve the sensitivity issue, we make a minor change in the algorithm by traversing the edges in a particular order (drawing on~\cite{day2016publishing}). We use a similar stable ordering $\Lambda$ defined in \Cref{sec:prelim-dp}. The new algorithm is shown in Algorithm~\ref{algo:dp_vertexcover}. We initialize an empty vertex cover set $C$, its size $c$, and an edge list (lines 1--2). We then start an iteration over all edges in the same ordering as the stable ordering $\Lambda$ (line 3). For each edge $e_i = \{u,v\} \in E$ that is part of the graph, we add both $u$ and $v$ to $C$ and correspondingly increment the size $c$ (lines 4--5). We remove all other edges, including $e_i$, connected to $u$ or $v$ from $E$ and continue the iteration (line 5). Finally, we return the noisy size of the vertex cover (line 6). The sensitivity of this algorithm is given by \cref{prop:vertexcover_sens}.

\begin{proposition}\label{prop:vertexcover_sens}
Algorithm~\ref{algo:dp_vertexcover} obtains a vertex cover, and its size has a sensitivity of 2.
\end{proposition}

\ifpaper
\else
The proof can be found in ~\cref{app:vertext_cover_sensitivity}.
\fi

\begin{example}
    Let us consider our running example in Figure~\ref{fig:examplegraph} as input to Algorithm~\ref{algo:dp_vertexcover} and use it to understand the proof. We have two graphs -- $\mathcal{G}$ which has $6$ vertices $V = [{\tt A, B, C, D, F, G}]$ and edges $E = [e_1, e_2, e_3, e_4, e_7]$ and $\mathcal{G}'$ has $7$ vertices $V' = [{\tt A, B, C, D, E, F, G}]$ and edges $E' = [e_1, e_2, \dots, e_7]$. The total possible number of edges is $\binom{7}{2} = 21$, and we can have a global stable ordering of the edges $\Lambda$ depending on the lexicographical ordering of the vertices as $e_1, e_2, e_3, \dots, e_{21}$. When the algorithm starts, both vertex cover sizes are initialized to $c=0, c'=0$, and the algorithm's state is in Case 1 with $v^* = {\tt E}$. We delineate the next steps of the algorithm below:
    \begin{itemize*}
        \item Iteration 1 (Subcase 1b) : $e_1 ({\tt A,B})$ is chosen. {\tt A} and {\tt B} are both in $E_0$ and $E'_0$. Hence $c = 2, c'= 2$.
        \item Iteration 2, 3 (Subcase 1c) : $e_2 ({\tt A,C})$ and $e_3 ({\tt B,C})$ are chosen. Both are removed in iteration 1. Hence $c = 2, c'= 2$.
        \item Iteration 4 (Subcase 1b) : $e_4 ({\tt C, D})$ is chosen. {\tt C} and {\tt D} are both in $E_3$ and $E'_3$. Hence $c = 4, c'= 4$.
        \item Iteration 5 (Subcase 1c) : $e_5 ({\tt D,E})$ is chosen, removed from $E'_4$ in Iteration 4, and was never present in $E$. Hence $c = 4, c'= 4$.
        \item Iteration 6 (Subcase 1a) : $e_6 ({\tt E, F})$ is chosen. It is in $E'_5$ but not in $E_5$. Hence, $c = 4, c'= 6$ and the new $v^* = {\tt E}$.
        \item Iteration 7 (Subcase 2a) : $e_7 ({\tt F,G})$ is chosen. It is in $E_6$ but removed from $E_6$ in Iteration 6. Hence, $c = 6, c'= 6$, and the algorithm is complete.
    \end{itemize*}
\end{example}

\paratitle{Privacy and utility analysis}
We now show the privacy and utility analysis of Algorithm~\ref{algo:dp_vertexcover} using Theorem~\ref{thm:vertex_cover_priv_util_analysis} below.

\begin{theorem}~\label{thm:vertex_cover_priv_util_analysis}
    \reva{Algorithm~\ref{algo:dp_vertexcover} satisfies $\epsilon$-node DP and, prior to adding noise in line 7, obtains a 2-approximation vertex cover size.} 
\end{theorem}

\begin{proof}
The Algorithm~\ref{algo:dp_vertexcover} satisfies $\epsilon$-node DP as we calculate the private vertex cover using the Laplace mechanism with sensitivity $2$ according to Proposition~\ref{prop:vertexcover_sens}. It is also a 2-approximation as the stable ordering $\Lambda$ in Algorithm~\ref{algo:dp_vertexcover} can be perceived as one particular random order of the edges and hence has the same utility as the original 2-approximation algorithm. 
\end{proof}

\section{Experiments}\label{sec:experiments}

This section presents our experiment results on computing the three measures outlined in Section~\ref{sec:graph-algorithms-graphproj} and Section~\ref{sec:vertex_cover}. The questions that we ask through our experiments are as follows:
\begin{enumerate}
    \item How far are the private measures from the true measures?
    \item How do the different strategies for the degree truncation bound compare against each other? 
    \item How do our methods perform at different privacy budgets?
\end{enumerate}

\begin{table}[b]
\small
    \centering
    \begin{tabular}{|c|c|c|c|c|c|c|}
    \hline
Dataset & \#Tuple & \#Attrs & \#DCs(\#FDs) & 
\begin{tabular}[c]{@{}c@{}}Max Deg \\ 1\% RNoise\end{tabular}
  \\
    \hline
    Adult~\cite{misc_adult_2} & 32561 & 15 & 3 (2) & 9635\\
    Flight~\cite{flight} & 500000 & 20 & 13 (13) & 1520\\
    Hospital~\cite{hospital} & 114919 & 15 & 7 (7) & 793\\
    Stock \cite{oleh_onyshchak_2020} & 122498 & 7 & 1 (1) & 1\\
    Tax \cite{chu2013discovering} & 1000000 & 15 & 9 (7) & 373\\
    \hline
    \end{tabular}
    \caption{Description of datasets. The max deg column shows the maximum degree of any node in a $10k$ rows subset of the conflict graph of the dataset with 1\% RNoise. }
    \label{tab:datasets}
\end{table}

\begin{figure*}
    \begin{subfigure}[b]{\textwidth}
    \centering
    \includegraphics[width=0.19\textwidth]{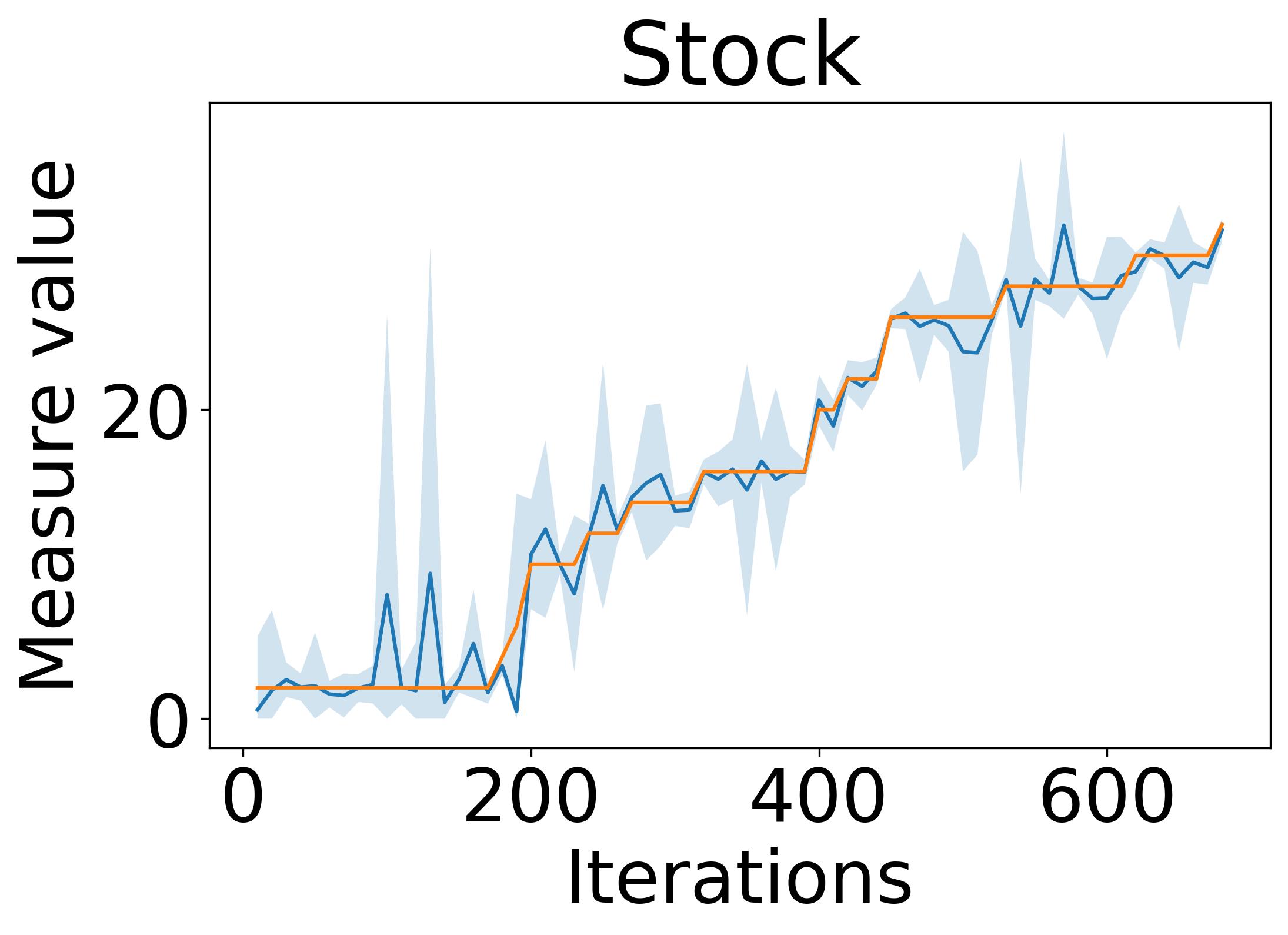}
    \hfill
    \includegraphics[width=0.19\textwidth]{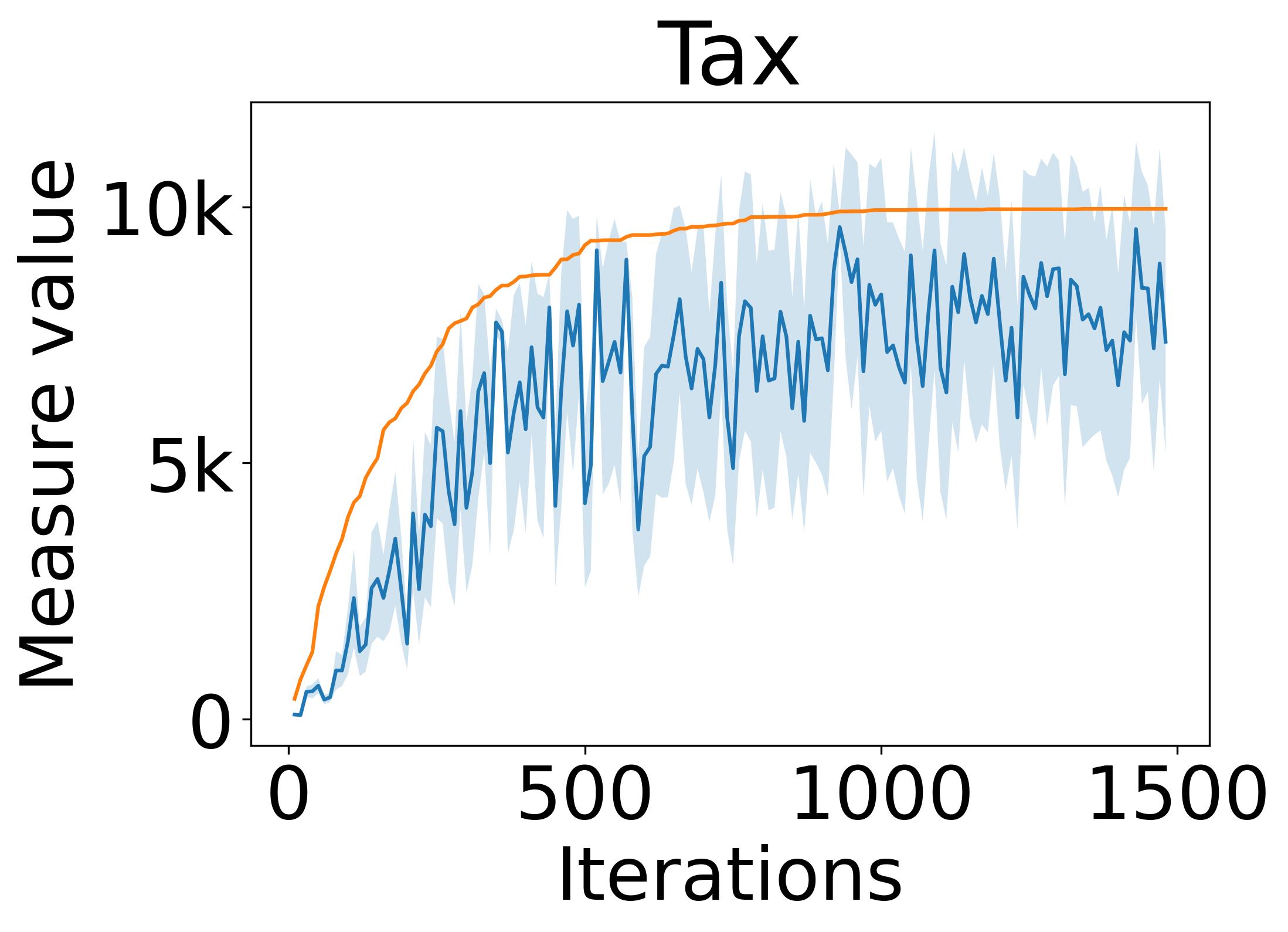}
    \hfill
    \includegraphics[width=0.19\textwidth]{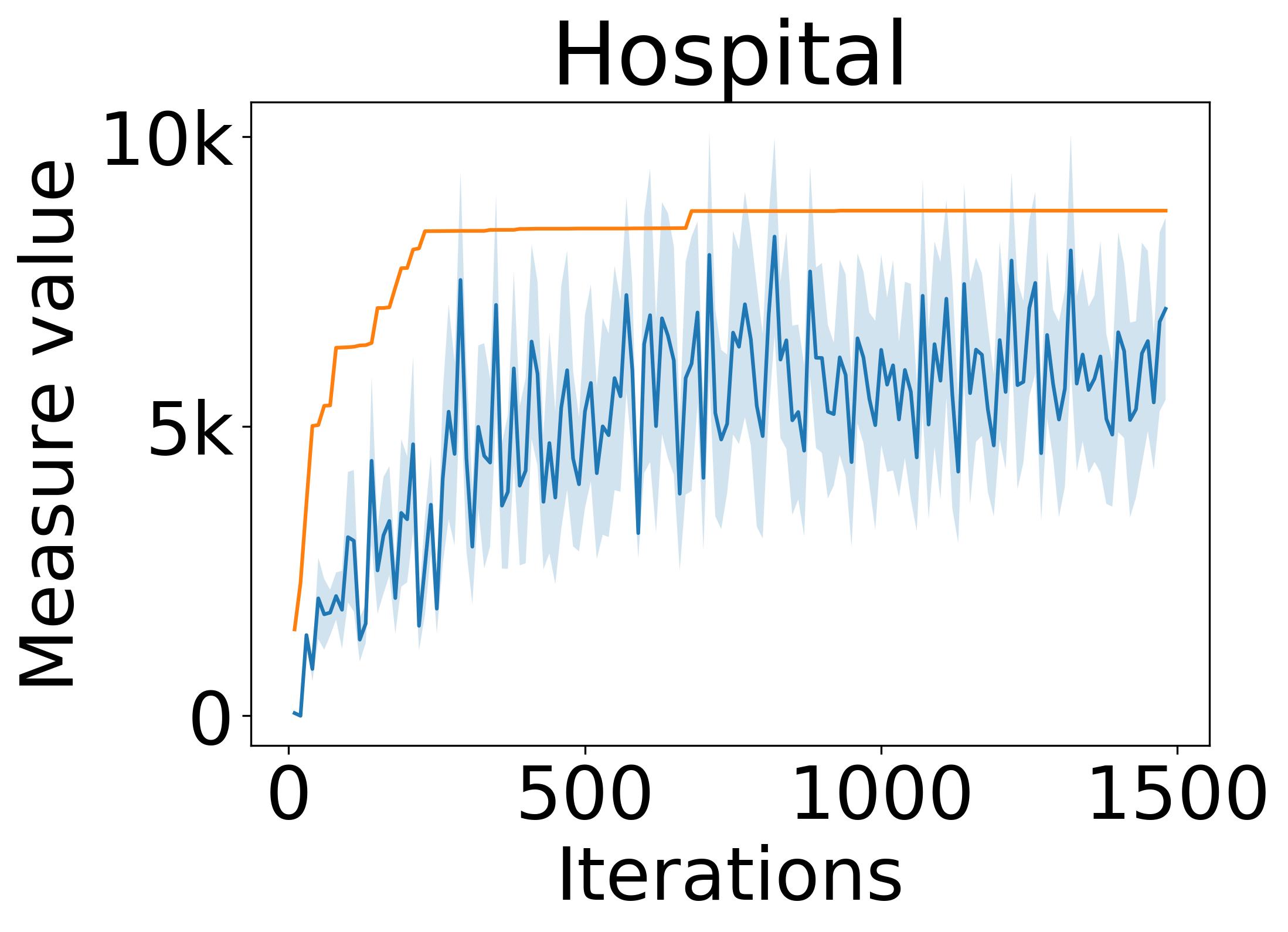}
    \hfill
    \includegraphics[width=0.19\textwidth]{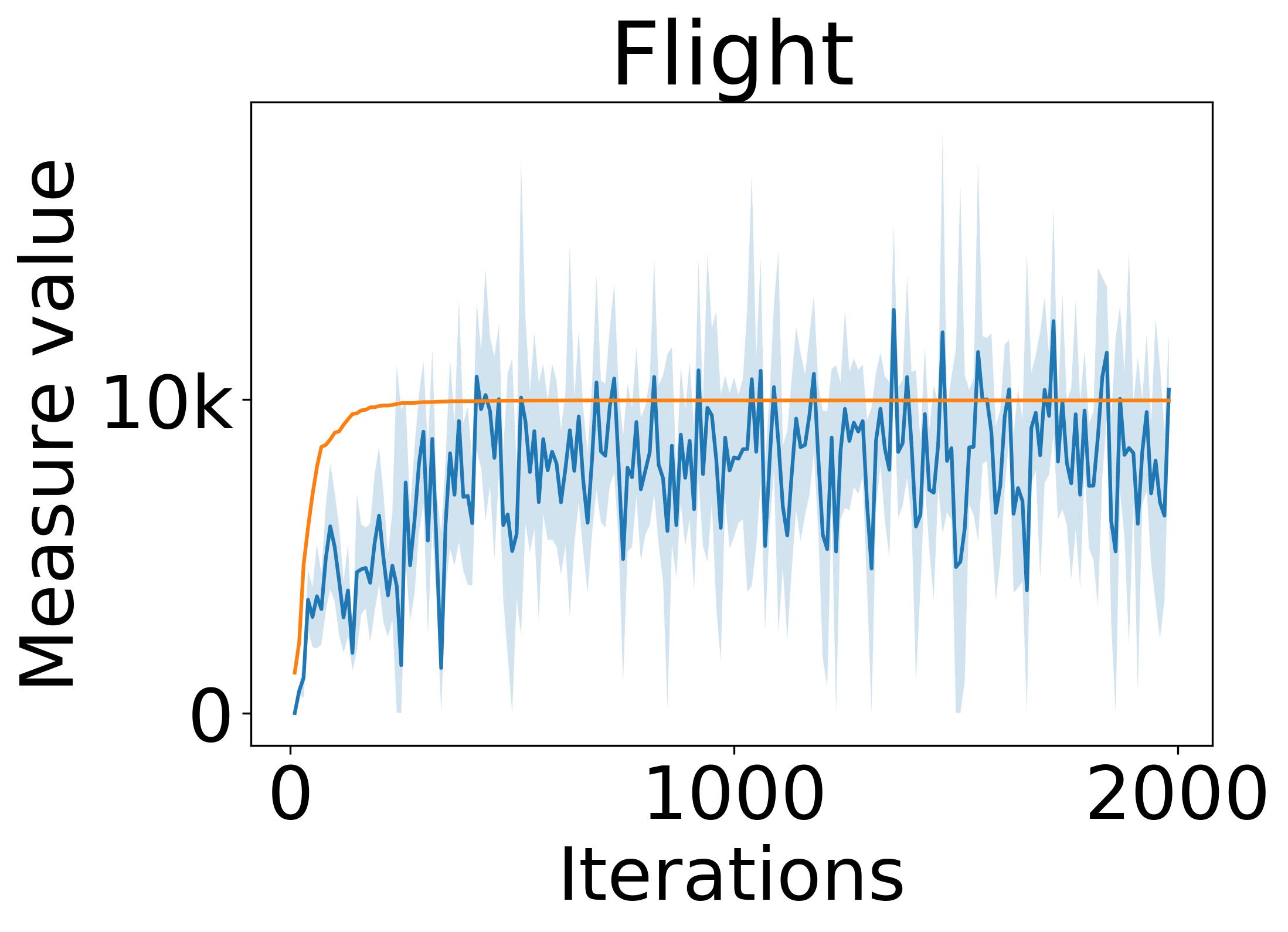}
    \hfill
    \includegraphics[width=0.19\textwidth]{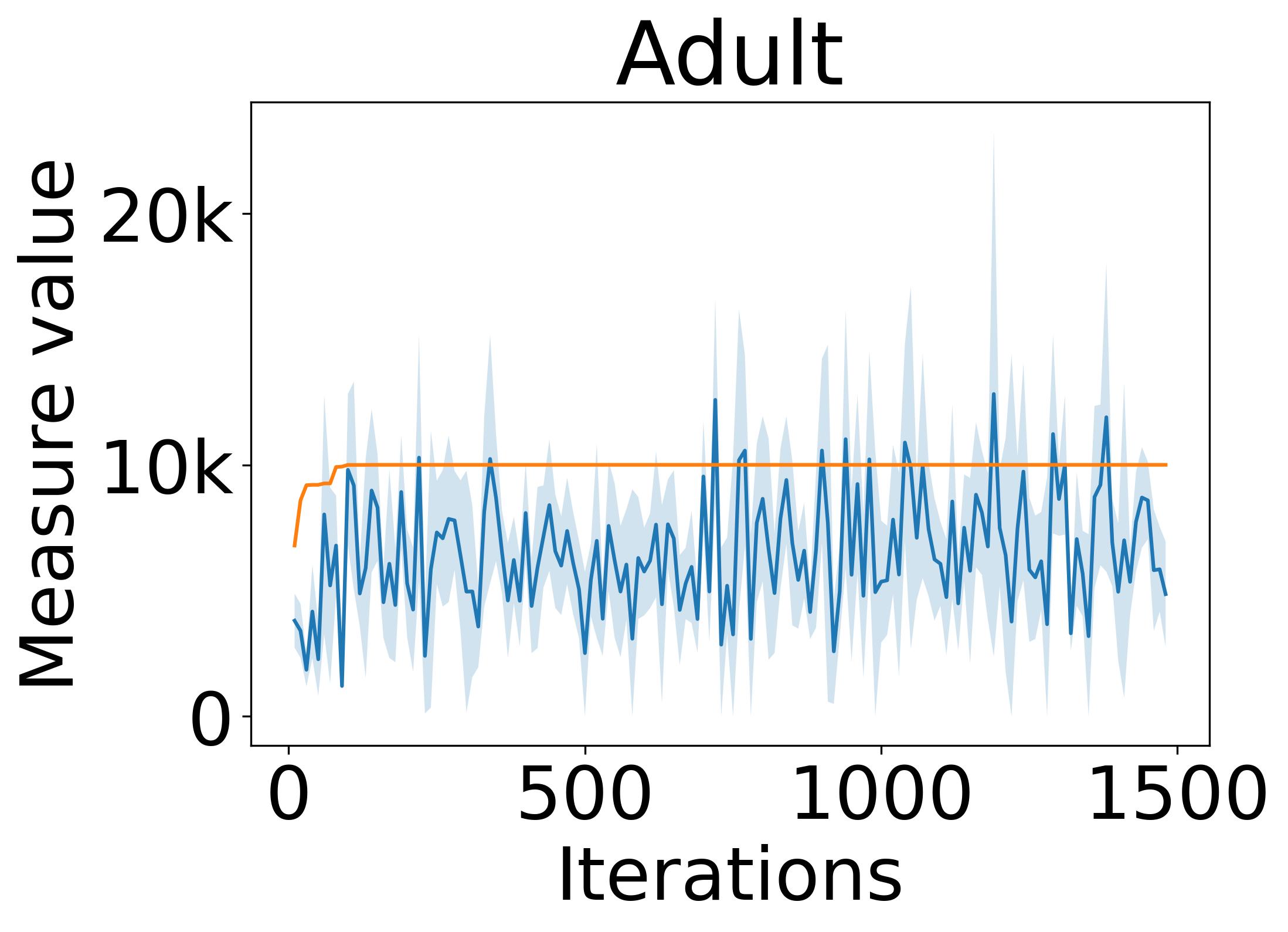}
    \includegraphics[width=0.3\textwidth]{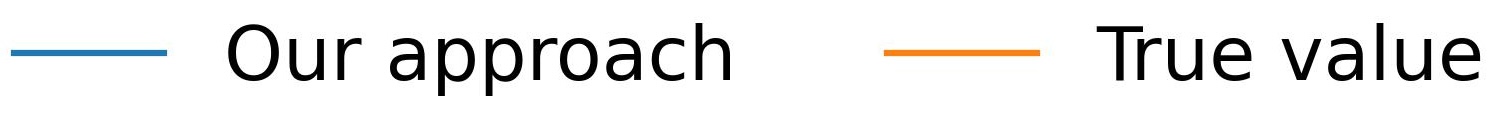}
    \caption{$\problematic$ (Positive degree nodes)}
    \label{fig:tp_rnoise_pdedges}
    \end{subfigure}
    \begin{subfigure}[b]{\textwidth}
    \centering
    \includegraphics[width=0.19\textwidth]{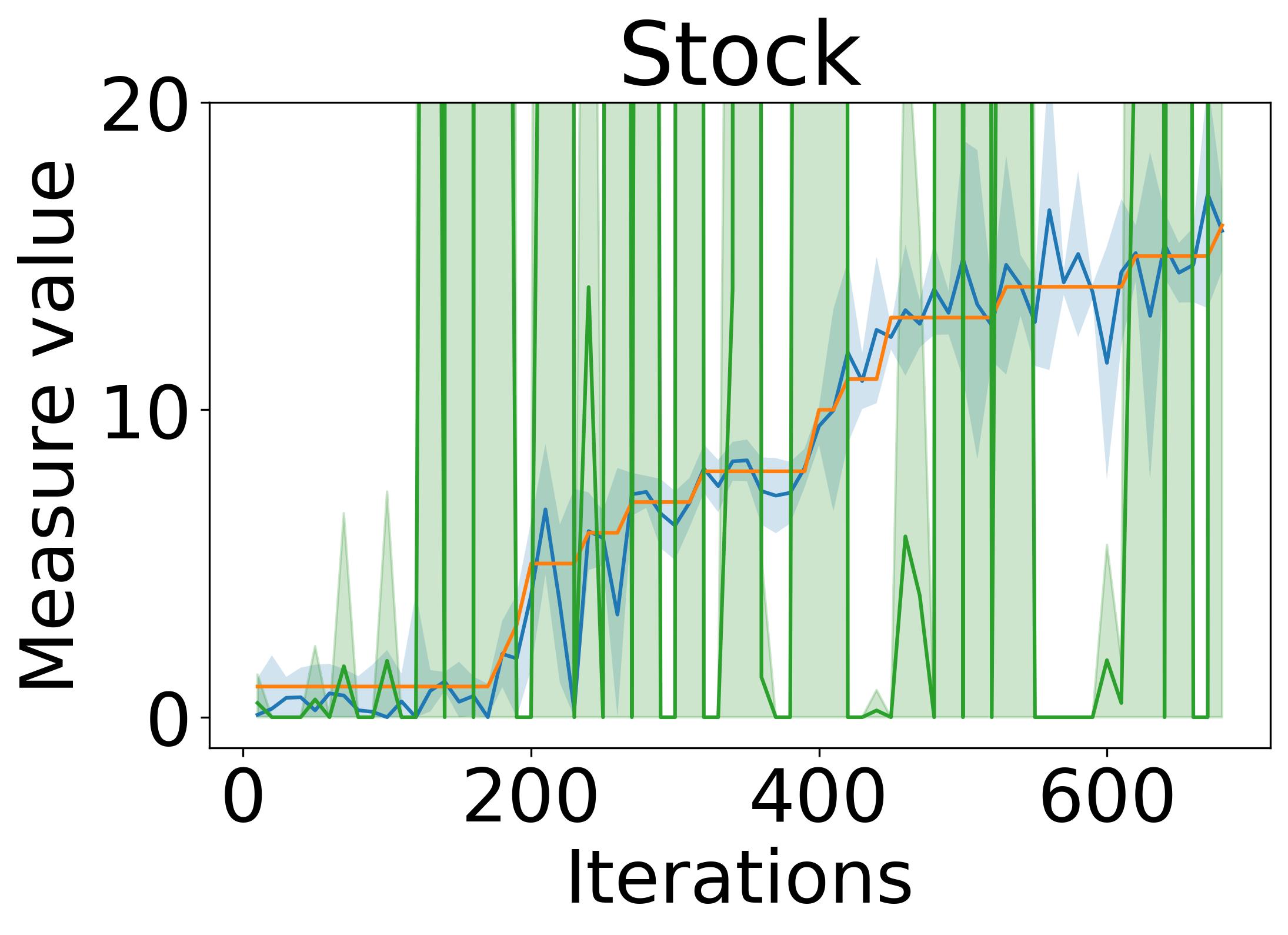}
    \hfill
    \includegraphics[width=0.19\textwidth]{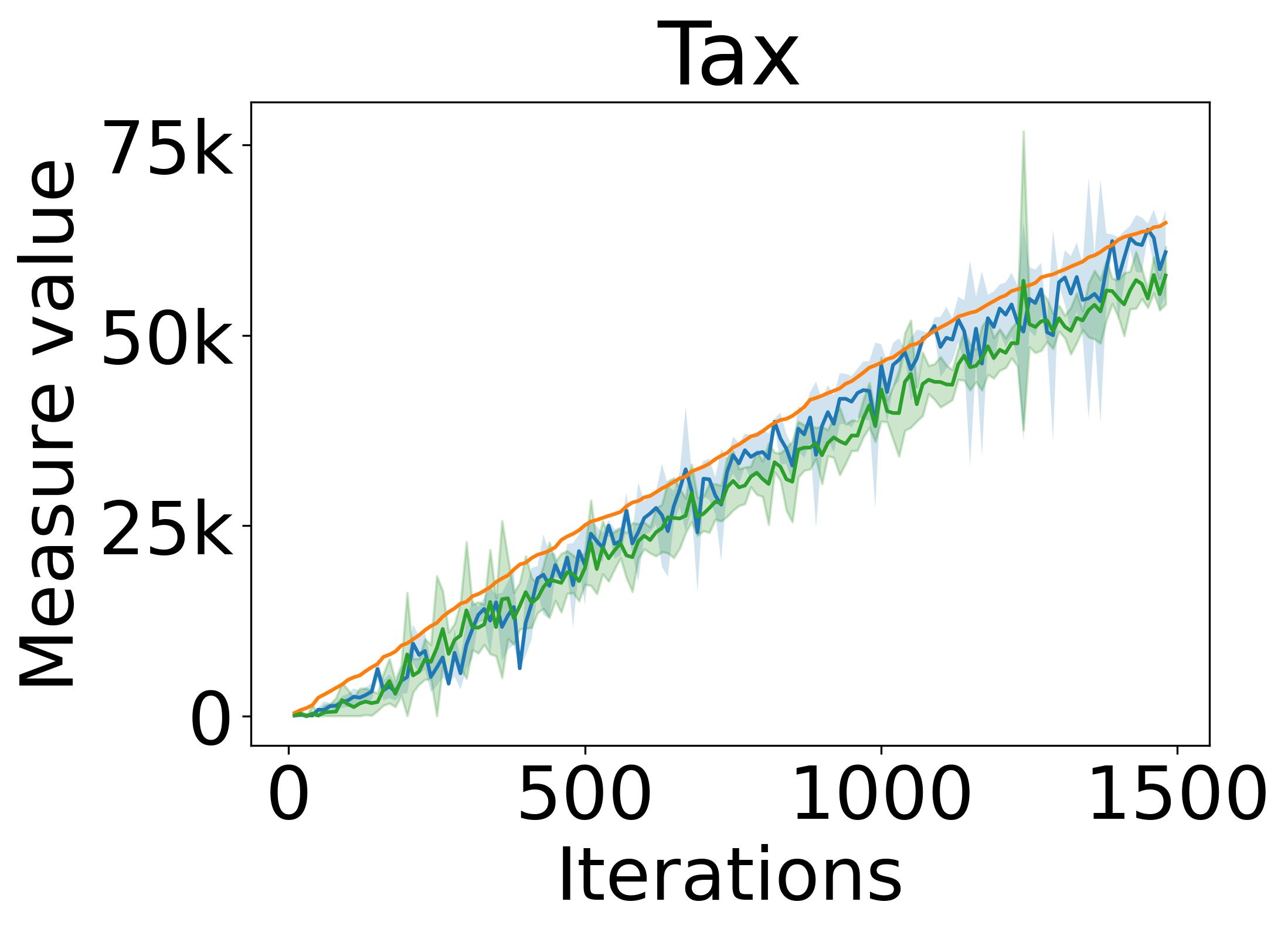}
    \hfill
    \includegraphics[width=0.19\textwidth]{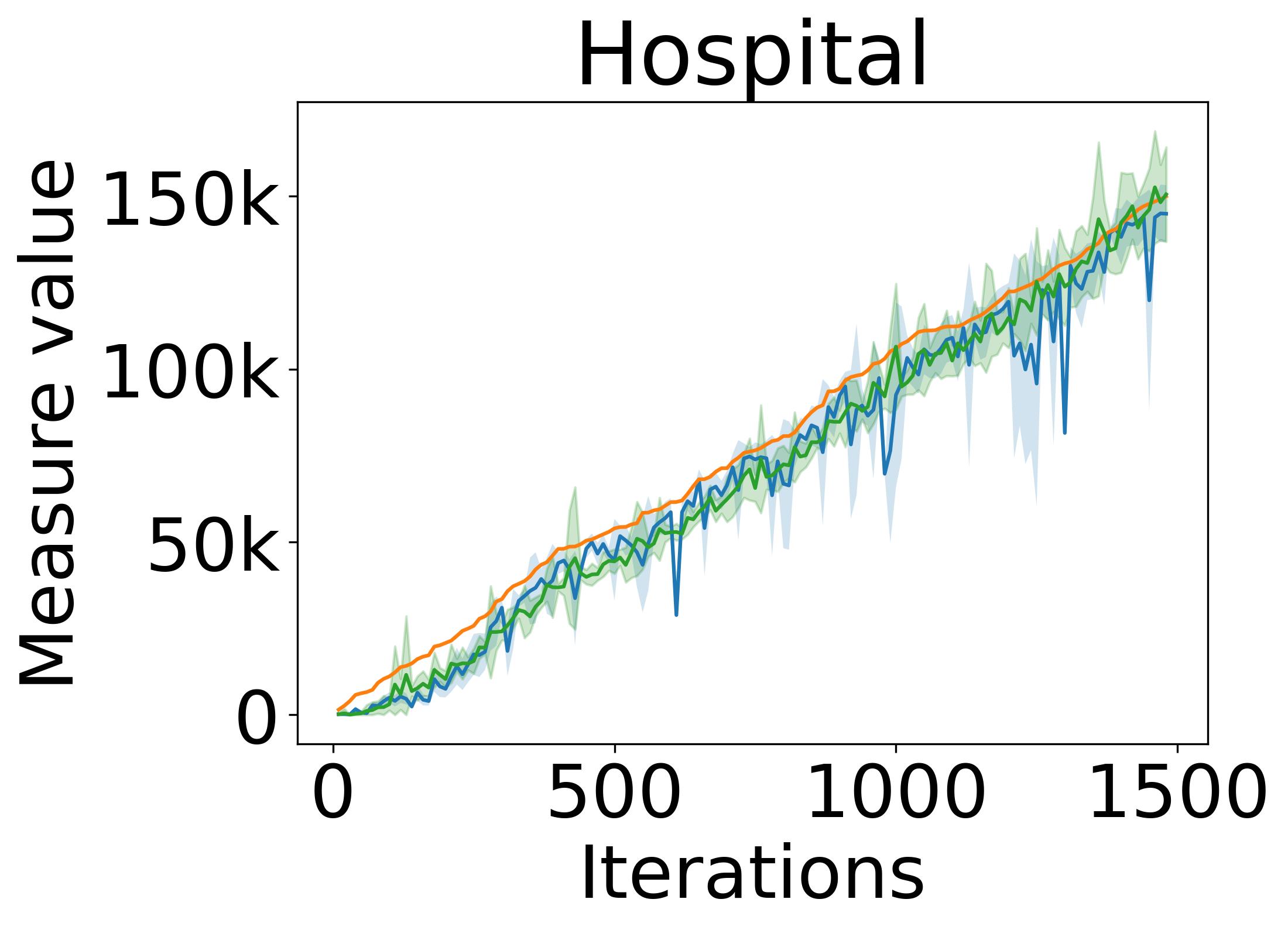}
    \hfill
    \includegraphics[width=0.19\textwidth]{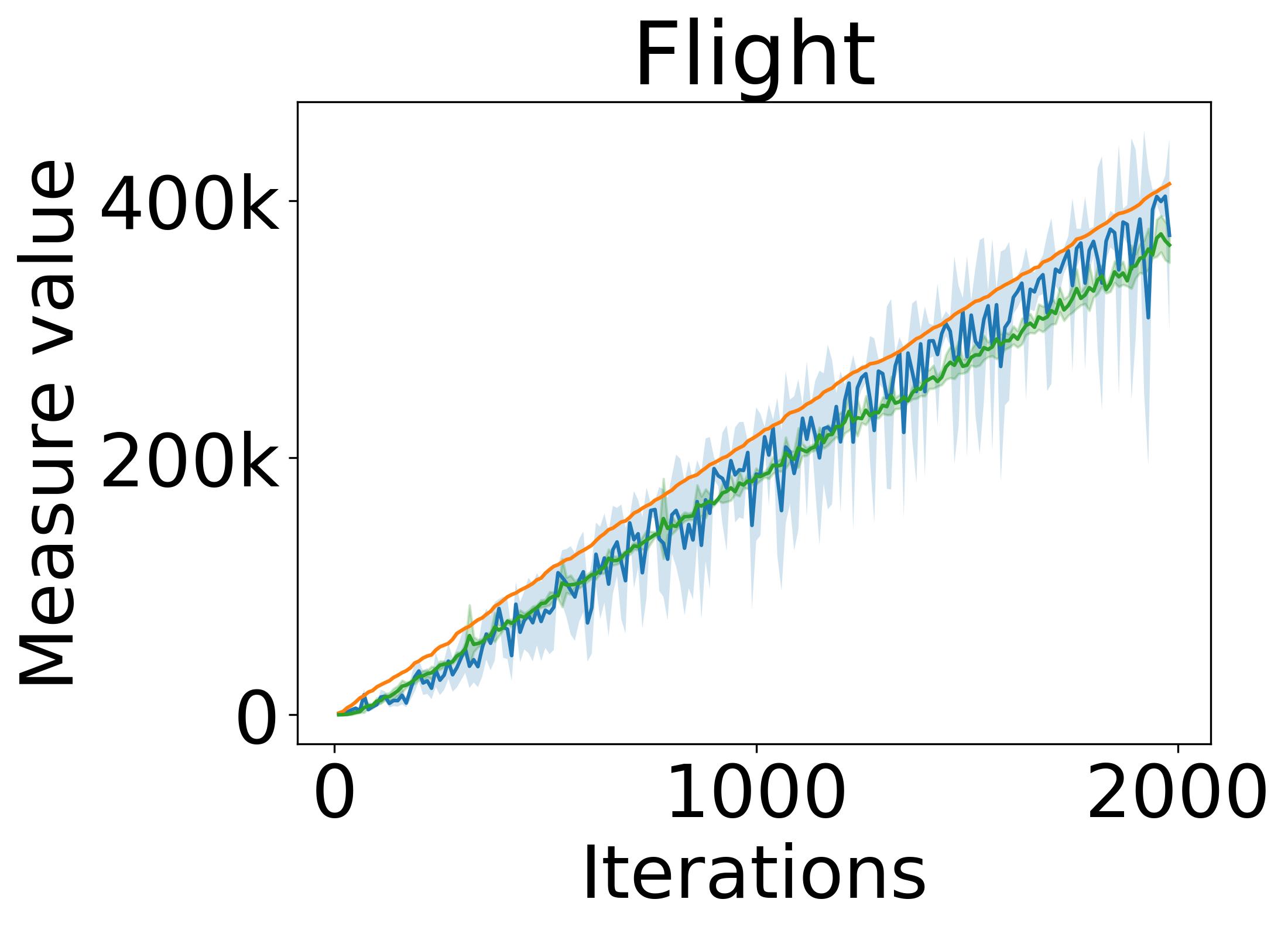}
    \hfill
    \includegraphics[width=0.19\textwidth]{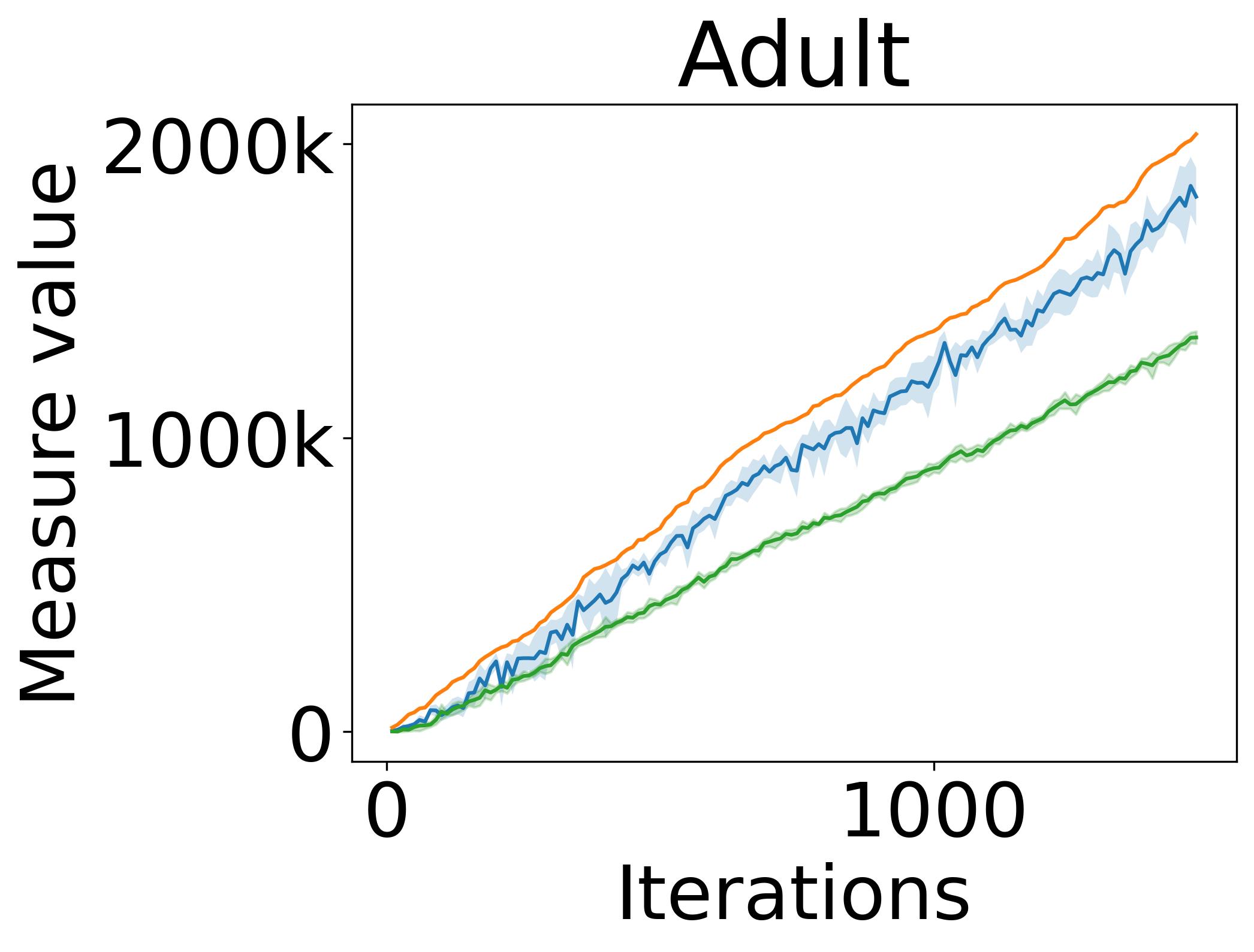}
    \includegraphics[width=0.4\textwidth]{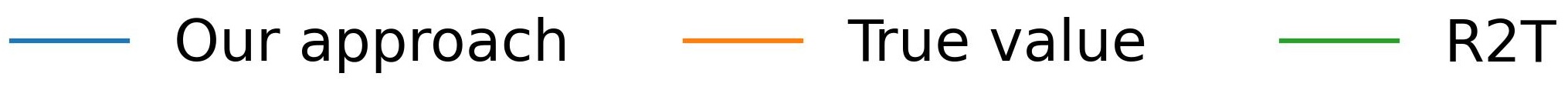}
    \caption{$\mininconsistency$ (Number of edges)}
    \label{fig:tp_rnoise_nedges}
    \end{subfigure}
      \begin{subfigure}[b]{\textwidth}
         \centering
         \includegraphics[width=0.19\textwidth]{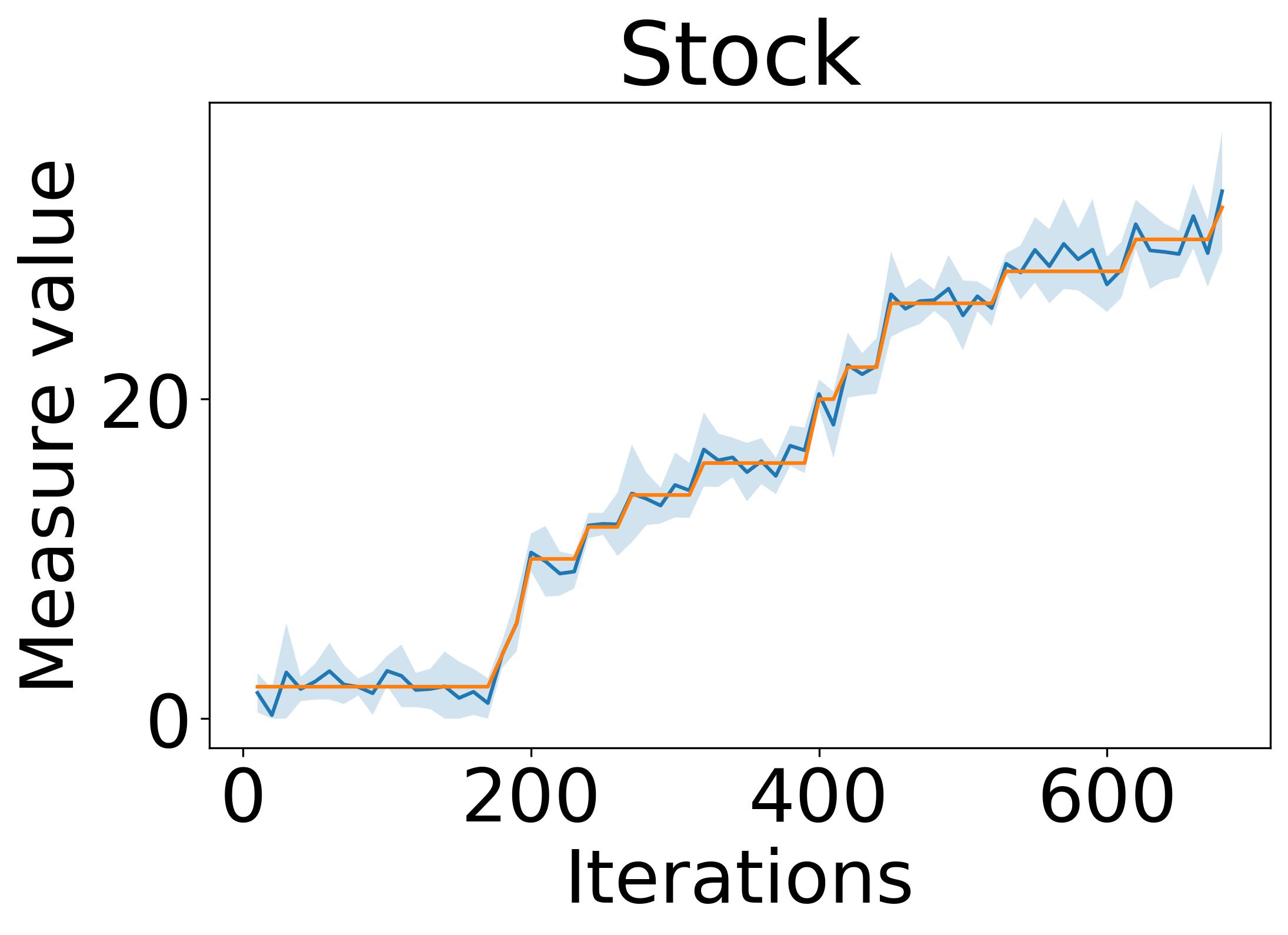}
         \hfill
         \includegraphics[width=0.19\textwidth]{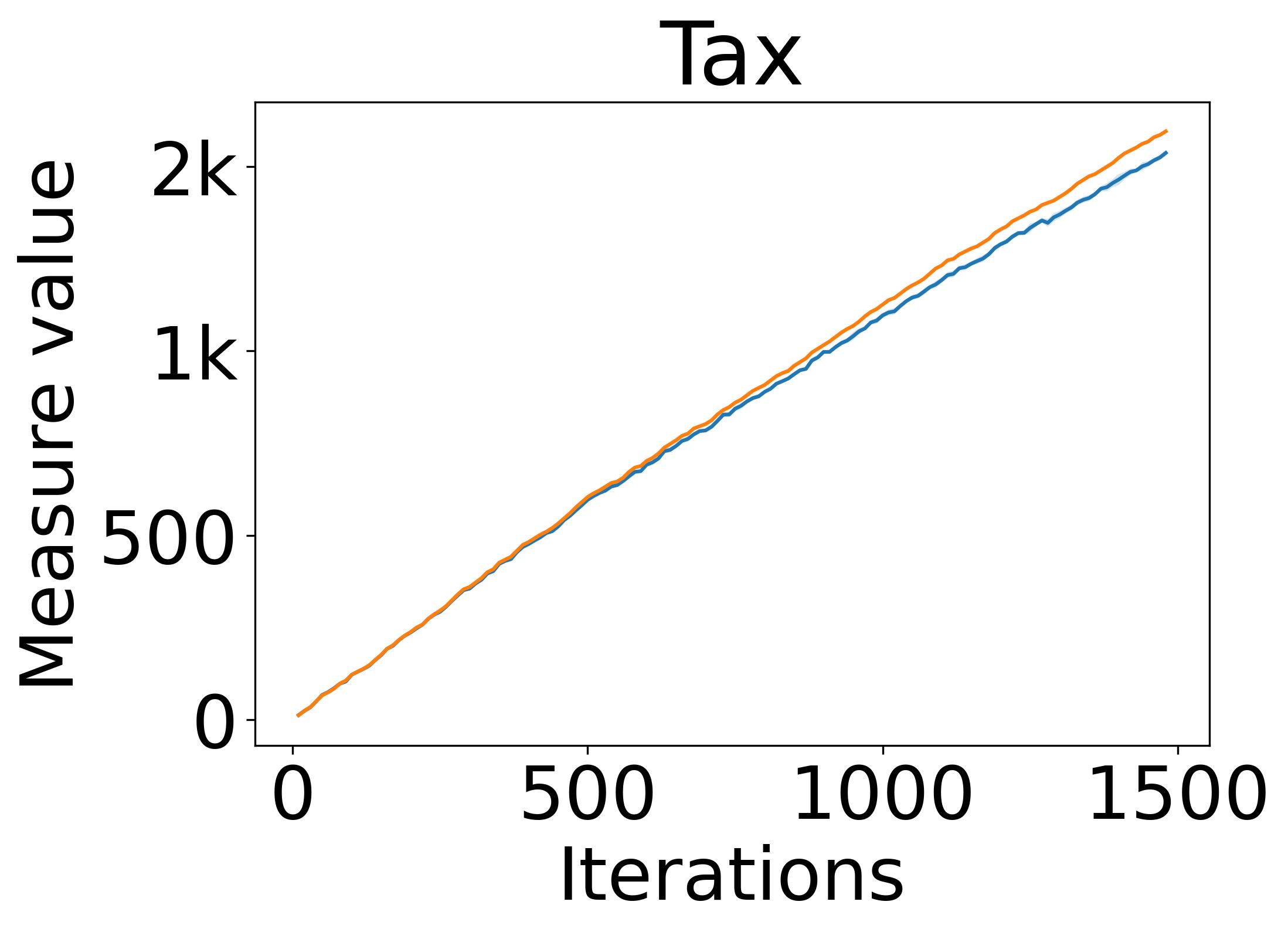}
         \hfill
         \includegraphics[width=0.19\textwidth]{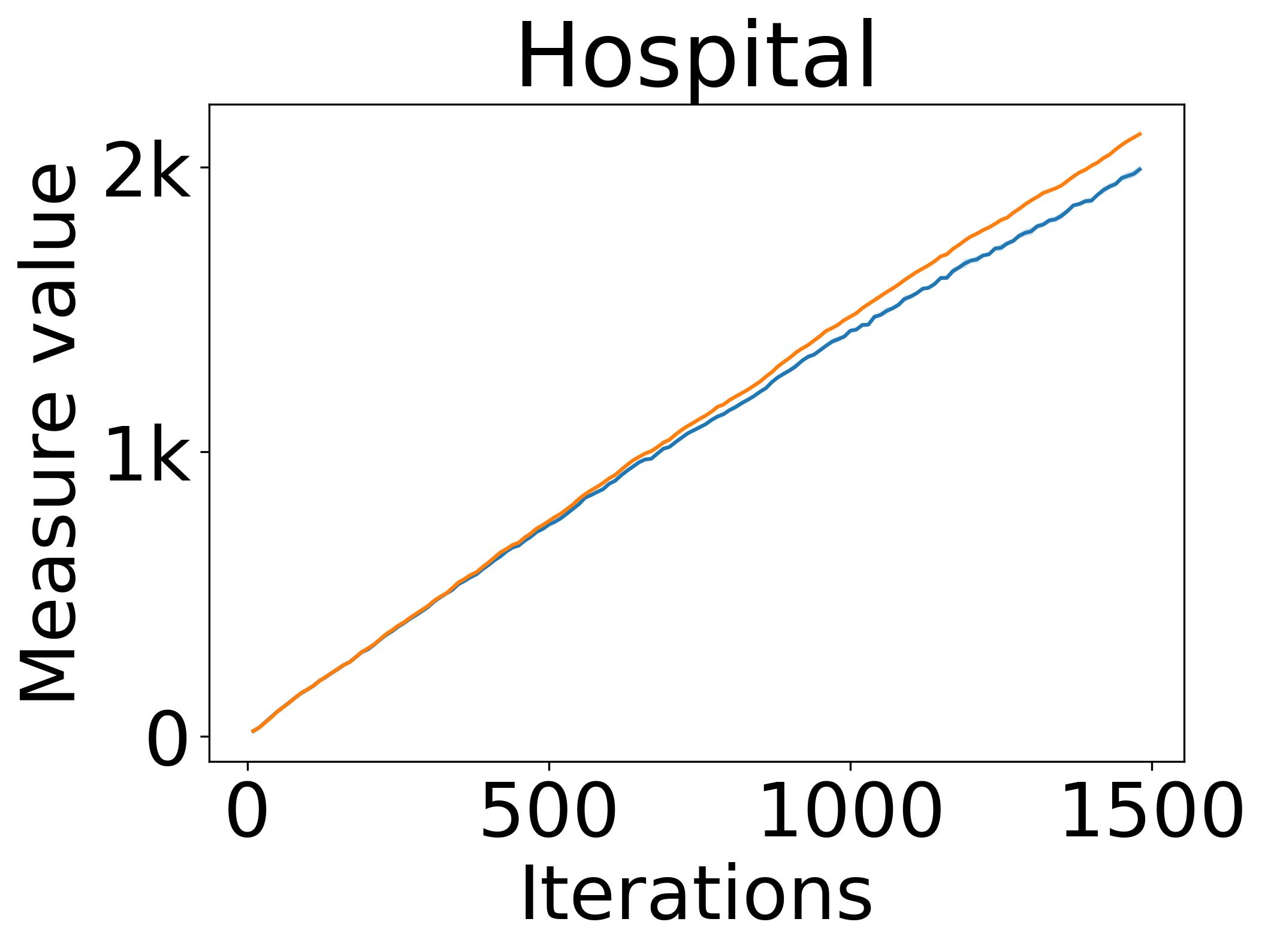}
         \hfill
         \includegraphics[width=0.19\textwidth]{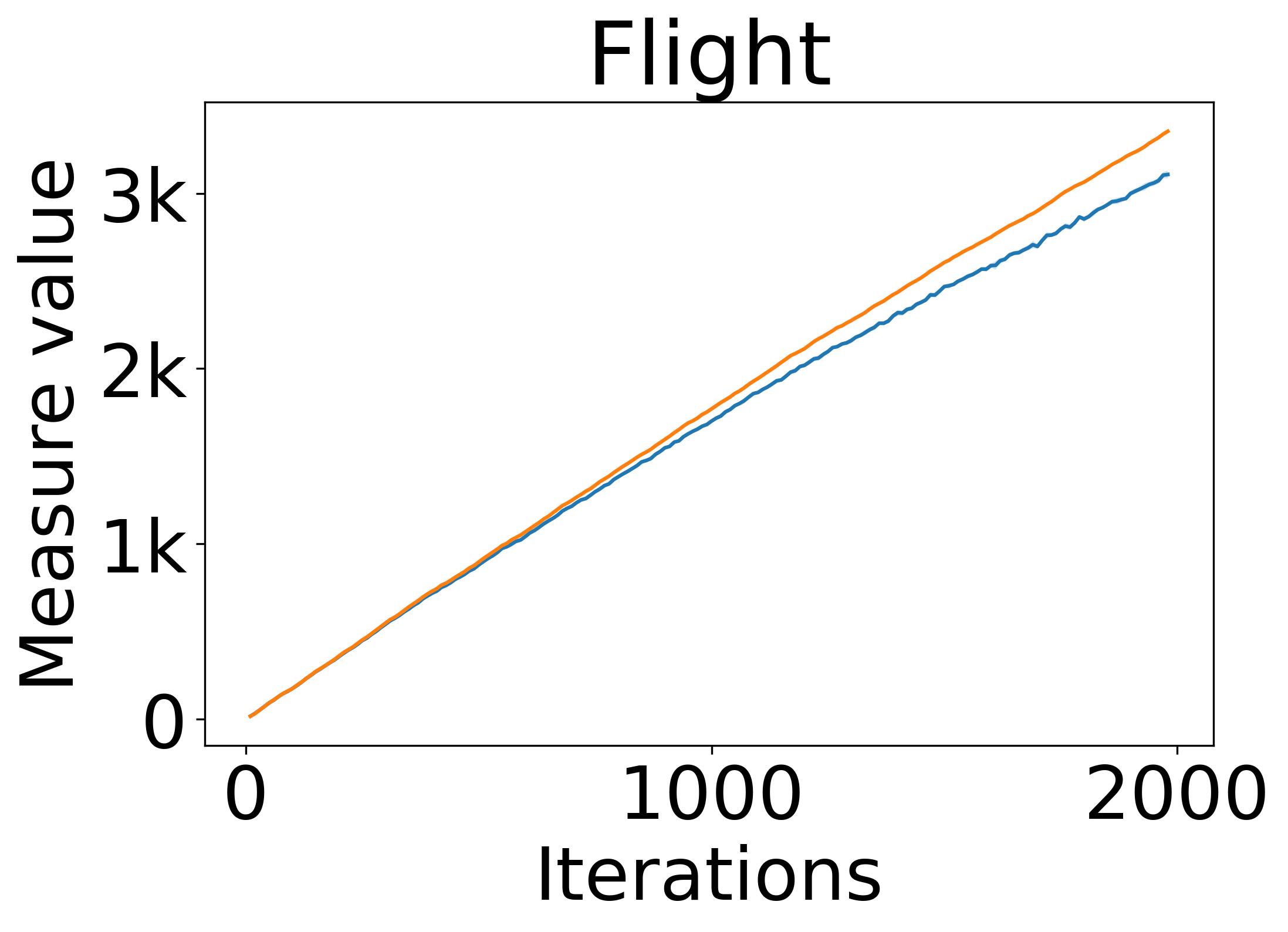}
         \hfill
         \includegraphics[width=0.19\textwidth]{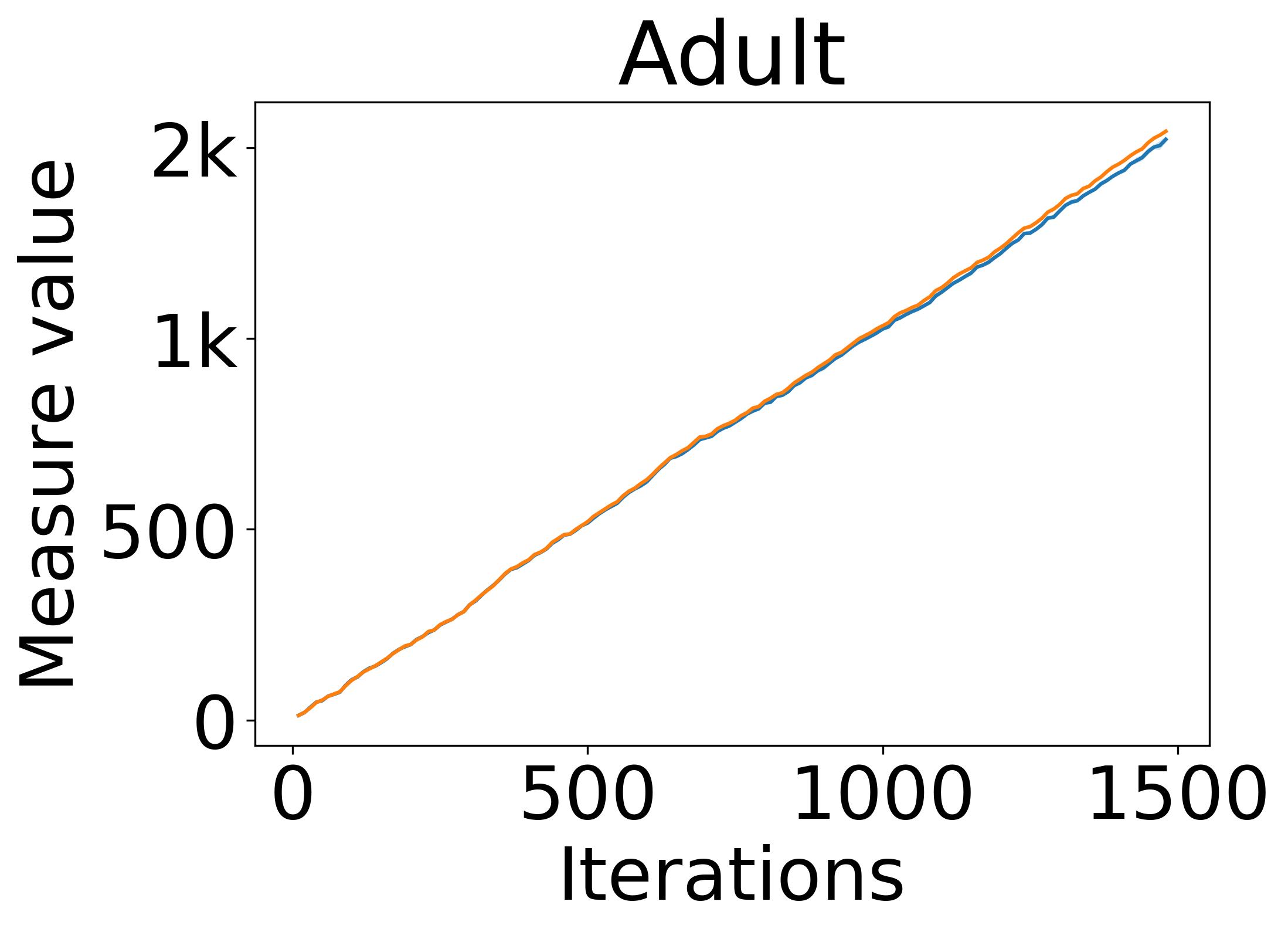}
         \includegraphics[width=0.3\textwidth]{images/legend_2.png}
         \caption{$\repair$ (Size of vertex cover)}
         \label{fig:tp_rnoise_vcover}
     \end{subfigure}
     \caption{True vs Private estimates for all dataset with RNoise $\alpha = 0.01$ at $\epsilon=1$. The $\problematic$ measure (figure a) and $\mininconsistency$ measure (figure b) are computed using our graph projection approach, and $\repair$ measure using our private vertex cover size approach. }
     \label{fig:tp_RNoise}
\end{figure*}

\subsection{Experimental Setup}

All our experiments are performed on a server with Intel Xeon Platinum 8358 CPUs (2.60GHz) and 1 TB RAM. Our code is in Python 3.11 and can be found in the artifact submission.
All experiments are repeated for $10$ runs, and the mean error value is reported. 

\paratitle{Datasets and violation generation}
We replicate the exact setup as \citet{LivshitsBKS20} for experimentation. We conduct experiments on five real-life datasets and their corresponding DCs as described in Table~\ref{tab:datasets}.
\ifpaper
\else
\begin{itemize}
    \item Adult~\cite{misc_adult_2}: Annual income results from various factors.
    \item Flight~\cite{flight}: Flight information across the US.
    \item Hospital~\cite{hospital}: Information about different hospitals across the US and their services.
    \item Stock~\cite{oleh_onyshchak_2020}: Trading stock information  on various dates
    \item Tax~\cite{chu2013discovering}: Personal tax infomation.
\end{itemize}
\fi

These datasets are initially consistent with the constraints.  All experiments are done on a subset of $10k$ rows, and violations are added similarly using both their proposed algorithms, namely CONoise (for Constraint-Oriented Noise) and RNoise (for Random Noise). 
\ifpaper
\else
CONoise introduces random violations of the constraints by running 200 iterations of the following procedure:
\begin{enumerate}
    \item Randomly select a constraint $\sigma$ from the constraint set $\constraintset$.
    \item Randomly select two tuples $t_i$ and $t_j$ from the database.
    \item For every predicate $\phi = (t_i[a_1] \circ t_j[a_2])$ of $\sigma$:
    \begin{itemize}
        \item If $t_i$ and $t_j$ jointly satisfy $\phi$, continue to the next predicate.
        \item If $\circ \in \{=, \leq, \geq\}$, change either $t_i[a_1]$ or $t_j[a_2]$ or vice versa (the choice is random).
        \item If $\circ \in \{<, >, \neq\}$, change either $t_i[a_1]$ or $t_2[a_2]$ (the choice is again random) to another value from the active domain of the attribute such that $\phi$ is satisfied, if such a value exists, or a random value in the appropriate range otherwise.
    \end{itemize}
\end{enumerate}
The second algorithm, RNoise, is parameterized by the parameter $\alpha$ that controls the noise level by modifying $\alpha$ of the values in the dataset. At each iteration of RNoise, we randomly select a database cell corresponding to an attribute that occurs in at least one constraint. Then, we change its value to another value from the active domain of the corresponding attribute (with probability 0.5) or a typo. 
\fi
The datasets vary immensely in the density of their conflict graphs as described in the max degree column of Table~\ref{tab:datasets}. For example, the Adult $10k$ nodes subset has a maximum degree of 9635, whereas the Stock dataset has a maximum of 1 with the same amount of conflict addition. 

\paratitle{Metrics} Following Livshits et al.~\cite{LivshitsBKS20}, we randomly select a subset of 10k rows of each dataset, add violations to the subset, and compute the inconsistency measures on the dataset with violations.   To measure performance, we utilize the normalized $\ell_1$ distance error~\cite{dwork2006calibrating}, $|\mathcal{I}(D,\Sigma)-\mathcal{M}(D,\Sigma,\epsilon)| /\mathcal{I}(D,\Sigma)$, where $\mathcal{M}(D,\Sigma,\epsilon)$ represents the estimated private value of the measure and $\mathcal{I}(D,\Sigma)$ denotes the true value. For $\repair$, we use the linear approximation algorithm from Livshits et al.~\cite{LivshitsBKS20} to estimate the non-private value.

\begin{figure*}
    \begin{subfigure}[b]{\textwidth}
         \centering
         \includegraphics[width=0.19\textwidth]{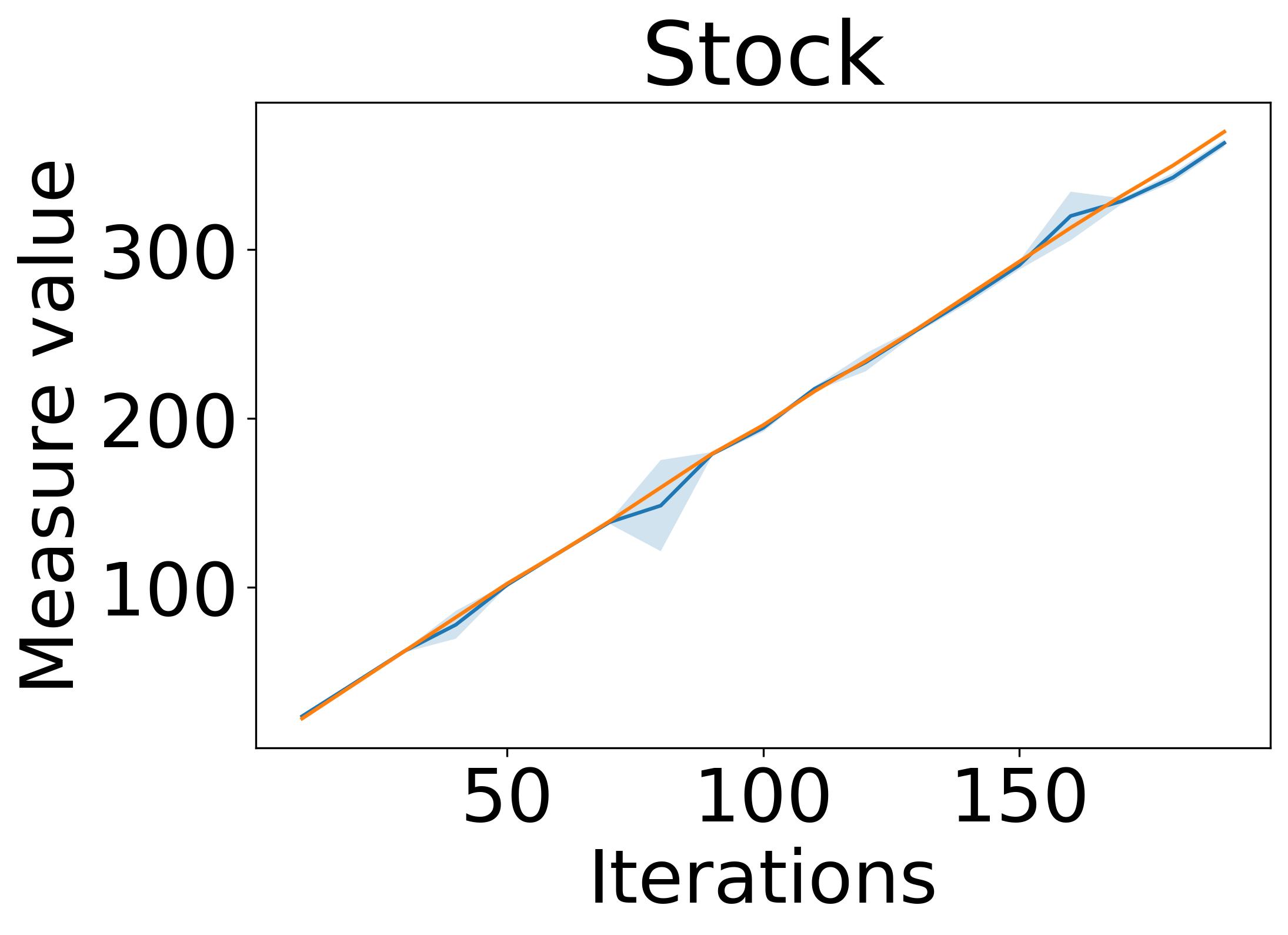}
         \hfill
         \includegraphics[width=0.19\textwidth]{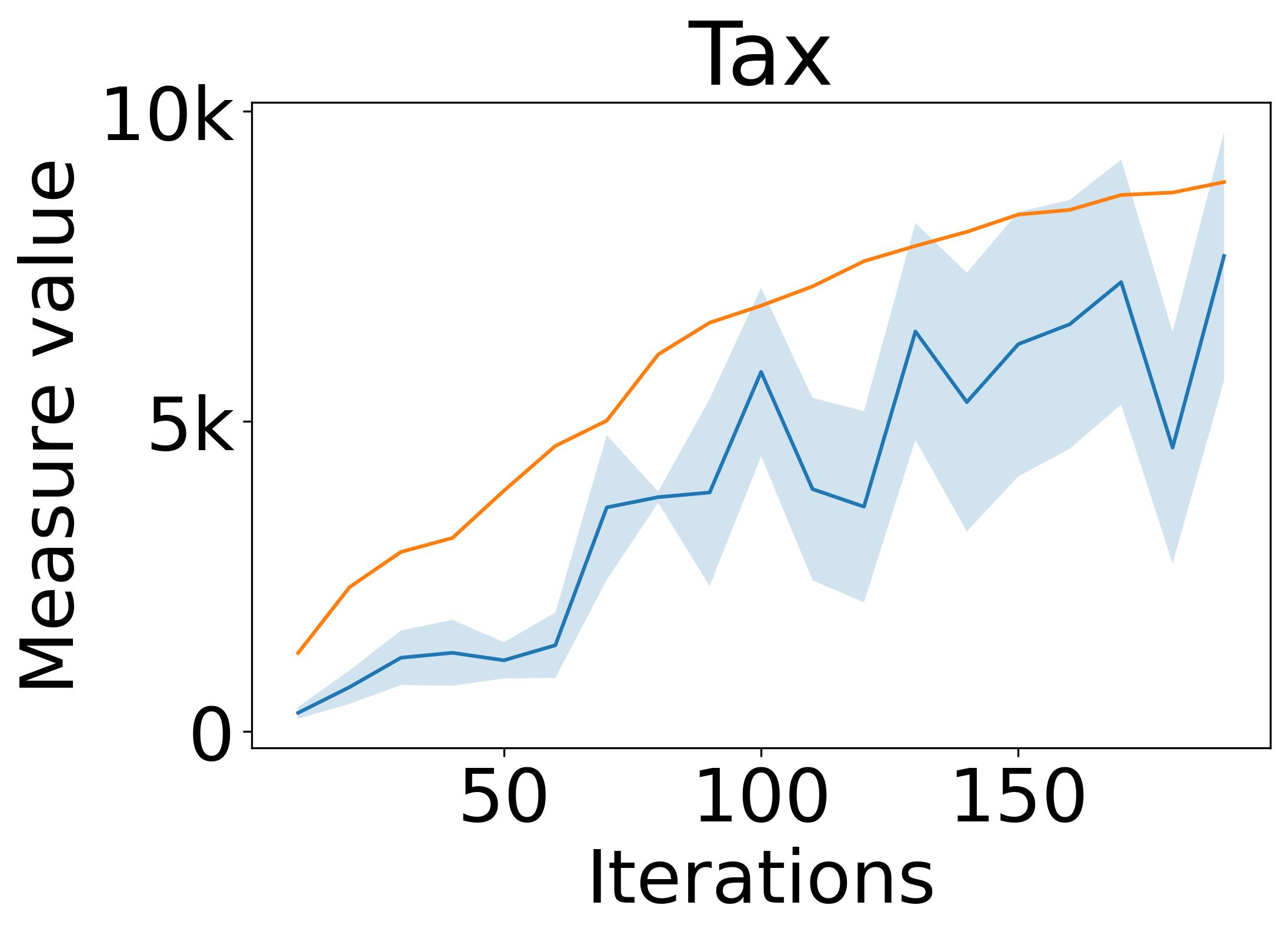}
         \hfill
         \includegraphics[width=0.19\textwidth]{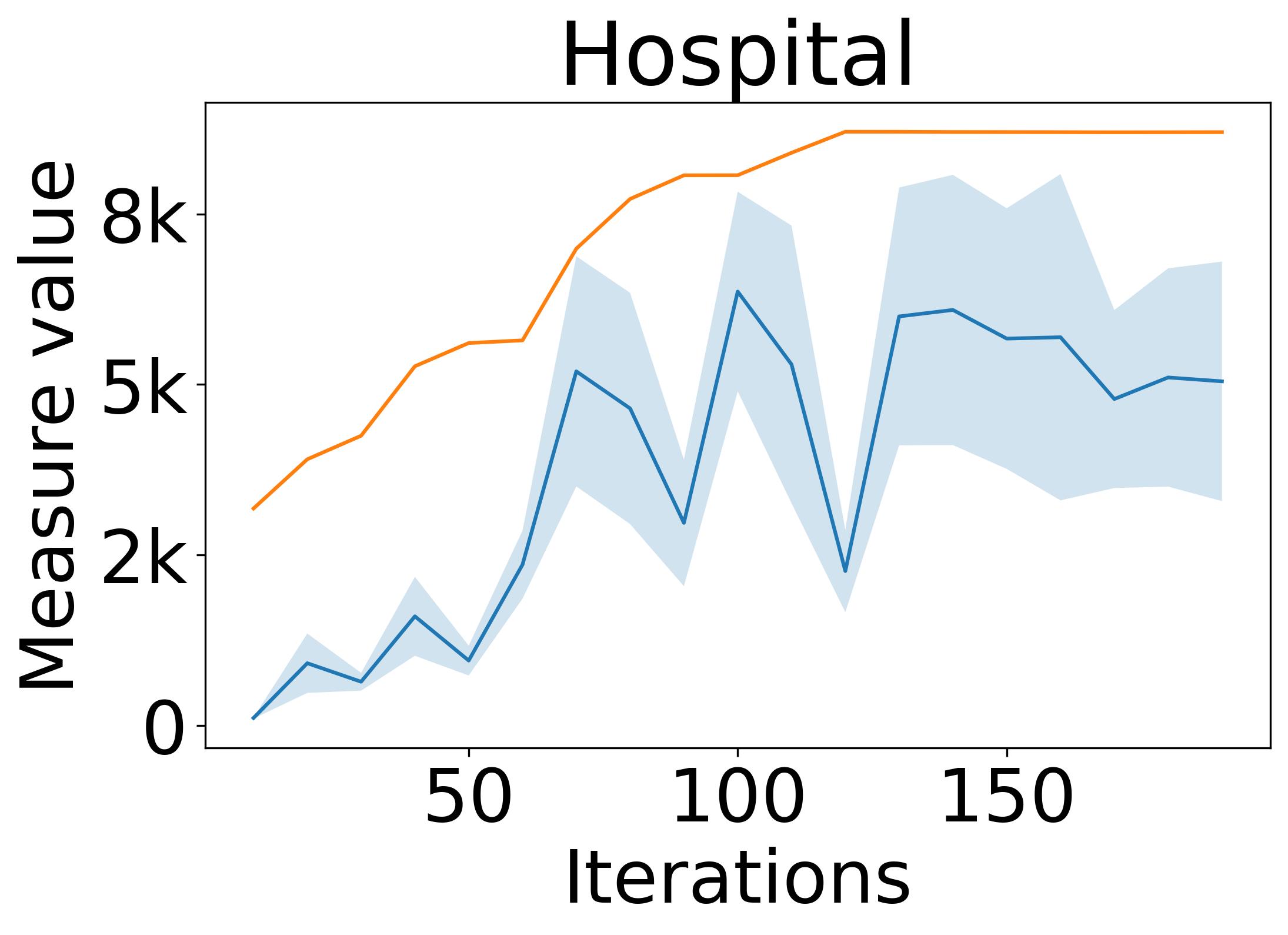}
         \hfill
         \includegraphics[width=0.19\textwidth]{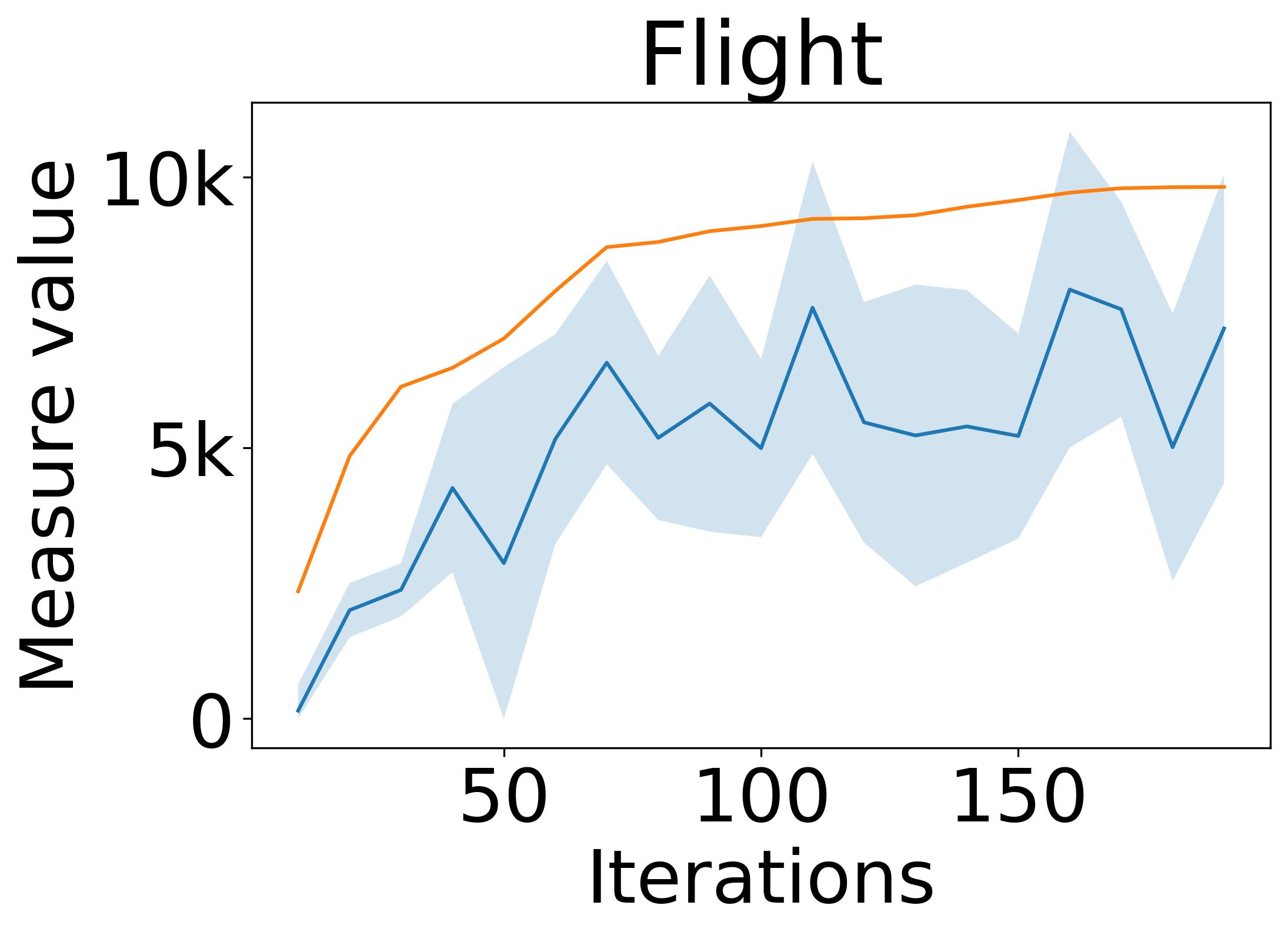}
         \hfill
         \includegraphics[width=0.19\textwidth]{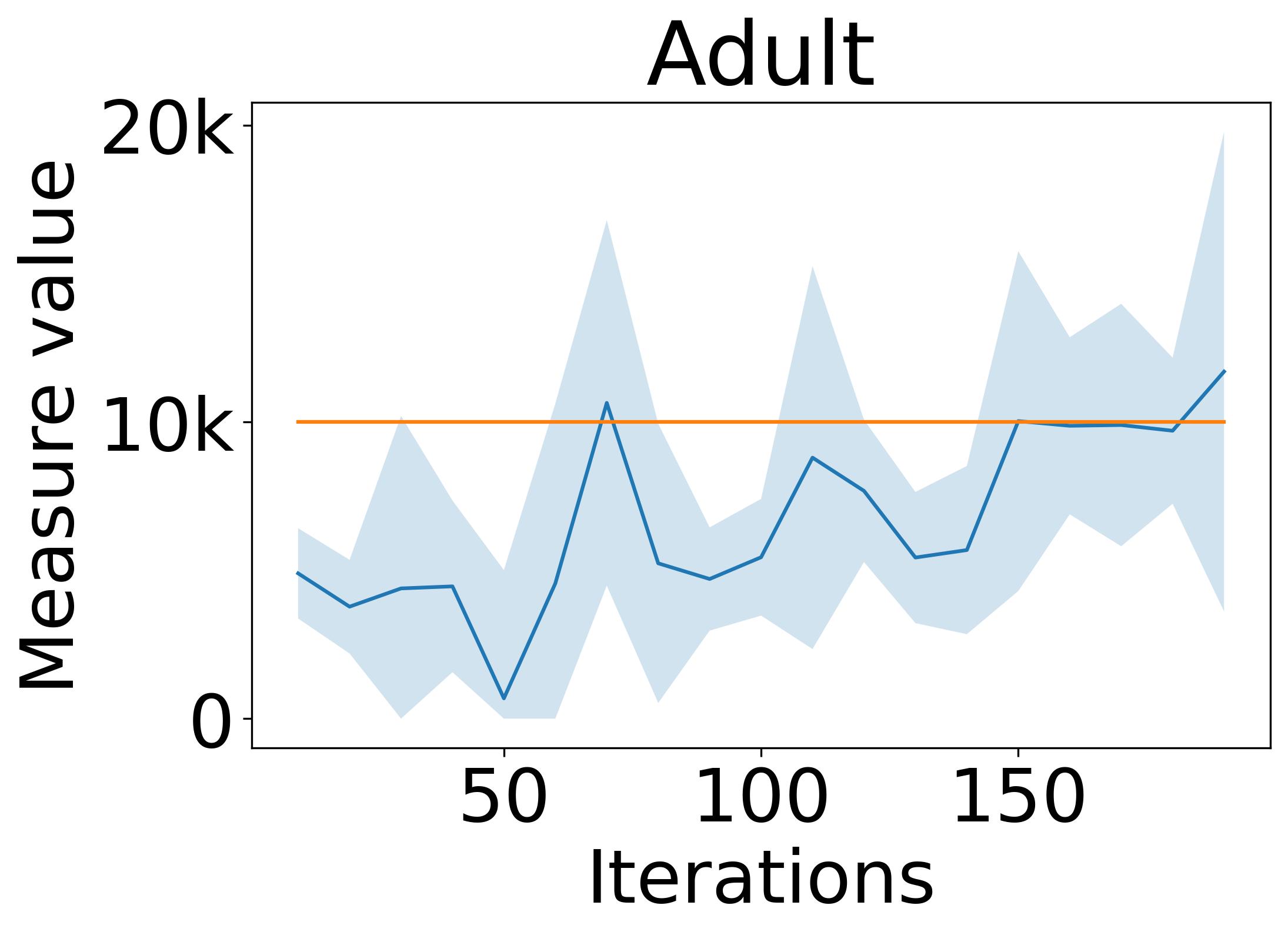}
         \includegraphics[width=0.3\textwidth]{images/legend_2.png}
         \caption{$\problematic$ (Positive degree nodes)}
         \label{fig:tp_conoise_pdedges}
     \end{subfigure}
    \begin{subfigure}[b]{\textwidth}
         \centering
         \includegraphics[width=0.19\textwidth]{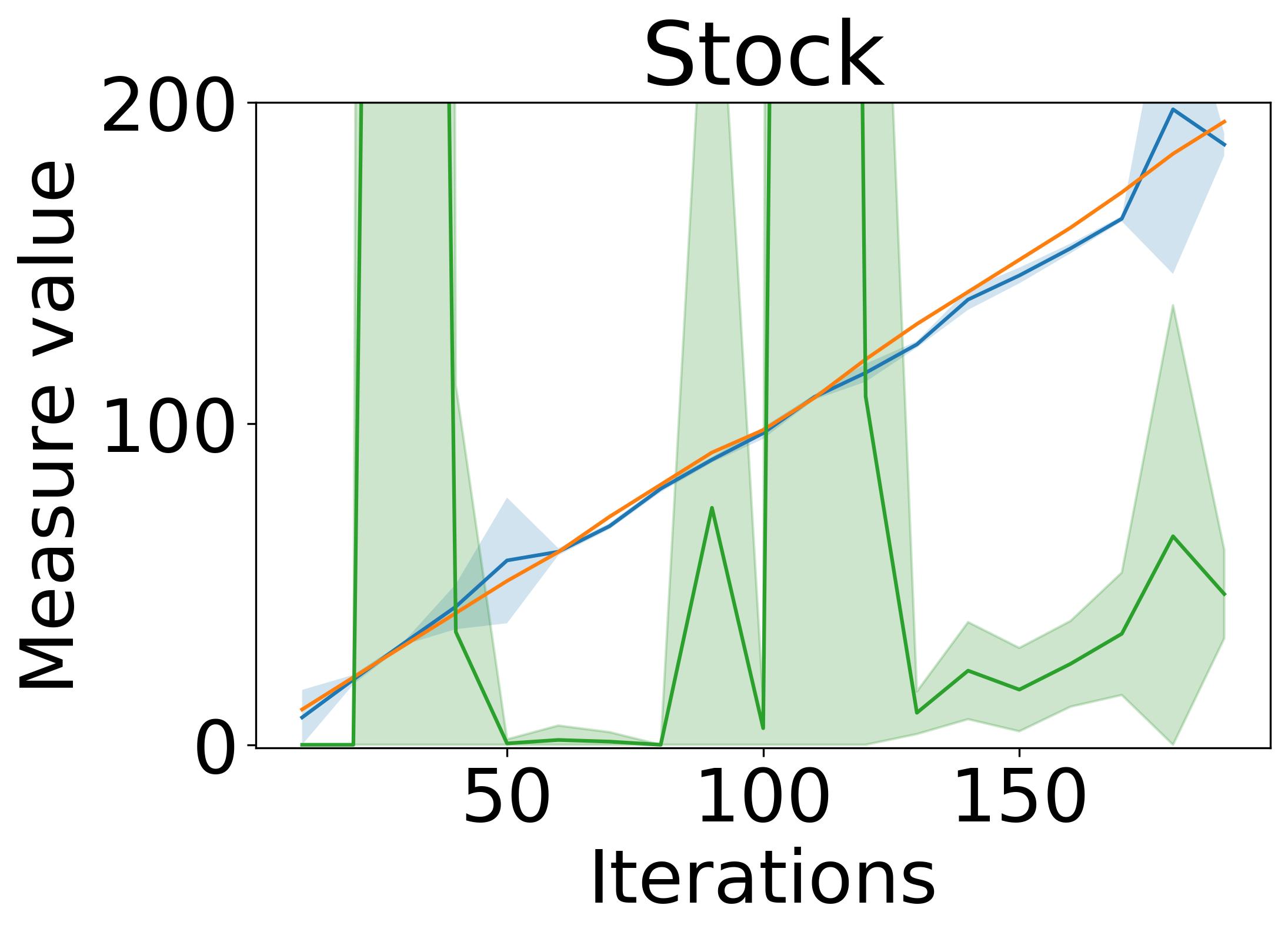}
         \hfill
         \includegraphics[width=0.19\textwidth]{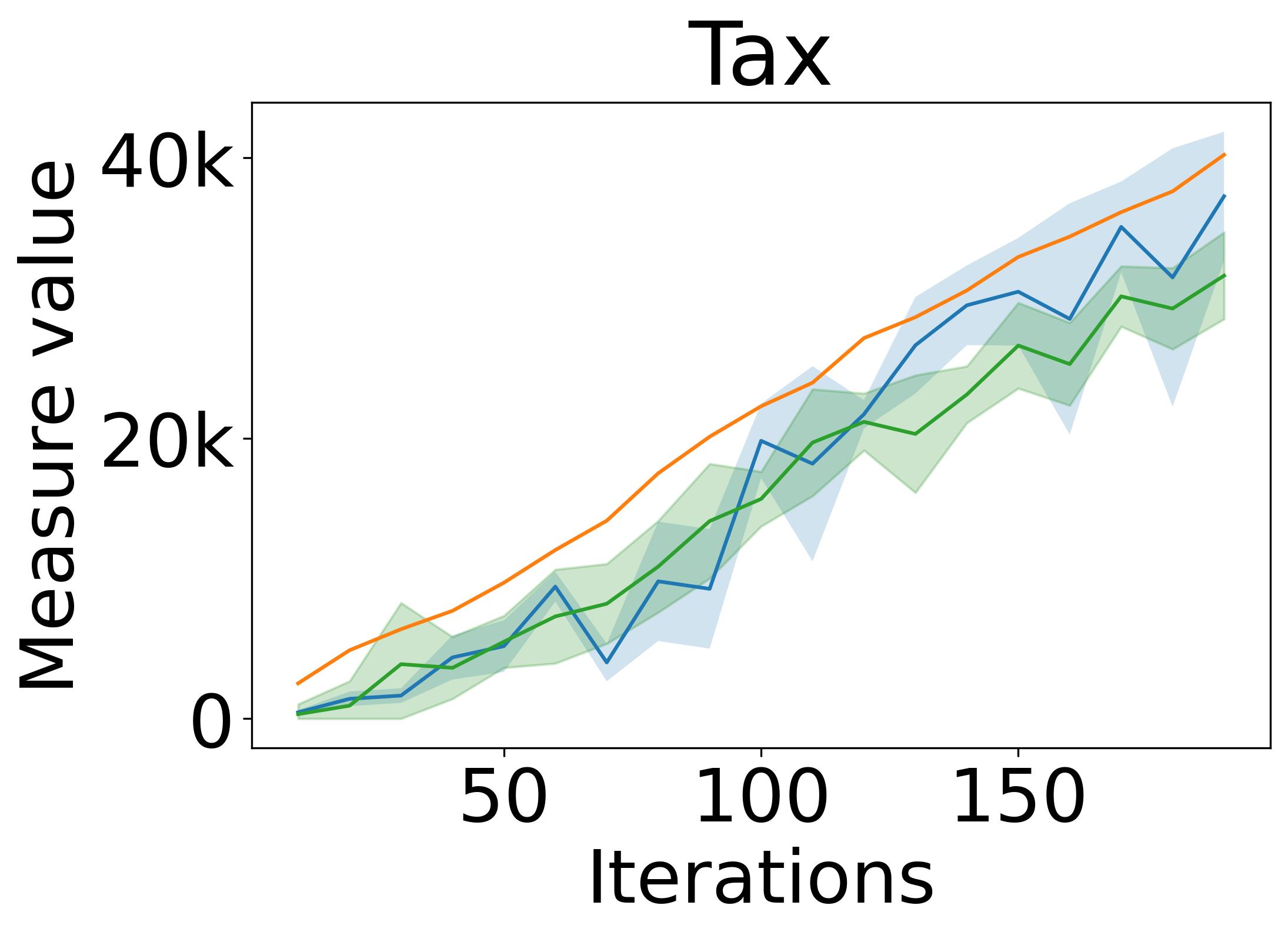}
         \hfill
         \includegraphics[width=0.19\textwidth]{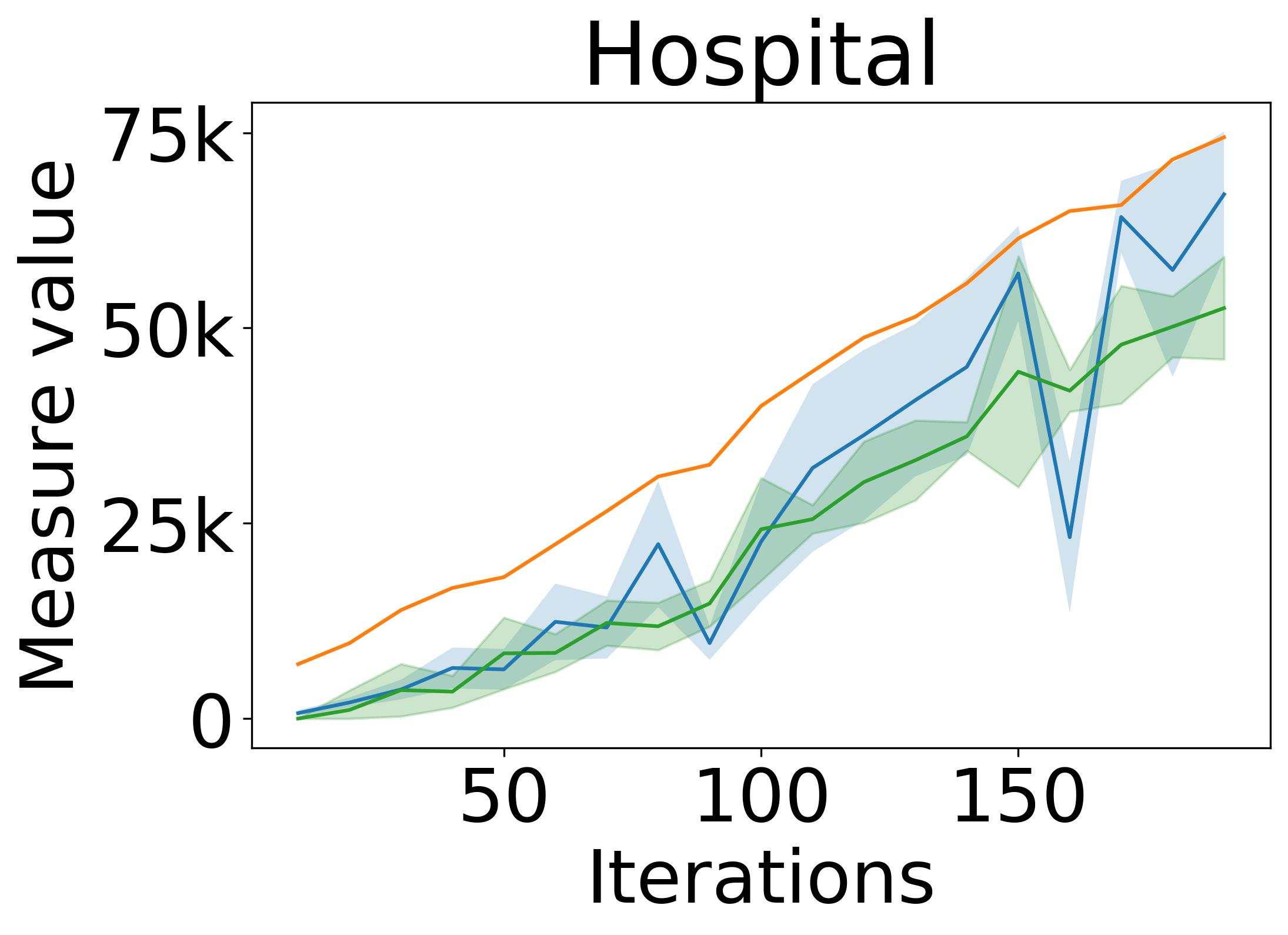}
         \hfill
         \includegraphics[width=0.19\textwidth]{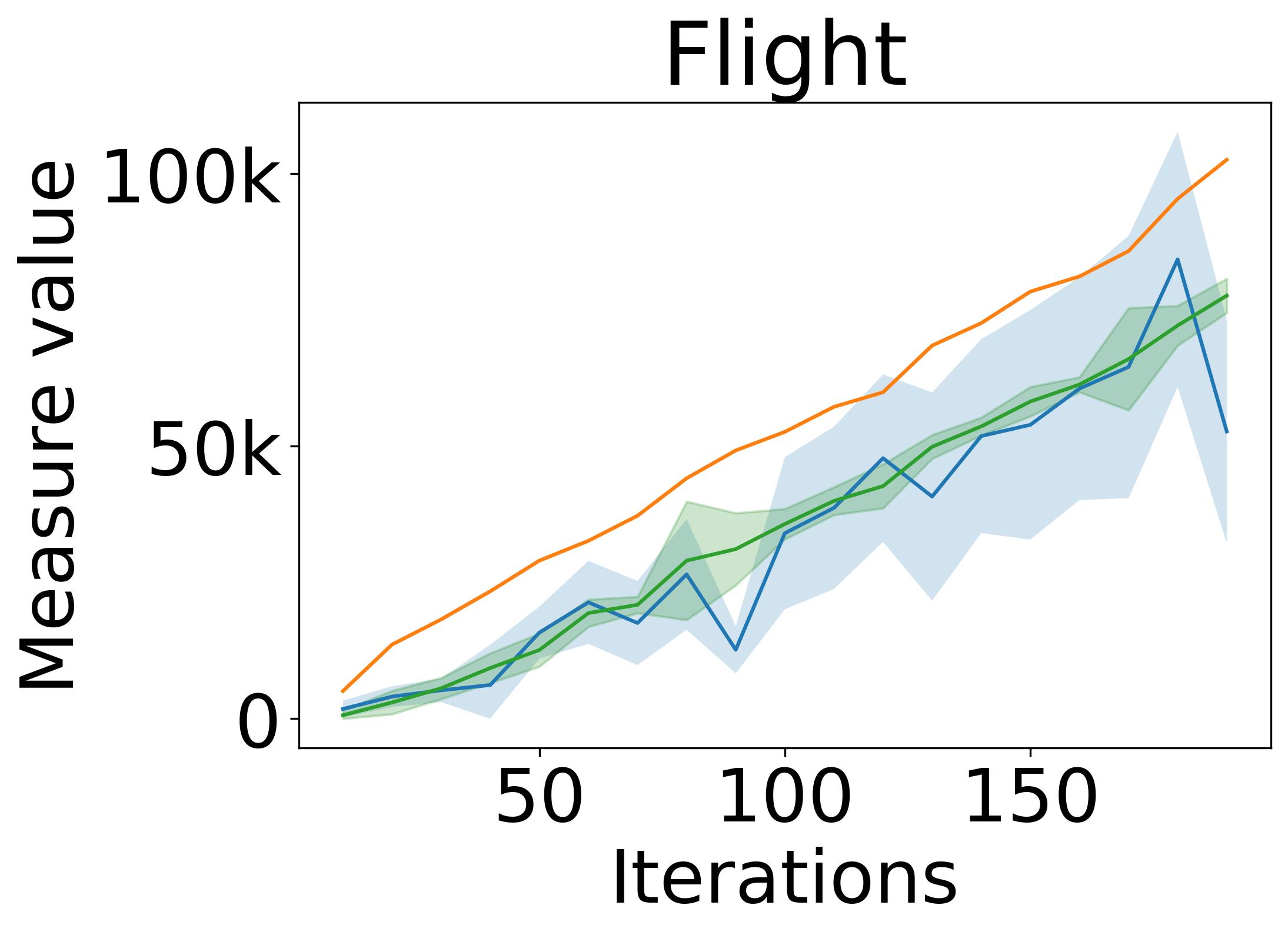}
         \hfill
         \includegraphics[width=0.19\textwidth]{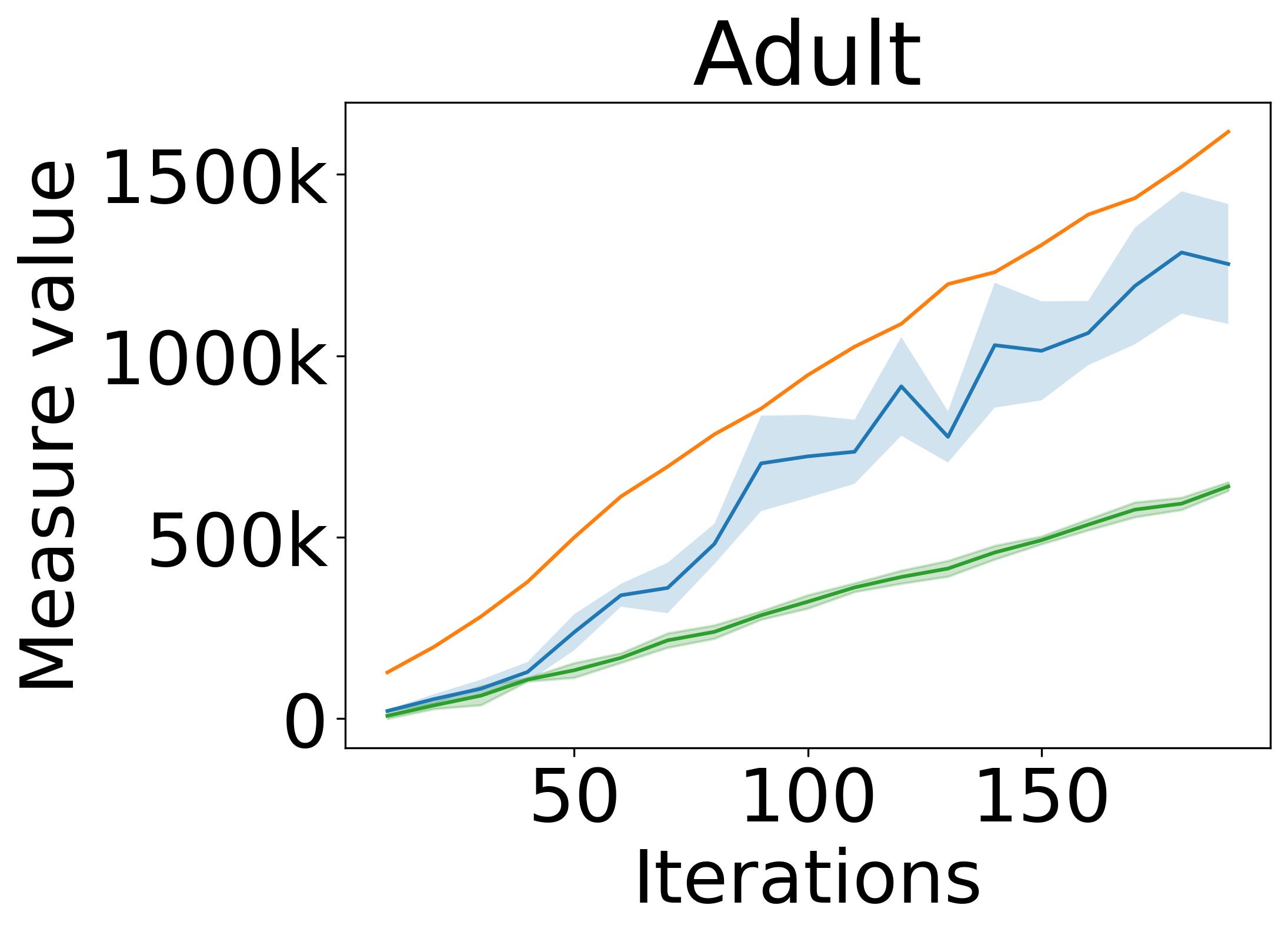}
         \includegraphics[width=0.4\textwidth]{images/legend2_r2t.png}
         \caption{$\mininconsistency$ (Number of edges)}
         \label{fig:tp_conoise_nedges}
     \end{subfigure}
     \begin{subfigure}[b]{\textwidth}
    \centering
    \includegraphics[width=0.19\textwidth]{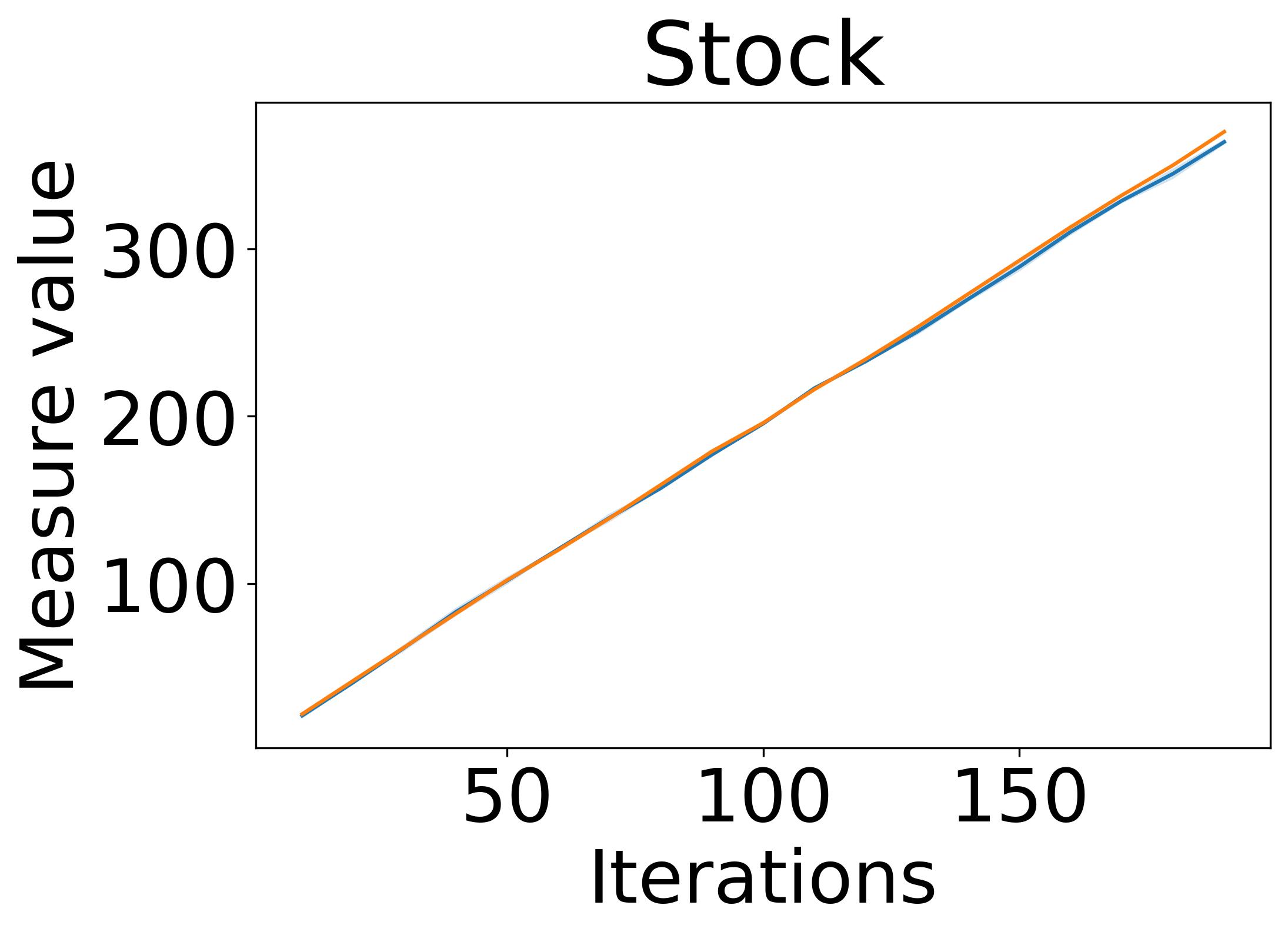}
    \hfill
    \includegraphics[width=0.19\textwidth]{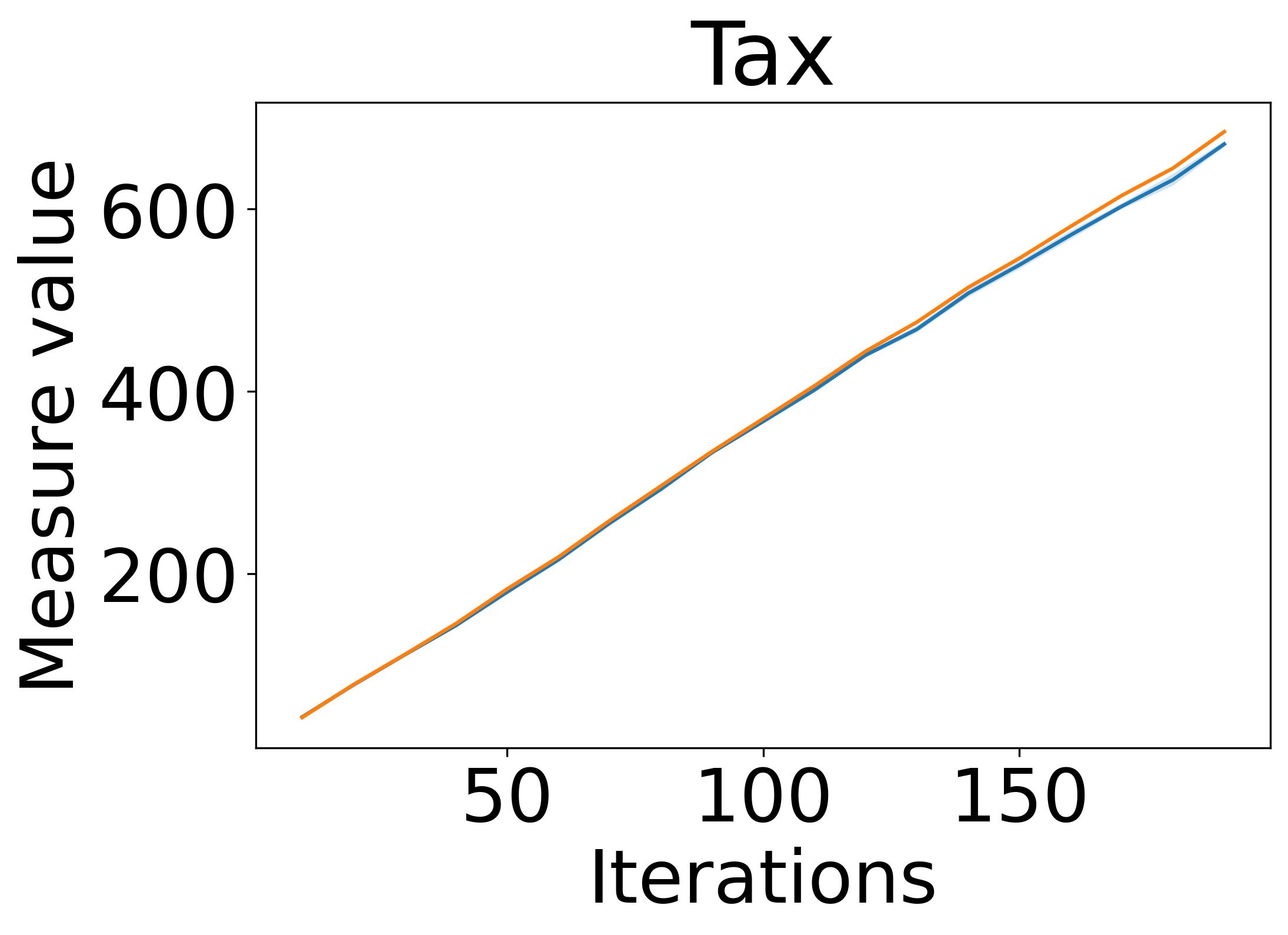}
    \hfill
    \includegraphics[width=0.19\textwidth]{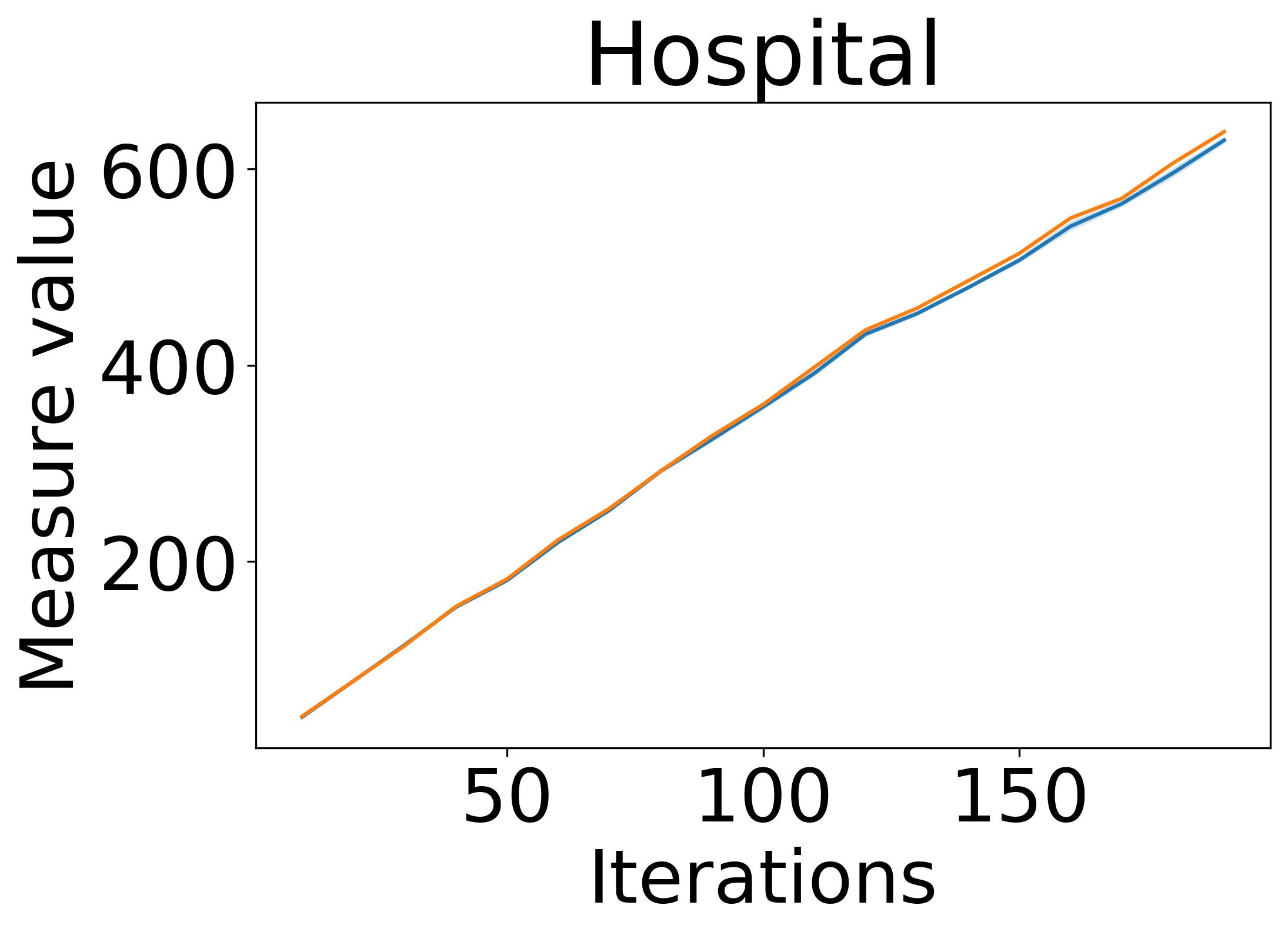}
    \hfill
    \includegraphics[width=0.19\textwidth]{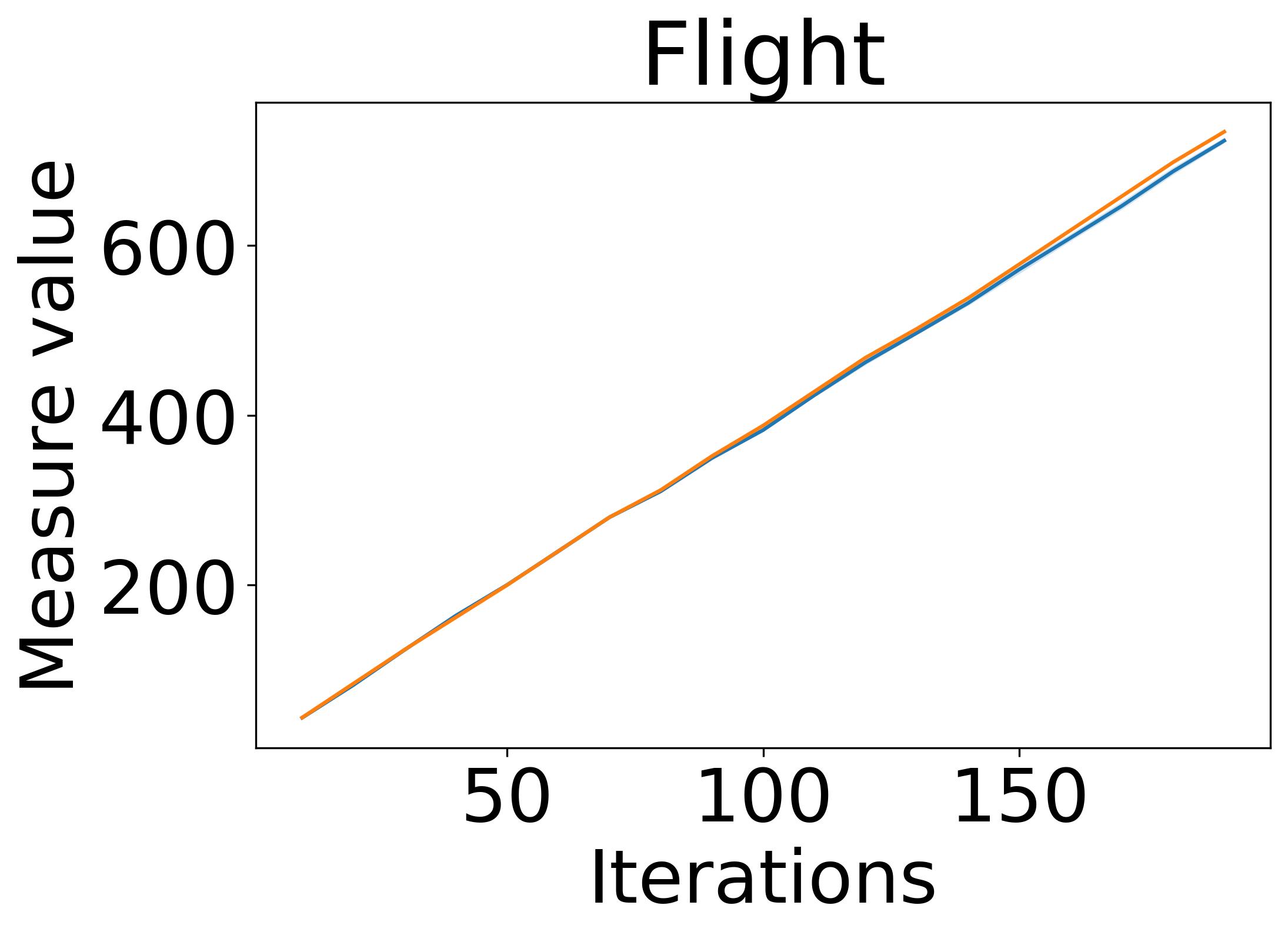}
    \hfill
    \includegraphics[width=0.19\textwidth]{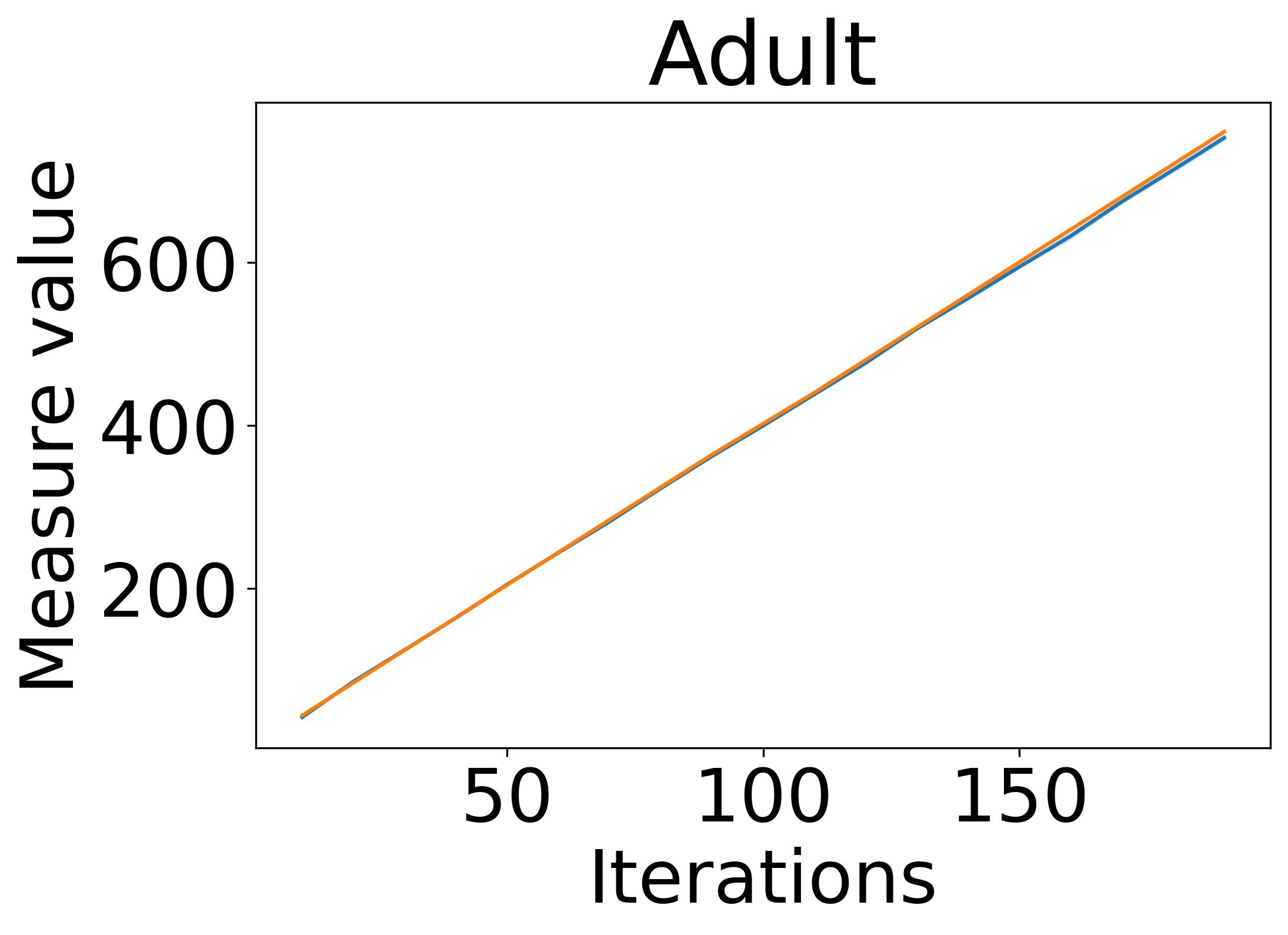}
    \includegraphics[width=0.3\textwidth]{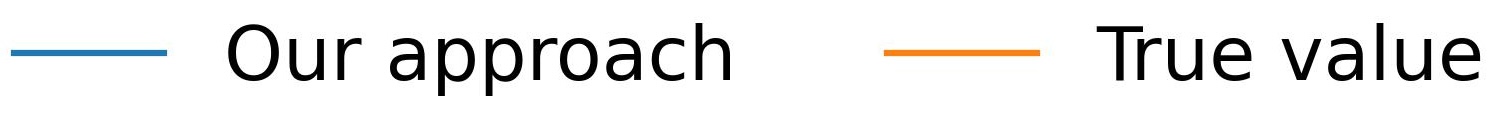}
    \caption{$\repair$ (Size of vertex cover)}
    \label{fig:tp_conoise_vcover}
\end{subfigure}
     \caption{True vs.~private estimates for all dataset with CONoise at $\epsilon = 1$ for 200 iterations. The $\problematic$ measure (a) and $\mininconsistency$ measure (b) are computed using our graph projection approach, and the $\repair$ measure (c) using our private vertex cover size approach.}
     \label{fig:tp_conoise}
\end{figure*}

\paratitle{Algorithm variations}  We experiment with multiple different variations of Algorithm~\ref{algo:graph_general} for $\mininconsistency$ and $\problematic$. The initial candidate set for the degree bound is $\Theta = [1, 5, 10, 100, 500, 1000, 2000, 3000, \dots, 10000]$ with multiples of $1000$ along with some small candidates.
\squishlist
    \item \textit{Baseline 1}: naively sets the bound $\theta^*$ to the maximum possible degree $|V|$ in Algorithm~\ref{algo:graph_general} by skipping line 2 and the unused privacy budget $\epsilon_1$ is used for the final noise addition step. 

    \item \textit{Baseline 2}: sets the bound $\theta^*$ to the actual maximum degree of the conflict graph  $\degree_{\max}(\graph)$. Note that this is a non-private baseline that only acts as an upper bound and is one of the best values that can be achieved without privacy constraints. 
    \item \textit{Exponential mechanism}: choose $\theta^*$ over the complete candidate set $\Theta$ using the basic EM in  Algorithm~\ref{algo:expo_mech_basic}.
    \item \textit{Hierarchical exponential mechanism}: chooses $\theta^*$ using a two-step EM with an equal budget for each step in Algorithm~\ref{algo:em_opt}, but skipping Lines 1-5 of the upper bound computation step.  
    \item \textit{Upper bound + hierarchical exponential mechanism (our approach)}: encompasses both the optimization strategies, including 
    the upper bound computation and the hierarchical exponential mechanism
    discussed in Section~\ref{sec:dc_aware} (the full Algorithm~\ref{algo:em_opt}). 
\squishend
By default, we experiment with a total privacy budget of $\epsilon=1$ unless specified otherwise. 
\subsection{Results}\label{sec:results}

\begin{figure*}
    \begin{subfigure}[b]{\textwidth}
         \centering
         \includegraphics[width=0.19\textwidth]{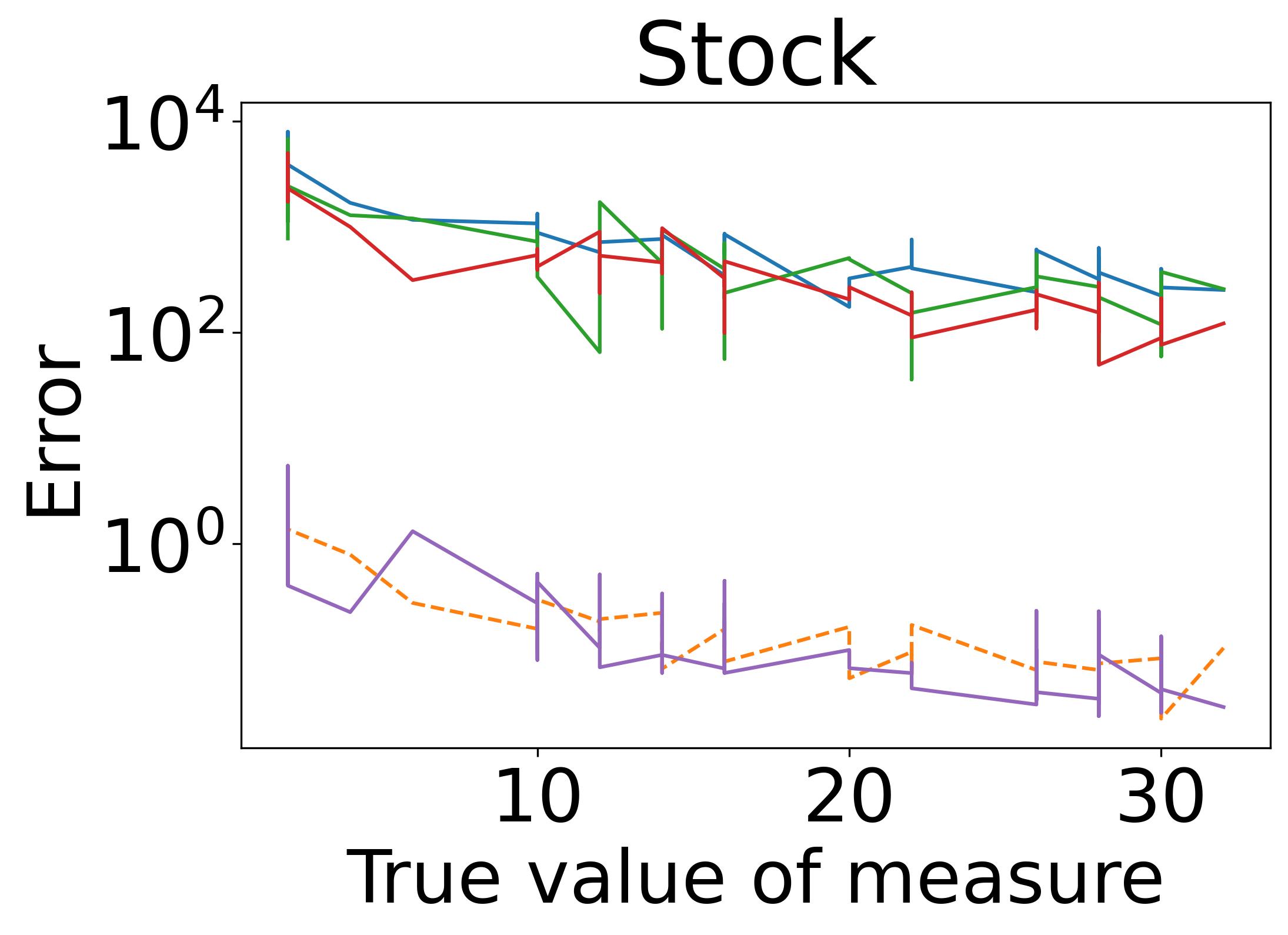}
         \hfill
         \includegraphics[width=0.19\textwidth]{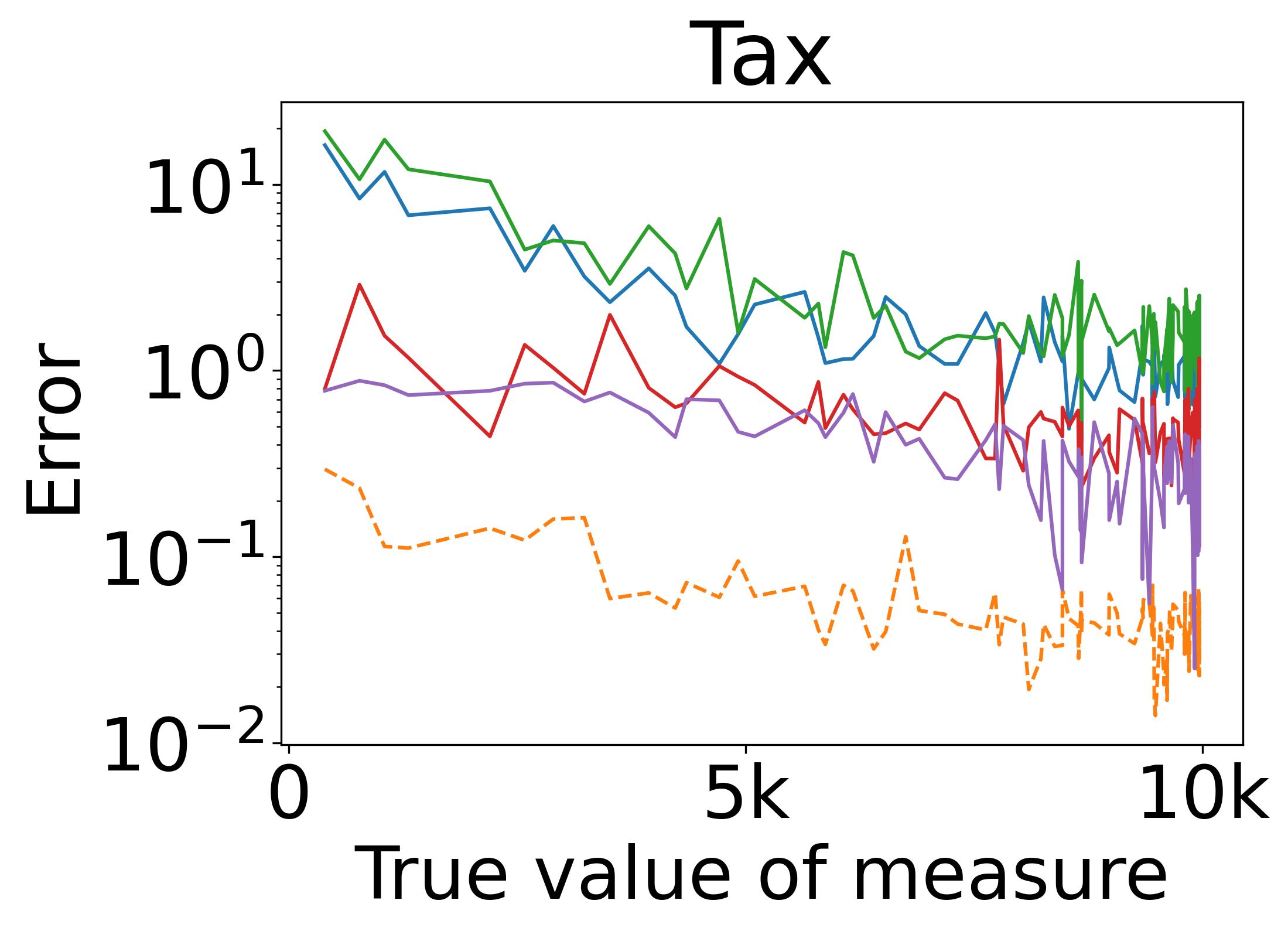}
         \hfill
         \includegraphics[width=0.19\textwidth]{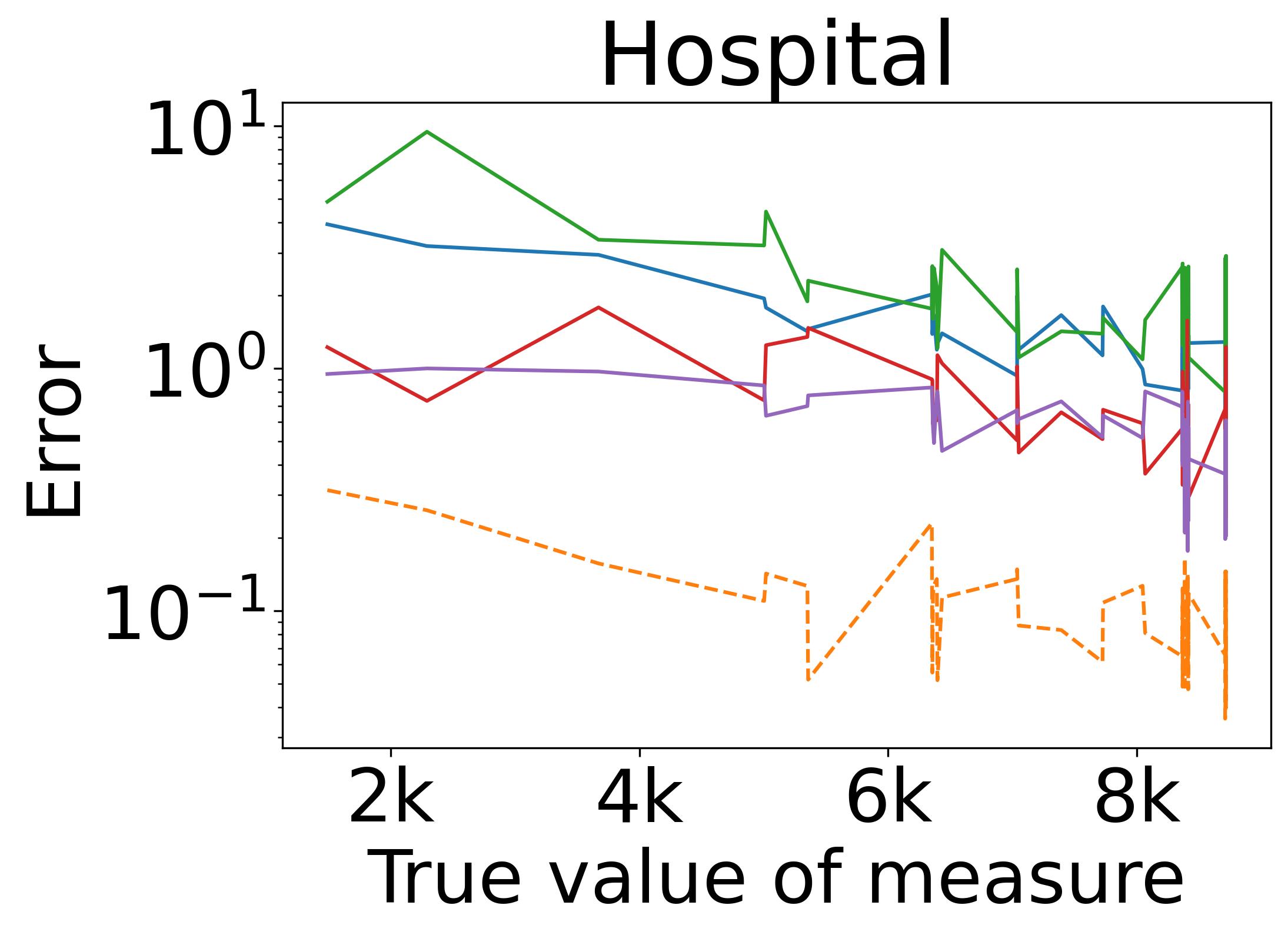}
         \hfill
         \includegraphics[width=0.19\textwidth]{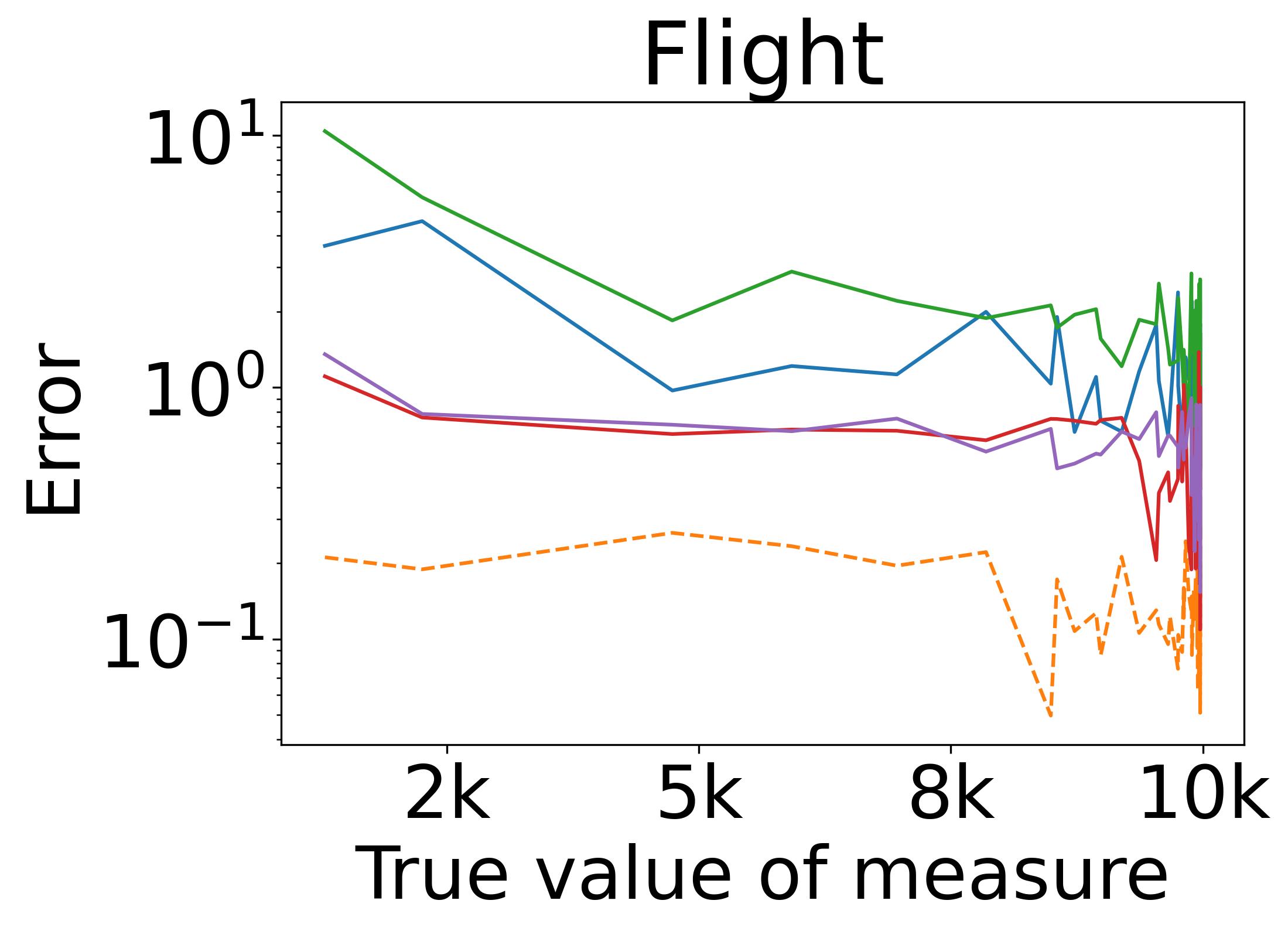}
        \hfill
         \includegraphics[width=0.19\textwidth]{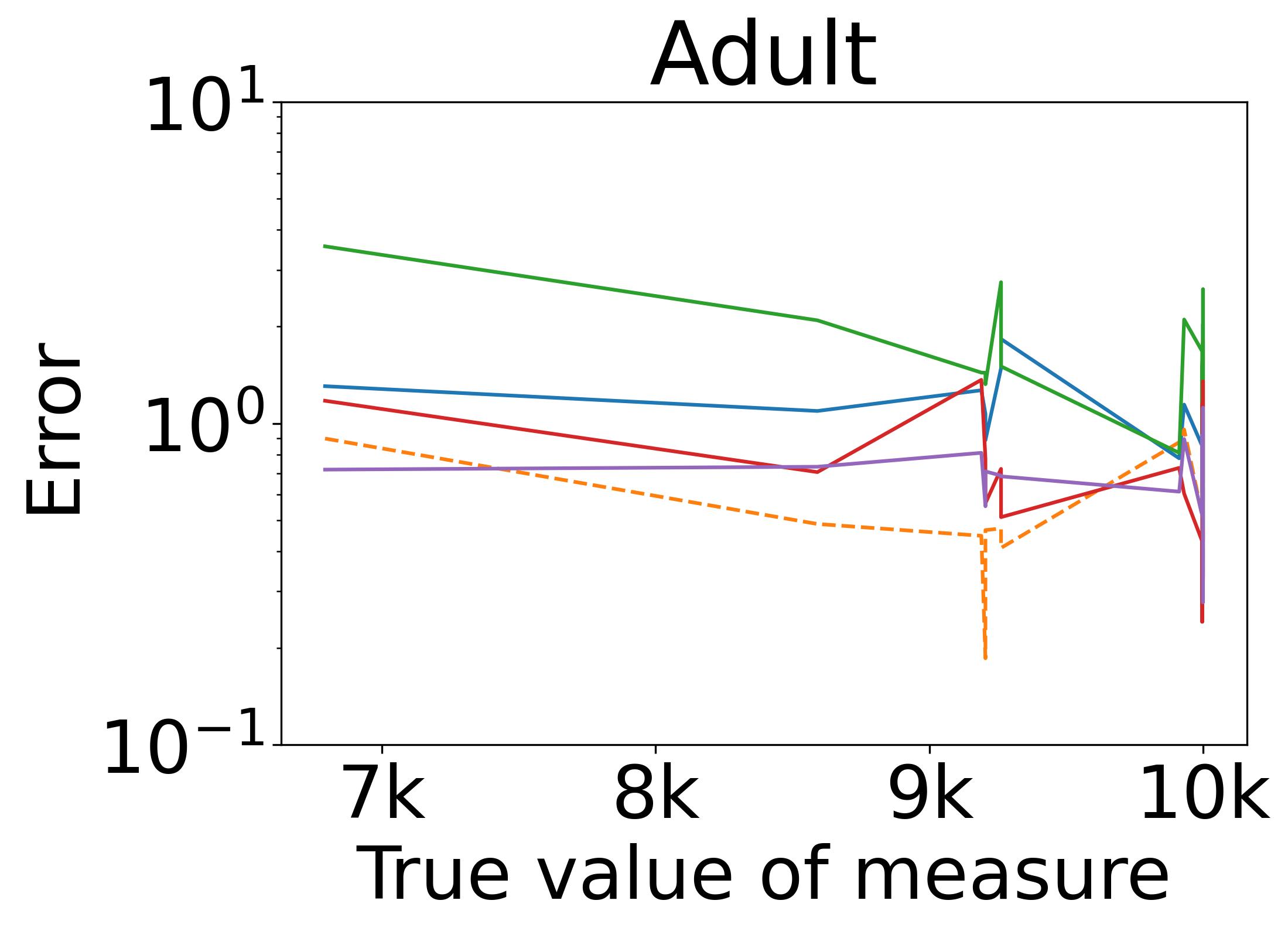}
        \includegraphics[width=\textwidth]{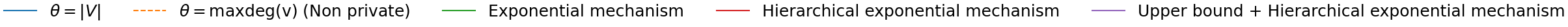}
         \caption{$\problematic$ (Positive degree nodes)}
         \label{fig:comparing_strategies_pdnodes}
     \end{subfigure}
     \begin{subfigure}[b]{\textwidth}
         \centering
         \includegraphics[width=0.19\textwidth]{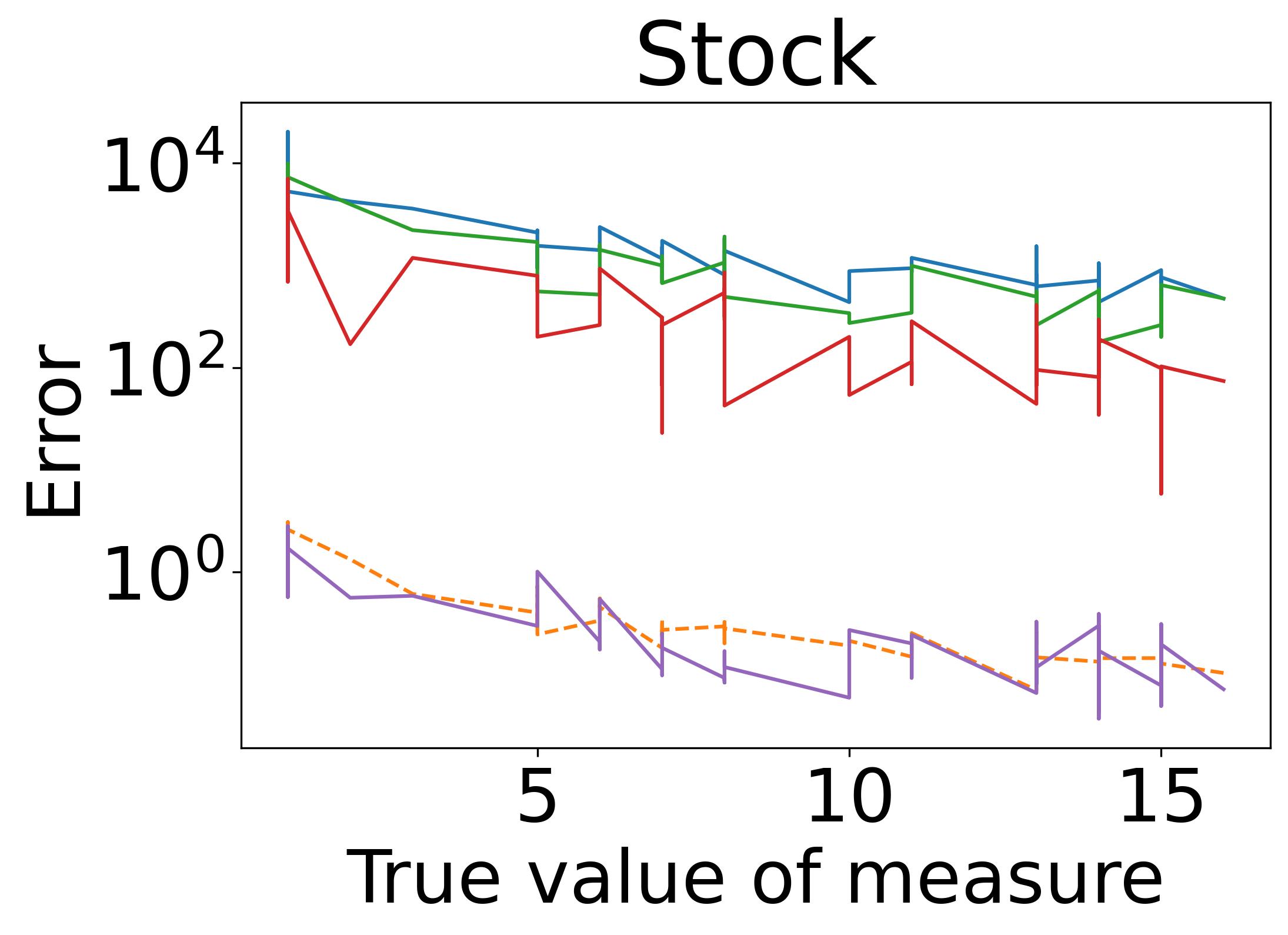}
         \hfill
         \includegraphics[width=0.19\textwidth]{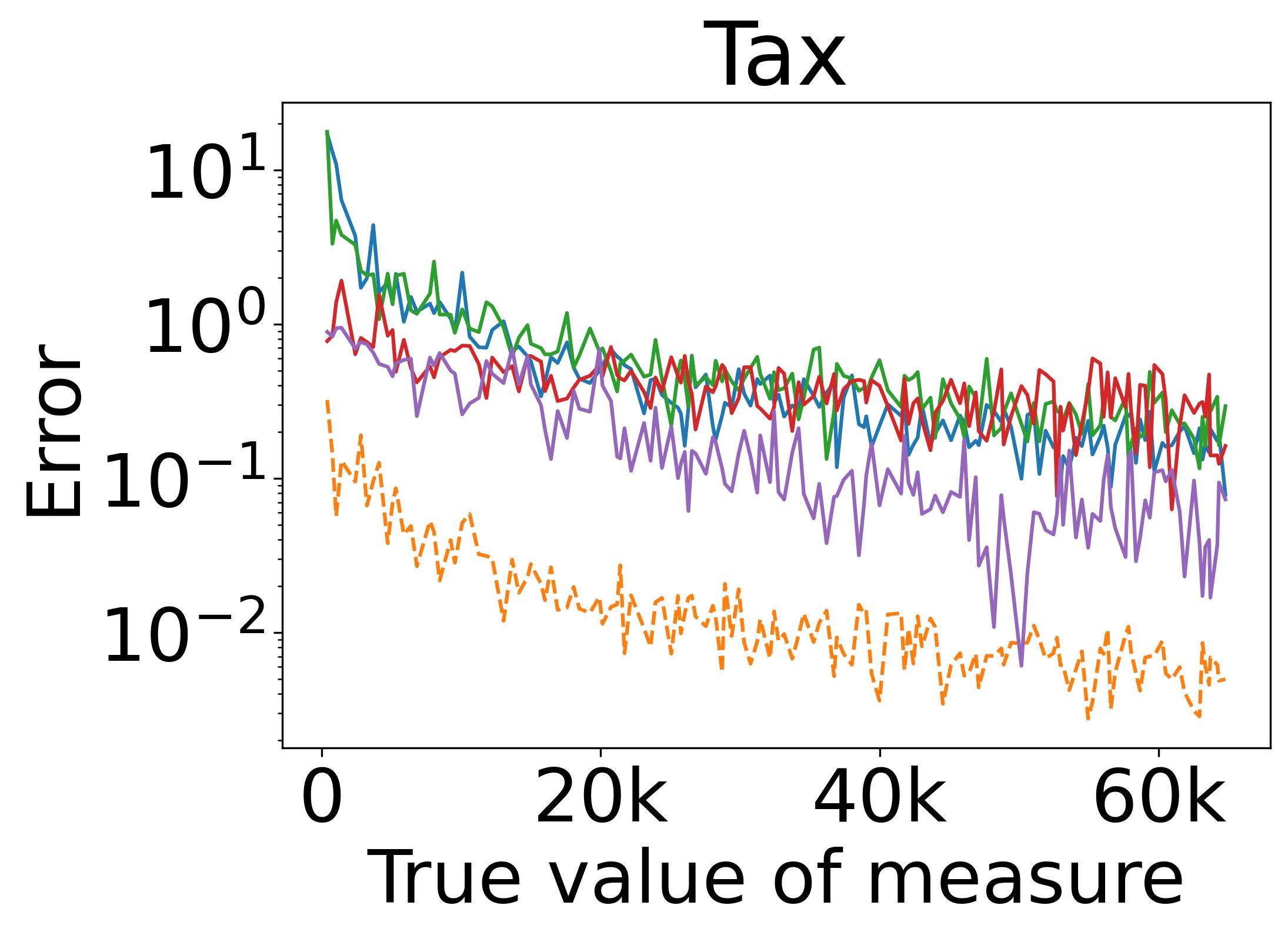}
         \hfill
         \includegraphics[width=0.19\textwidth]{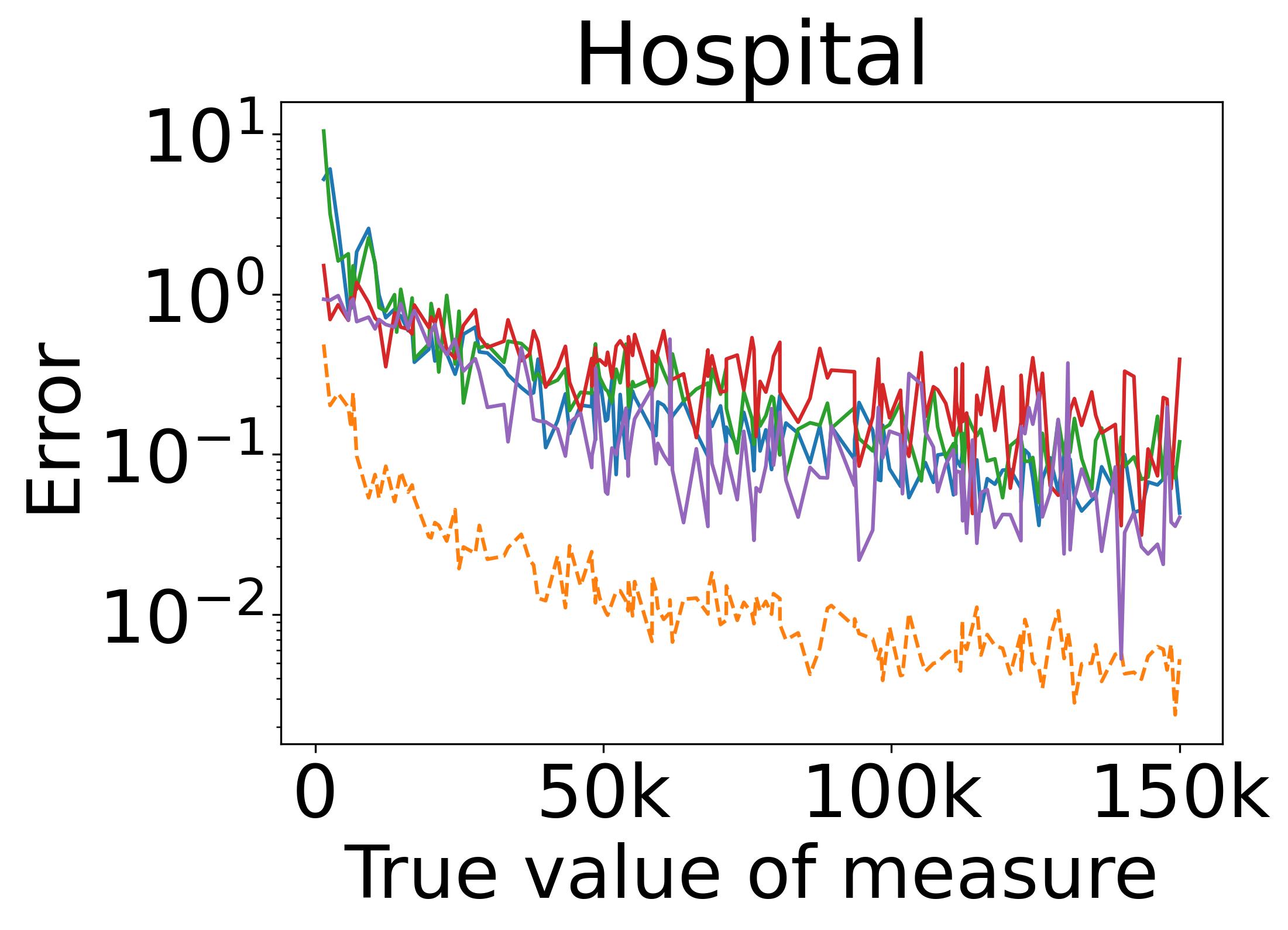}
         \hfill
         \includegraphics[width=0.19\textwidth]{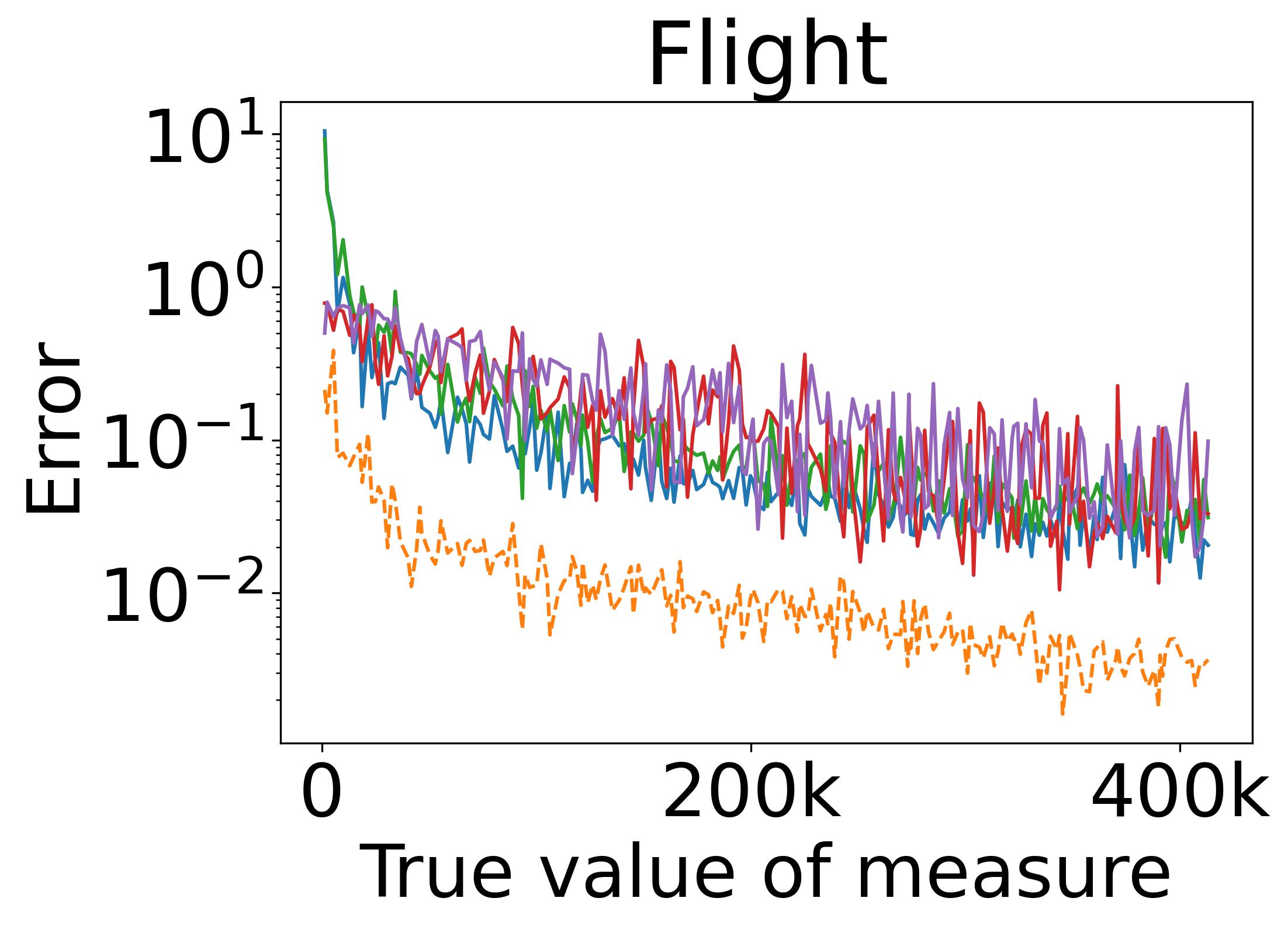}
        \hfill
         \includegraphics[width=0.19\textwidth]{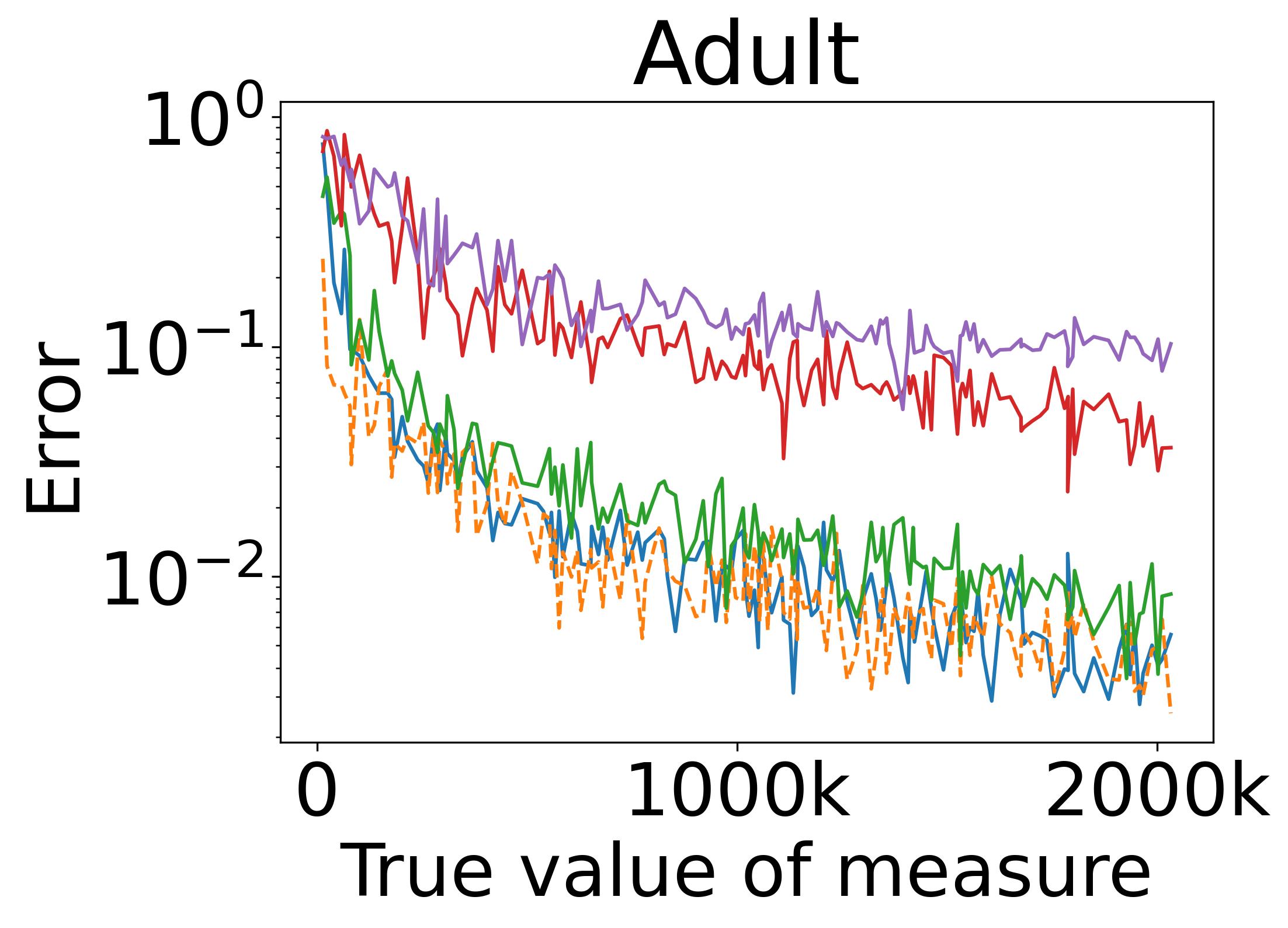}
        \includegraphics[width=\textwidth]{images/legend_1.png}
         \caption{$\mininconsistency$ (Number of edges)}
         \label{fig:comparing_strategies_nedges}
     \end{subfigure}
     \caption{Computing different strategies for choosing $\theta$ for all datasets with RNoise at $\alpha=0.01$ and $\epsilon=1$. The datasets are arranged according to their densities from sparsest (left) to densest (right). }
     \label{fig:comparing_strategies}
\end{figure*}

\paratitle{True vs private estimation}
In \Cref{fig:tp_RNoise,fig:tp_conoise}, we plot the true vs. private estimates at $\epsilon=1$ for all datasets with RNoise ($\alpha=0.01$) and CONoise (200 iterations) respectively. The datasets are ordered according to their densities from left to right. Each figure contains the measured value ($\mininconsistency$, $\problematic$, or $\repair$) on the Y-axis and the number of iterations on the X-axis. For the CONoise, the number of iterations is set to 200 for every dataset, and for RNoise, the iterations correspond to the number of iterations required to reach 1\% ($\alpha = 0.01$) number of random violations. The orange line corresponds to the true value of the measure, and the blue line corresponds to the private measure using our approach. For $\mininconsistency$ and $\problematic$ measures, the blue line represents the upper bound + hierarchical exponential mechanism strategy described in Section~\ref{sec:dc_aware} along with its standard deviation in shaded blue. For the $\repair$ measure, it represents the private minimum vertex cover size algorithm. We also add a baseline approach using a state-of-the-art private SQL approach called R2T~\cite{dong2022r2t}. We add this baseline only for the $\mininconsistency$ measure as $\repair$ cannot be written with SQL, and $\problematic$ requires the DISTINCT/GROUP BY clause that R2T does not support. Based on the experiments, we draw three significant observations. 

First, compared to the SQL baseline (R2T), our approach has a better relative error on average across all datasets. However, R2T is slightly behind for moderate to high dense datasets such as Tax ($0.207$ vs. $0.334$), Hospital ($0.209$ vs. $0.386$), and Flight($0.202$ vs. $0.205$) and Adult ($0.187$ vs. $0.269$) but \reva{falls short} for sparse datasets such as Stock ($0.492$ vs. $137.05$). This is because the true value of the measure is small, and R2T adds large amounts of noise. 

Second, our approach for the $\mininconsistency$ and $\problematic$ fluctuates more and has a higher standard deviation compared to the $\repair$ measure. This is because of the privacy noise due to the relatively high sensitivity of our upper bound + hierarchical exponential mechanism approach.
On the other hand, the vertex cover size approach for $\repair$ has a sensitivity equal to $2$ and, therefore, does not show much fluctuation when the true measure value is large enough. 

Third, we observe that our approach generally performs well across all five datasets and all inconsistency measures. The $\mininconsistency$ and $\problematic$ measures have average errors of $0.25$ and $0.46$, respectively, across all datasets where Stock is the worst performing dataset for $\mininconsistency$ and Adult is the worst performing for $\problematic$. The $\repair$ performs the best with an average error of $0.08$, with Stock as the worst-performing dataset. We investigate the performance of each dataset in detail in our next experiment and find out that the density of the graphs plays a significant role in the performance of our algorithms.

\paratitle{Comparing different strategies for choosing $\theta$}
In \Cref{fig:comparing_strategies}, we present the performance of different algorithm variations in computing $\problematic$ and $\mininconsistency$ for all datasets using RNnoise at $\alpha=0.01$ and $\epsilon=1$. The y-axis in each figure shows the logarithmic scaled error, while the x-axis displays the actual measure value, with different colors representing the strategies. The graphs are ordered from most sparse (Stock) to least sparse (Adult) to compare methods for choosing the $\theta$ value at $\epsilon = 1$ ($\epsilon_1 = 0.4$, $\epsilon_2=0.6$). The methods include all variations described in the algorithm variation section.
\eat{
: \xh{we can drop this part}
\begin{itemize}
 \item Baseline 1, $\theta = |V|$  (blue)
    \item Baseline 2, $\theta = \degree_{\max}(\graph)$ (orange dashed)
    \item Exponential mechanism (EM) with $\epsilon_1 = 0.4$ (green)
    \item Hierarchical EM with $\epsilon_1 = 0.4$ (red)
    \item Upper bound + hierarchical EM with $\epsilon_0 = 0.1$ and $\epsilon_1= 0.3$ (purple)
\end{itemize}}

We note all the strategies are private except baseline 2 (orange dash line) that sets $\theta$ as the true maximum degree of the conflict graph. 

Our experimental results, based on error trends and graph density, reveal several key observations.

First, we consistently observed that the initial error was higher at smaller iterations across all five datasets and inconsistency measures. This is because, at smaller iterations, the true value of the measures is small due to fewer violations, and the privacy noise dominates the signal of true value.

Second, for the sparsest dataset (Stock), all strategies have errors of magnitude 3-4 larger, except the non-private baseline (orange) and our approach using both upper bound and hierarchical exponential mechanism (purple). This is because the candidate set contains many large candidates, and it is crucial to prune it using the upper bound strategy to get meaningful results. 

Third, for the moderately sparse graphs (Tax and Hospital), our approach consistently (purple) consistently outperformed other private methods.  However, the two-step hierarchical exponential mechanism (red), which had a $3$ magnitude higher error for Stock, demonstrated comparable performance within a 1-magnitude error difference for Tax and Hospital.  This suggests that when the true max degree is not excessively low, estimating it without the upper bound strategy can be effective.

Finally, for the densest graphs (Flight and Adult), we observe that the optimized exponential mechanisms (red and purple) outperform the private baselines (blue and green) for the $\problematic$ measure (nodes with positive degree) plots (above). However, they fail to beat even the naive baseline (blue) for the $\mininconsistency$ (number of edges) measure (below). This is because the optimal degree bound value $\mininconsistency$ over the dense graphs is close to the largest possible value $|V|$. For such a case, our optimized EM is not able to prune too many candidates and lower the sensitivity, and hence, it wastes some of the privacy budget in the pruning process. However, the relative errors of all the algorithms are reasonably small for dense graphs, and the noisy answers do preserve the order of the true measures (shown in previous experiments in Figures~\ref{fig:tp_RNoise} and \ref{fig:tp_conoise}).

\paratitle{Varying privacy budget}
Figure~\ref{fig:varying_eps} illustrates how our algorithms perform at $\epsilon \in [0.1, 0.2, 0.5, 1.0, 2.0, 3.0, 5.0]$ with varying privacy budgets. The rightmost figure for the repair measure $\repair$ has a log scale on the y-axis for better readability. We experiment with three datasets of different density properties (sparsest Stock and densest Flight) and show that our algorithm gracefully scales with the $\epsilon$ privacy budget for all three inconsistency measures. We also observe that the algorithm has a more significant error variation at smaller epsilons ($<1$) except for Stock, which has a larger variation across all epsilons. This happens when the true value for this measure is small, and adding noise at a smaller budget ruins the estimate drastically. For the $\repair$ measure, the private value reaches with $0.05$ error at $\epsilon=3$ for Stock and as early as $\epsilon=0.1$ for others. 

\begin{figure}
    \centering
    \includegraphics[width=\linewidth]{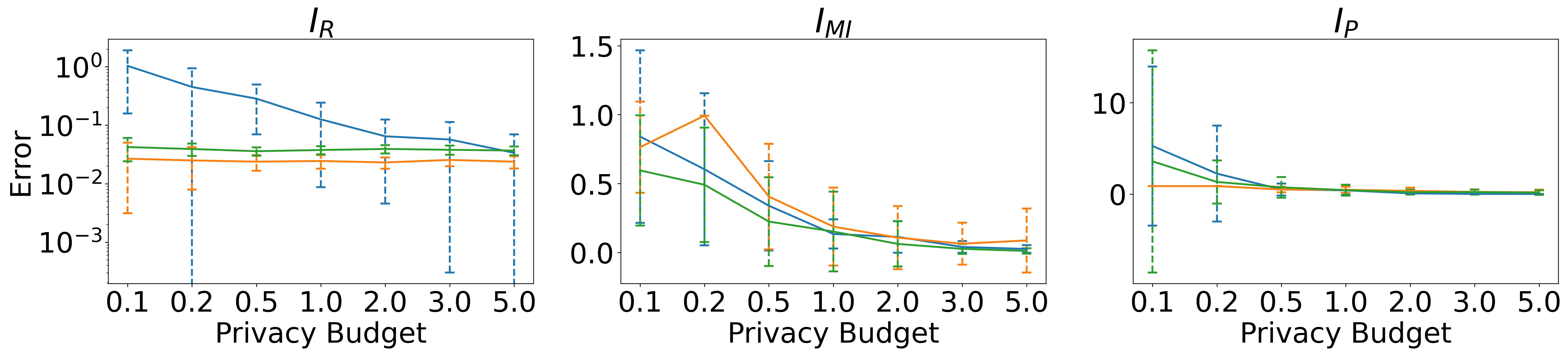}
    \includegraphics[width=0.5\linewidth]{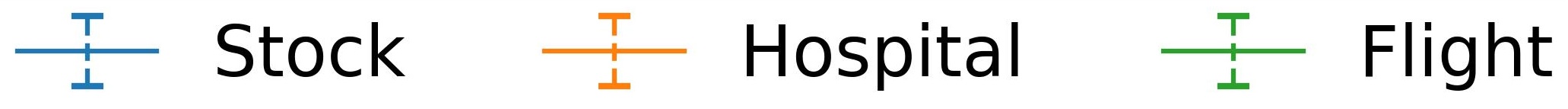}
    \caption{Computing inconsistency measures for different datasets with RNoise at $\alpha = 0.005$ and varying privacy budgets. $\repair$ plot has a y-axis in the log scale.}
    \label{fig:varying_eps}
\end{figure}

\revc{
\paratitle{Runtime and scalability analysis}
Figure~\ref{fig:runtime} presents the runtime of our methods for each measure. We fix the privacy budget $\epsilon=1$ and run three experiments by varying the graph size, numbers of DCs, and dataset. For the first experiment, we use our largest dataset, Tax, and vary the number of nodes from $10^2$ to $10^6$. RNoise uses  $\alpha=0.005$ in the left plot and $\alpha=0.01$ in the center plot. We observe that the number of edges scales exponentially when we increase the number of nodes, and the time taken by our algorithm is proportional to the graph size. With a graph of $10^2$ nodes and $\leq 10$ edges, our algorithm takes $10^{-3}$ seconds and goes up to $4500$ seconds with $10^6$ nodes and $322$ million edges. We omit the experiment with $10^6$ nodes and $\alpha=0.01$ as the graph size for this experiment went over 30GB and was not supported by the pickle library we use to save our graph. This is not an artifact of our algorithm and can be scaled in the future using other graph libraries. 

For our second experiment, we choose a subset of $10k$ rows of the Flight dataset and vary the number of DCs to $13$ with $\alpha=0.005$. With one DC and $\alpha=0.005$, our algorithm takes approximately 5 seconds for $\mininconsistency$ and $\problematic$ and $\leq 1$ second for $\repair$, and goes up to $25$ seconds and $5$ seconds, respectively, for $\alpha=0.065$ and 13 DCs. We also notice some dips in the trend line (e.g., at 10 and 13 constraints) because the exponential mechanism chooses larger thresholds at those points, and the edge addition algorithm takes slightly longer with chosen thresholds. Our third experiment on varying datasets behaves similarly and is deferred to Appendix A.4 in the full version~\cite{full_paper} for lack of space.  
}
\begin{figure}
    \centering
    \includegraphics[width=\linewidth]{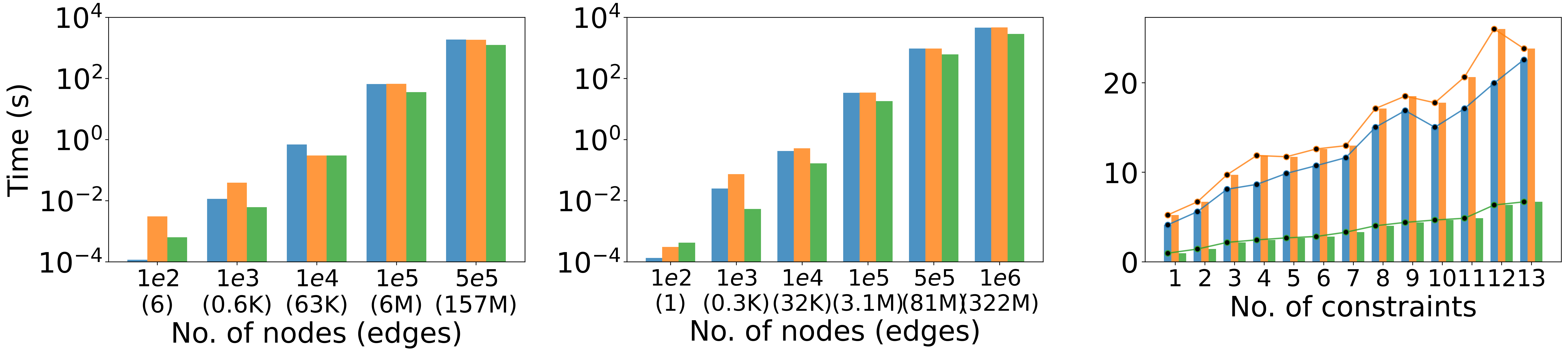}
    \includegraphics[width=0.4\linewidth]{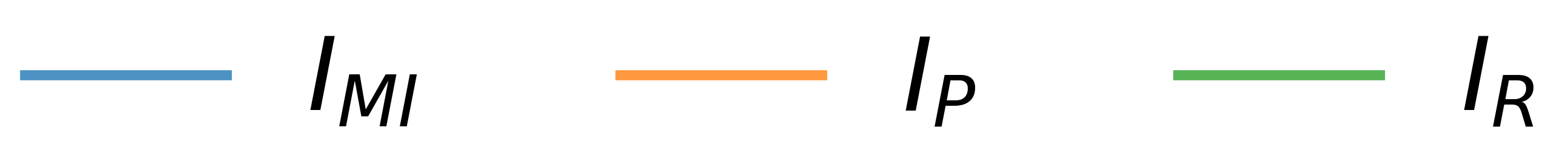}
    \caption{\revc{Runtime analysis for all measures. Varying graph size on Tax dataset with RNoise $\alpha=0.005$ (left) and $\alpha=0.01$ (center) and varying \#DCs for Flight dataset (right). $\mininconsistency$ and $\problematic$ plots have y-axis in the log scale.}}
    \label{fig:runtime}
\end{figure}

\section{Related Work}\label{sec:related}
We survey relevant works in the fields of DP and data repair.
\paragraph{Differential privacy}
DP has been studied in multiple settings~\cite{johnson2017practical,kotsogiannis2019privatesql,wilson2019differentially,johnson2018towards,DBLP:journals/pvldb/McKennaMHM18,tao2020computing}, including systems that support complex SQL queries that, in particular, can express integrity constraints~\cite{dong2022r2t}. The utilization of graph databases under DP has also been thoroughly explored, with both edge privacy~\cite{hay2009accurate, karwa2011private,karwa2012differentially,sala2011sharing,karwa2014differentially,zhang2015private} and node privacy~\cite{KasiviswanathanNRS13,blocki2013differentially,day2016publishing}. 
Our approach draws on \cite{day2016publishing} to allow efficient DP computation of the inconsistency measures over the conflict graph. In contrast, we have seen worse performance from alternative approaches for releasing graph statistics that tend to truncate edges or nodes aggressively. 
In the context of data quality, previous work~\cite{KrishnanWFGK16} has proposed a framework for releasing a private version of a database relation for publishing, supporting specific repair operations, while more recent, work~\cite{GeMHI21} provides a DP synthetic data generation mechanism that considers soft DCs~\cite{ChomickiM05}.

\paragraph{Data repair}
Various classes of constraints were proposed in the literature, including FDs, conditional FDs~\cite{bohannon2007conditional}, metric FDs~\cite{koudas2009metric}, and DCs~\cite{ChomickiM05}. We focused on DCs, a general class of integrity constraints that subsumes the aforementioned constraints. 
While computing the minimal data repair in some cases has been shown to be polynomial time~\cite{LivshitsKR20}, computing the minimal repair in most general cases, corresponding to \repair, is NP-hard. 
Therefore, a prominent vein of research has been devoted to utilizing these constraints for data repairing~\cite{Afrati2009,RekatsinasCIR17,ChuIP13,BertossiKL13,FaginKK15,LivshitsKR20,GiladDR20,GeertsMPS13}. The repair model in these works varies between several options: tuple deletion, cell value updates, tuple addition, and combinations thereof. 
The process of data repairing through such algorithms can be time consuming due to the size of the data and the size and complexity of the constraint set. 
Hence, previous work~\cite{LivshitsBKS20,LivshitsKTIKR21} proposed to keep track of the repairing process and ensure that it progresses correctly using inconsistency measures. 
In our work, we capitalize on the suitability of these measures to DP as they provide an aggregate form that summarizes the quality of the data for a given set of constraints.  

\common{
\section{Future Work}\label{sec:future}
This paper analyzes a novel problem of computing inconsistency measures privately. There are many interesting directions for the continuation of this work. This paper shows a naive threshold bound for general DCs that can be improved for better performance of our algorithm. Our proposed conflict graphs algorithm is intractable for the $\drastic$ and $\maxconsistency$ measures that other heuristic or approximation-based approaches may solve in the future. The vertex cover size algorithm using the stable ordering of edges is general purpose and may be used in other problems outside of inconsistency measures. It can also be analyzed further in future work to return the vertex cover set. Another interesting direction is to develop these measures in a multi-relational database setting. Our approach can be extended to multi-relational tables as long as we can create conflict graphs representing the violations. 
However, in the multi-table setting, we must consider additional constraints that require tackling several challenges. In particular, these challenges may arise when we have non-binary or non-anti-monotonic constraints. Non-binary constraints with more than two tuples participating in a constraint lead to hypergraphs, and constraints like foreign key and inclusion constraints are non-anti-monotonic. Thus, they cannot be represented as conflict graphs, and as such, are outside the scope of our work. 
Furthermore, in the context of DP, constraints on multi-relational tables also have implications for defining neighboring datasets and sensitivity that must be carefully considered. Our future work also includes studying the problem of private inconsistency measures with different privacy notions, such as k-anonymity, and using these private inconsistency measures in real-world data cleaning applications.
}

\section{Conclusions}\label{sec:conclusions}
We proposed a new problem of inconsistency measures for private datasets in the DP setting. We studied five measures and showed that two are intractable with DP, and the others face a significant challenge of high sensitivity. To solve this challenge, we leveraged the dataset's conflict graph and used graph projection and a novel approximate DP vertex cover size algorithm to accurately estimate the private inconsistency measures. We found that parameter selection was a significant challenge and were able to overcome it using optimization techniques based on the constraint set. To test our algorithm, we experimented with five real-world datasets with varying density properties and showed that our algorithm could accurately calculate these measures across all datasets.

\begin{acks}
Shubhankar Mohapatra was supported by an Ontario graduate, vector graduate research, and David R. Cheriton scholarships. The work of Xi He was supported by NSERC through a Discovery Grant, an alliance grant, and the Canada CIFAR AI Chairs program. The work of Amir Gilad was funded by the Israel Science Foundation (ISF) under grant 1702/24, the Scharf-Ullman Endowment, and the Alon Scholarship. 
\end{acks}


\newpage
\balance
\bibliographystyle{ACM-Reference-Format}
\bibliography{bibtex}

\ifpaper

\else
\newpage 
\appendix

\section{Theorems and Proofs}
\subsection{Proof for \cref{prop:sensitivity}}
\label{app:proof_sensitivity}
Recall the proposition states that given a database $D$ and a set of DCs $\constraintset$, where $|D|=n$, the following holds:
\begin{enumerate}
    \item The global sensitivity of \drastic\ is 1. 
    \item The global sensitivity of 
    \mininconsistency\ is $n$.       
    \item The global sensitivity of \problematic\ is $n$.
    \item The global sensitivity of \maxconsistency\ is exponential in $n$.
    \item The global sensitivity of \repair\ is 1.
\end{enumerate}
\begin{proof}
Consider two neighbouring datasets, $D$ and $D'$.
\paragraph{\drastic}
Adding or removing one tuple from the dataset will affect the addition or removal of all conflicts related to it in the dataset. In the worst case, this tuple could remove all conflicts in the dataset, and the $\drastic$ could go from 1 to 0 or vice versa.
\paragraph{\mininconsistency and \problematic}
The $\mininconsistency$ and $\problematic$ are concerned with the set of minimally inconsistent subsets $MI_\constraintset(D)$. The $\mininconsistency$ measure computes the total number of inconsistent subsets $|MI_\constraintset(D)|$ and $\problematic$ computes the total number of unique rows participating in $MI_\constraintset(D)$, $|\cup MI_\constraintset(D)|$. Now, without loss in generality, let's assume $D'$ has an additional tuple compared to $D$. In the worst case, the extra row could violate all other rows in the dataset, adding $|D|$ inconsistent subsets to $MI_\constraintset(D)$. Therefore, in the worst case, the $\mininconsistency$ measure and $\problematic$ could change by $|D|$.
\paragraph{\maxconsistency}
The \maxconsistency\ measure is \#P-complete and can be computed for a conflict graph $\graph$ if the dataset has only FDs in the constraint set $\constraintset$ and $\graph$ is $P_4$-free~\cite{KimelfeldLP20}. The maximum number of maximal independent sets~\cite{moon1965cliques, GriggsGG88} $f(n)$ for a graph with $n$ vertices is given by :  If $n \geqq 2$, then $f(n)= \begin{cases}3^{n / 3}, & \text { if } n \equiv 0(\bmod 3) \text {; } \\ 4.3^{[n / 3]-1}, & \text { if } n \equiv 1(\bmod 3) \text {; } \\ 2.3^{[n / 3]}, & \text { if } n \equiv 2(\bmod 3) .\end{cases}$

Using the above result, we can see that adding or removing a node in the graph can affect the total number of maximal independent sets in the order of $3^{n}$.
\paragraph{\repair}
Without loss in generality, let's assume $D'$ has an additional tuple compared to $D$. This extra row may or may not add extra violations. However, repairing this extra row in $D'$ will always remove these added violations. Therefore, \repair\ for $D'$ can only be one extra than $D$.
\end{proof} 

\subsection{DP Analysis for  $\mininconsistency$ and $\problematic ($\cref{algo:graph_general})}\label{app:graph_general}
\subsubsection{Proof for \cref{thm:privacy_proof_dc_oblivious}}\label{app:privacy_proof_dc_oblivious}
The theorem states that \cref{algo:graph_general} satisfies $(\epsilon_1 + \epsilon_2)$-node DP for $\graph$ and $(\epsilon_1 + \epsilon_2)$-DP for the input database $D$.
\begin{proof}
The proof is straightforward using the composition property of DP~\cref{prop:DP-comp-post}.
\end{proof}

\subsubsection{Proof for \cref{lemma:sens_quality}}\label{app:sens_quality}
\cref{lemma:sens_quality} states that the sensitivity of the quality function $q_{\epsilon_2}(\mathcal{G}, \theta_i)$ in Algorithm~\ref{algo:expo_mech_basic} defined in Equation~\eqref{eq:quality_function} is $\theta_{\max}=\max(\Theta)$. 

\begin{proof}
We prove the lemma for the $\mininconsistency$ measure and show that it is similar for $\problematic$. Let us assume that $\mathcal{G}$ and $\mathcal{G}'$ are two neighbouring graphs and $\mathcal{G}'$ has one extra node $v^*$. 
    \begin{equation*}
        \begin{split}
            &\|q_{\epsilon_2}(\mathcal{G}, \theta) - q_{\epsilon_2}(\mathcal{G}^\prime, \theta)\| \leq -|\mathcal{G}_{\theta_{\max}}.E| + |\mathcal{G}_{\theta}.E| - \sqrt{2}\frac{\theta}{\epsilon_1} \\ &+ |\mathcal{G}'_{\theta_{\max}}.E| - |\mathcal{G}'_{\theta}.E| + \sqrt{2}\frac{\theta}{\epsilon_1} \\
            &\leq \left(|\mathcal{G}'_{\theta_{\max}}.E| - |\mathcal{G}_{\theta_{\max}}.E|\right) - \left(|\mathcal{G}'_\theta.E| - |\mathcal{G}_\theta.E|\right)\\
            &\leq \theta_{\max} - \left(|\mathcal{G}'_\theta.E| - |\mathcal{G}_\theta.E|\right) 
            \leq \theta_{\max}    
        \end{split}
    \end{equation*}
    The second last inequality is due to Lemma~\ref{lemma:sens_lap} that states that $|\mathcal{G}'_{\theta_{max}}.E| - |\mathcal{G}_{\theta_{max}}.E| \leq \theta_{max}$. The last inequality is because $|\mathcal{G}'_\theta.E| \geq |\mathcal{G}_\theta.E|$. Note that the neighboring graph $\mathcal{G}'$ contains all edges of $\mathcal{G}$ plus extra edges of the added node $v^*$. Due to the stable ordering of edges in the edge addition algorithm, each extra edge of $v^*$ either substitutes an existing edge or is added as an extra edge in $\mathcal{G}_\theta$. Therefore, the total edges $|\mathcal{G}'_\theta.E|$ is equal or larger than $|\mathcal{G}_\theta.E|$. We elaborate this detail further in the proof for Lemma~\ref{lemma:sens_lap}. For the $\problematic$ measure, the term in the last inequality changes to $|\mathcal{G}'_\theta.V_{>0}| - |\mathcal{G}_\theta.V_{>0}|$ and is also non-negative because $\mathcal{G}'$ contains an extra node that can only add and not subtract from the total number of nodes with positive degree.
\end{proof}

\subsubsection{Proof for ~\cref{lemma:sens_lap}}
This lemma states that the sensitivity of $f\circ\pi_\theta(\cdot)$ in Algorithm~\ref{algo:graph_general} is $\theta$, where $\pi_\theta$ is the edge addition algorithm with the user input $\theta$, and $f(\cdot)$ return returns edge count for $\mininconsistency$ and the number of nodes with positive degrees for $\problematic$.

\begin{proof}
Let's assume without loss of generality that
$\mathcal{G}^{\prime}=\left(V^{\prime}, E^{\prime}\right)$ has an additional node $v^{+}$compared to $\mathcal{G}=$ $(V, E)$, i.e., $V^{\prime}=V \cup\left\{v^{+}\right\}, E^{\prime}=E \cup E^{+}$, and $E^{+}$is the set of all edges incident to $v^{+}$in $\mathcal{G}^{\prime}$. We prove this lemma separately for $\mininconsistency$ and $\problematic$.

\paragraph{For \mininconsistency}
Let $\Lambda^{\prime}$ be the stable orderings for constructing $\pi_\theta\left(\mathcal{G}^{\prime}\right)$, and $t$ be the number of edges added to $\pi_\theta\left(\mathcal{G}^{\prime}\right)$ that are incident to $v^{+}$. Clearly, $t \leq \theta$ because of the $\theta$-bounded algorithm. Let $e_{\ell_1}^{\prime}, \ldots, e_{\ell_t}^{\prime}$ denote these $t$ edges in their order in $\Lambda^{\prime}$. Let $\Lambda_0$ be the sequence obtained by removing from $\Lambda^{\prime}$ all edges incident to $v^{+}$, and $\Lambda_k$, for $1 \leq k \leq t$, be the sequence obtained by removing from $\Lambda^{\prime}$ all edges that both are incident to $v^{+}$and come after $e_{\ell_k}^{\prime}$ in $\Lambda^{\prime}$. Let $\pi_\theta^{\Lambda_k}\left(\mathcal{G}^{\prime}\right)$, for $0 \leq k \leq t$, be the graph reconstructed by trying to add edges in $\Lambda_k$ one by one on nodes in $\mathcal{G}^{\prime}$, and $\lambda_k$ be the sequence of edges from $\Lambda_k$ that are actually added in the process. Thus $\lambda_k$ uniquely determines $\pi_\theta^{\Lambda_k}\left(\mathcal{G}^{\prime}\right)$; we abuse the notation and use $\lambda_k$ to also denote $\pi_\theta^{\Lambda_k}\left(\mathcal{G}^{\prime}\right)$. We have $\lambda_0=\pi_\theta(\mathcal{G})$, and $\lambda_t=\pi_\theta\left(\mathcal{G}^{\prime}\right)$.

In the rest of the proof, we show that $\forall k$ such that $1 \leq k \leq t$, at most 1 edge will differ between $\lambda_k$ and $\lambda_{k-1}$. This will prove the lemma because there are at most $t$ (upper bounded by $\theta$) edges that are different between $\lambda_t$ and $\lambda_0$.

To prove that any two consecutive sequences differ by at most 1 edge, let's first consider how the sequence $\lambda_k$ differs from $\lambda_{k-1}$. Recall that by construction, $\Lambda_k$ contains one extra edge in addition to $\Lambda_{k-1}$ and that this edge is also incident to $v^*$. Let that additional differing edge be $e_{\ell_k}^\prime = (u_j, v^+)$. In the process of creating the graph $\pi_\theta^{\Lambda_k}(\mathcal{G}^{\prime})$, each edge will need a decision of either getting added or not. The decisions for all edges coming before $e_{\ell_k}^{\prime}$ in $\Lambda^{\prime}$ must be the same in both $\lambda_k$ and $\lambda_{k-1}$. Similarly, after $e_{\ell_k}^{\prime}$, the edges in $\Lambda_k$ and $\Lambda_{k-1}$ are exactly the same. However, the decisions for including the edges after $e_{\ell_k}^{\prime}$ may or may not be the same. Assuming that there are a total of $s \geq 1$ different decisions, we will observe how the additional edge $e_{\ell_k}^{\prime}$ makes a difference in decisions.

When $s=1$, the only different decision must be regarding differing edge $e_{\ell_k}^\prime = (u_j, v^+)$ and that must be including that edge in the total number of edges for $\lambda_k$. Also note that due to this addition, the degree of $u_j$ gets added by 1 which did not happen for $\lambda_{k-1}$. When $s>1$, the second different decision must be regarding an edge incident to $u_j$ and that is because degree of $u_j$ has reached $\theta$, and the last one of these, denoted by $(u_j, u_{i \theta})$ which was added in $\lambda_{k-1}$, cannot be added in $\lambda_k$. In this scenario, $u_j$ has the same degree (i.e., $\theta$ ) in both $\lambda_k$ and $\lambda_{k-1}$. Now if $s$ is exactly equal to 2, then the second different decision must be not adding the edge $(u_j, u_{i \theta})$ to $\lambda_k$. Again, note here that as $(u_j, u_{i \theta})$ was not added in $\lambda_k$ but was added in $\lambda_{k-1}$, there is still space for one another edge of $u_{i \theta}$. If $s>2$, then the third difference must be the addition of an edge incident to $u_{i \theta}$ in $\lambda_k$. This process goes on for each different decision in $\lambda_k$ and $\lambda_{k-1}$. Since the total number of different decisions $s$ is finite, this sequence of reasoning will stop with a difference of at most 1 in the total number of the edges between $\lambda_{k-1}$ and $\lambda_k$.

\paragraph{For \problematic}
Assume, in the worst case, the graph $\mathcal{G}$ is a star graph with $n$ nodes such that there exists an internal node that is connected to all other $n-1$ nodes. In this scenario, there are no nodes that have 0 degrees, and the $\problematic$ measure $= n-0 = 0$. If the neighbouring graph $\mathcal{G}^\prime$ differs on the internal node, all edges of the graph are removed are the $\problematic = n$. The edge addition algorithm $\pi_\theta$ would play a minimal role here as $\theta$ could be equal to $n$.
\end{proof}

\subsubsection{Proof for \cref{thm:privacy_proof_dc_aware}}
\label{app:privacy_proof_dc_aware}

    Algorithm~\ref{algo:graph_general} with the optimized EM in  Algorithm~\ref{algo:em_opt} satisfies $(\epsilon_1 + \epsilon_2)$-DP.

\begin{proof}
    The total privacy budget is split in two ways for optimization and measure computation. These budgets can be composed using \cref{prop:DP-comp-post}. 
\end{proof}

\subsubsection{Proof for \cref{lemma:sens_quality_2stepEM}}
\label{app:sens_quality_2stepEM}
This lemma states that the sensitivity of $q_{\epsilon_2}(\mathcal{G}, \theta_i)$ in the 2-step EM (Algorithm~\ref{algo:em_opt})
defined in Equation~\eqref{eq:quality_function} is $\theta_{\max}= \min(\noisydegreebound,|V|)$ for 1st EM step and $\theta_{\max}=\theta^*$ for the 2nd EM step.

\begin{proof}
Using Lemma~\ref{lemma:sens_quality}, we know that the sensitivity of quality function $q_{\epsilon_2}(\mathcal{G}, \theta_i)$ is given by the max over all candidates in $\max(\Theta)$. For the first step of the 2-step EM, the maximum candidate apart from $|V|$ is given by $\min(\noisydegreebound,|V|)$. For the candidate $|V|$, the quality function $q$ differs. It only depends on the Laplace error $\frac{\sqrt{2}|V|}{\epsilon_2}$ and has no error from $e_{\text{bias}}$ as no edges are truncated due to $|V|$. For the 2nd step of EM, we truncate all values in the set greater than the output of the first step, i.e., $\theta^*$. Therefore, the sensitivity becomes $\theta^*$.
\end{proof}

\subsubsection{Utility Analysis for \cref{algo:graph_general}}\label{app:graph_general_utility}
\cref{thm:graph_general_utility} states that on any database instance $D$ and its respective conflict graph $\graph$, let $o$ be the output of Algorithm~\ref{algo:graph_general} with Algorithm~\ref{algo:expo_mech_basic} over $D$.  
Then,  with a probability of at least $1-\beta$, we have 
\begin{equation}
|o-a| \leq -\tilde{q}_{\opt}(D,\epsilon_2) + \frac{2 \theta_{\max}}{\epsilon_1} (\ln \frac{2|\Theta|}{|\Theta_{\opt}|\cdot \beta}) 
\end{equation}
where $a$ is the true inconsistency measure over $D$ and $\beta\leq \frac{1}{e^{\sqrt{2}}}$.

\begin{proof}
By the utility property of the exponential mechanism~\cite{mcsherry2007mechanism}, with at most probability $\beta/2$, Algorithm~\ref{algo:expo_mech_basic} will sample a bad $\theta^*$ with a  quality value as below
\begin{equation}
    q_{\epsilon_2}(\graph,\theta^*) \leq q_{\opt}(D,\epsilon_2) - \frac{2 \theta_{\max}}{\epsilon_1} (\ln \frac{2|\Theta|}{|\Theta_{\opt}|\beta})
\end{equation}
which is equivalent to 
\begin{equation}\label{eq:goodtheta}
    e_{\text{bias}}(\mathcal{G},\theta^*)  \geq -q_{\opt}(D,\epsilon_2) + \frac{2 \theta_{\max}}{\epsilon_1} (\ln \frac{2|\Theta|}{|\Theta_{\opt}|\beta}) -  \frac{\sqrt{2}\theta^*}{\epsilon_2}.
\end{equation}

With probability $\beta/2$, where $\beta\leq \frac{1}{e^{\sqrt{2}}}$,
we have 
\begin{equation}    
\text{Lap}(\frac{\theta^*}{\epsilon_2}) \geq
      \frac{\ln(1/\beta)\theta^{*}}{\epsilon_2} \geq \frac{\sqrt{2}\theta^*}{\epsilon_2}
      \end{equation}
Then, by union bound, with at most probability $\beta$, we have 
\begin{eqnarray}
   && |o-a| \nonumber\\
       &=& 
|f(\mathcal{G}_{\theta^*})+\text{Lap}(\frac{\theta^*}{\epsilon_2})-a|
            \nonumber  \\
    &\geq& a- f(\mathcal{G}_{\theta^*})+ \frac{\sqrt{2}\theta^*}{\epsilon_2}
 \nonumber \\
    &=& f(\mathcal{G})-f(\mathcal{G}_{\theta^*})+
     \frac{\sqrt{2}\theta^*}{\epsilon_2} \nonumber \\
    &=& f(\mathcal{G})-f(\mathcal{G}_{\theta_{\max}}) +
f(\mathcal{G}_{\theta_{\max}}) - f(\mathcal{G}_{\theta^*})
 +    \frac{\sqrt{2}\theta^*}{\epsilon_2} \nonumber\\    
    &=& f(\mathcal{G})-f(\mathcal{G}_{\theta_{\max}}) +
        e_{\text{bias}}(\mathcal{G},\theta^*)  + \frac{\sqrt{2}\theta^*}{\epsilon_2} \nonumber\\
     &\geq& -q_{\opt}(D,\epsilon_2) + f(\mathcal{G})-f(\mathcal{G}_{\theta_{\max}}) + \frac{2 \theta_{\max}}{\epsilon_1} (\ln \frac{2|\Theta|}{|\Theta_{\opt}|\beta})  \nonumber \\
      &=& -\tilde{q}_{\opt}(D,\epsilon_2) + \frac{2 \theta_{\max}}{\epsilon_1} (\ln \frac{2|\Theta|}{|\Theta_{\opt}|\beta})
\end{eqnarray}

\end{proof}

\subsection{DP Analysis for $\repair$ (\cref{algo:dp_vertexcover})} \label{app:repair}

\subsubsection{Proof for \cref{thm:vertex_cover_priv_util_analysis}}
Algorithm~\ref{algo:dp_vertexcover} satisfies $\epsilon$-node DP and always outputs the size of a 2-approximate vertex cover of graph $\mathcal{G}$.
\begin{proof}
    The privacy analysis of Algorithm~\ref{algo:dp_vertexcover} is straightforward as we calculate the private vertex cover using the Laplace mechanism with sensitivity $2$ according to Proposition~\ref{prop:vertexcover_sens}. The utility analysis can be derived from the original 2-approximation algorithm. The stable ordering $\Lambda$ in Algorithm~\ref{algo:dp_vertexcover} can be perceived as one particular random order of the edges and hence has the same utility as the original 2-approximation algorithm.
\end{proof}

\subsubsection{Proof for \cref{prop:vertexcover_sens}}\label{app:vertext_cover_sensitivity}
Algorithm~\ref{algo:dp_vertexcover} obtains a vertex cover, and its size has a sensitivity of 2.

\begin{proof}
Let's assume without loss of generality that
$\mathcal{G}^{\prime}=\left(V^{\prime}, E^{\prime}\right)$ has an additional node $v^{+}$compared to $\mathcal{G}=$ $(V, E)$, i.e., $V^{\prime}=V \cup\left\{v^{+}\right\}, E^{\prime}=E \cup E^{+}$, and $E^{+}$ is the set of all edges incident to $v^{+}$ in $\mathcal{G}^{\prime}$. We prove the theorem using a mathematical induction on $i$ that iterates over all edges of the global stable ordering $\Lambda$.

\underline{Base}: At step 0, the value of $c$ and $c'$ are both 0.

\underline{Hypothesis}: As the algorithm progresses at each step $i$ when the edge $e_i$ is chosen, either the edges of graph $\mathcal{G}'$ which is denoted by $E'_i$ has an extra vertex or the edge of graph $\mathcal{G}$ has an extra vertex. Thus, we can have two cases depending on some node $v^*$ and its edges $\{v^*\}$. Note that at the beginning of the algorithm, $v^*$ is the differing node $v^+$ and $\mathcal{G}'$ has the extra edges of $v^*/v^+$, but $v^*$ may change as the algorithm progresses. The cases are as illustrated below:
\begin{itemize}
    \item Case 1: $E_i$ does not contain any edges incident to $v^*$, $E'_i = E_i + \{ v^* \}$ and the vertex cover sizes at step $i$ could be $c_i = c'_i$ or $c_i = c'_i + 2$.
    \item Case 2: $E'_i$ does not contain any edges incident to $v^*$, $E_i = E'_i + \{ v^* \}$ and the vertex cover sizes at step $i$ could be $c_i = c'_i$ or $c'_i = c_i + 2$.
    \item Case 3: $E_i=E'_i$ and the vertex cover sizes at step $i$ is $c_i = c'_i$. This case occurs only when the additional node $v^+$ has no edges. 
\end{itemize}

\underline{Induction}: At step $i+1$, lets assume an edge $e_{i+1} = \{u, v\}$ is chosen. Depending on the $i^{th}$ step, we can have 2 cases as stated in the hypothesis.

\begin{itemize}
    \item Case 1 (When $E'_i$ has the extra edges of $v^*$): We can have the following subcases at step $i+1$ depending on $e_{i+1}$.
        \begin{enumerate}[label=\alph*),ref=\alph*]
            \item If the edge is part of $E'_{i}$ but not of $E_i$ ($e_{i+1} \in E'_{i} \setminus E_{i}$): Then $e_{i+1} = \{u, v\}$ should not exist in $E_i$ (according to the hypothesis at the $i$ step) and one of $u$ or $v$ must be $v^*$. Let's assume without loss of generality that $v$ is $v^*$. The algorithm will add $(u, v)$ to $C'$ and update $c'_{i+1} = c_i + 2$. Hence, we have either $c'_{i+1} = c_{i+1}$ or $c'_{i+1} = c + 2$.

            In addition, all edges of $u$ and $v/v^*$ will be removed from $E'_{i+1}$. Thus, we have $E_{i+1} = E'_{i+1} + \{u\}$, where $\{u\}$ represent edges of $u$. Now $u$ becomes the new $v^*$ and moves to Case 2 for the $i+1$ step.  
            
            \item If the edge is part of both $E'_i$ and $E_i$($e_{i+1} \in E'_i$ and $e_{i+1} \in E_i$): In this case $(u,v)$ will be added to both $C$ and $C'$ and the vertex sizes with be updated as $c_{i+1} = c_i + 2$ and $c'_{i+1} = c' + 2$. 

            Also, the edges adjacent to u and v will be removed from $E_i$ and $E'_i$. We still have $E'_i = E + {v^*}$ (the extra edges of $v^*$ and remain in case 1 for step i+1. 

            \item If the edge is part of neither $E'_i$ nor $E_i$ (If $e_{i+1} \in E'_i$ and $e_{i+1} \in E_i$): the algorithm makes no change. The previous state keep constant: $E'_{i+1} = E'_i, E_{i+1} = E_i$ and $c'_{i+1} = c'_i, c_{i+1} = c_i$. The extra edges of $v^*$ are still in $E'_{i+1}$.
        \end{enumerate}
        
    \item Case 2 (When $E_i$ has the extra edges of $v^*$) : This case is symmetrical to Case 1. There will be three subcases similar to Case 1 -- a) in which after the update the state of the algorithm switches to Case 1, b) in which the state remains in Case 2, and c) where no update takes place.  

    \item Case 3 (When $E_i = E'_i$): In this case, the algorithm progresses similarly for both the graphs, and remains in case 3 with equal vertex covers, $c_{i+1} = c'_{i+1}$.
\end{itemize}

Our induction proves that our hypothesis is true and the algorithm starts with Case 1 and either remains in the same case or oscillates between Case 1 and Case 2. Hence as per our hypothesis statement, the difference between the vertex cover sizes are upper bounded by 2.
\end{proof}

\subsection{Deferred experiment}
This experiment is similar to our runtime analysis experiment. In Figure~\ref{fig:time_datasets}, we fix the total number of nodes to 10k and noise to RNoise at $\alpha=0.001$ and vary the dataset. The x-axis is ordered according to the density of the dataset. We observe that the runtime is proportional to the density of the dataset and increases exponentially with the total number of edges in the graph. 
\begin{figure}
    \centering
    \includegraphics[width=0.6\linewidth]{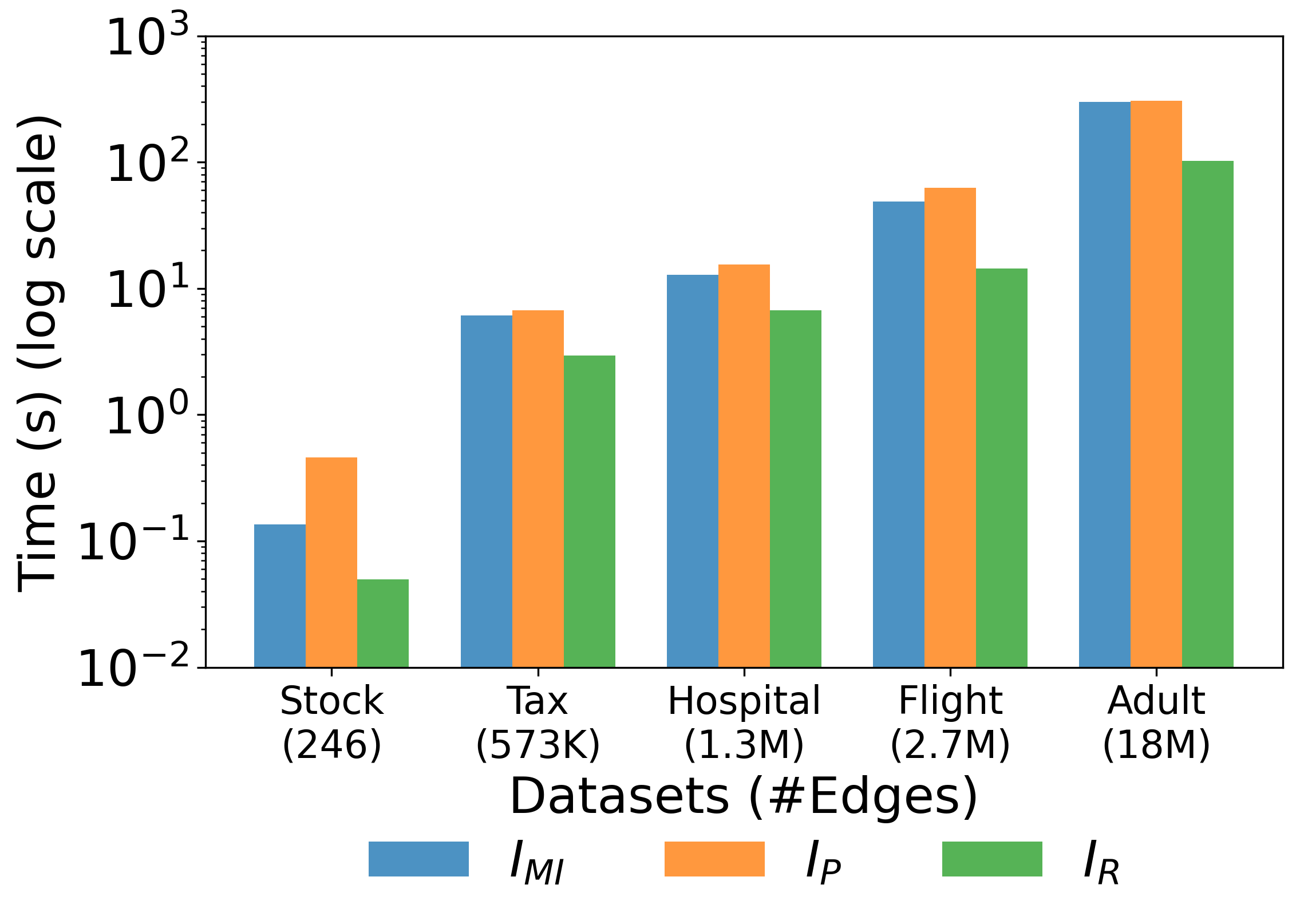}
    \caption{Runtime analysis for all measures by varying datasets.}
    \label{fig:time_datasets}
\end{figure}

\subsection{Utility vs $\beta$ for Theorem~\ref{thm:graph_general_utility}}
In Figure~\ref{fig:utility_vs_beta}, we show the trend of the utility analysis according to Theorem~\ref{thm:graph_general_utility} by varying the $\beta$ parameter. We follow Example~\ref{example:running_example} and set $\epsilon_1$ and $\epsilon_2$ to 1. The figure confirms that as the value of $\beta$ increases, there is a decrease in the distance between the true answer and the output, i.e., the utility increases. However, with higher values of $\beta$, the utility analysis holds with lesser probability, i.e., $1-\beta$. 

\begin{figure}
    \centering
    \includegraphics[width=0.6\linewidth]{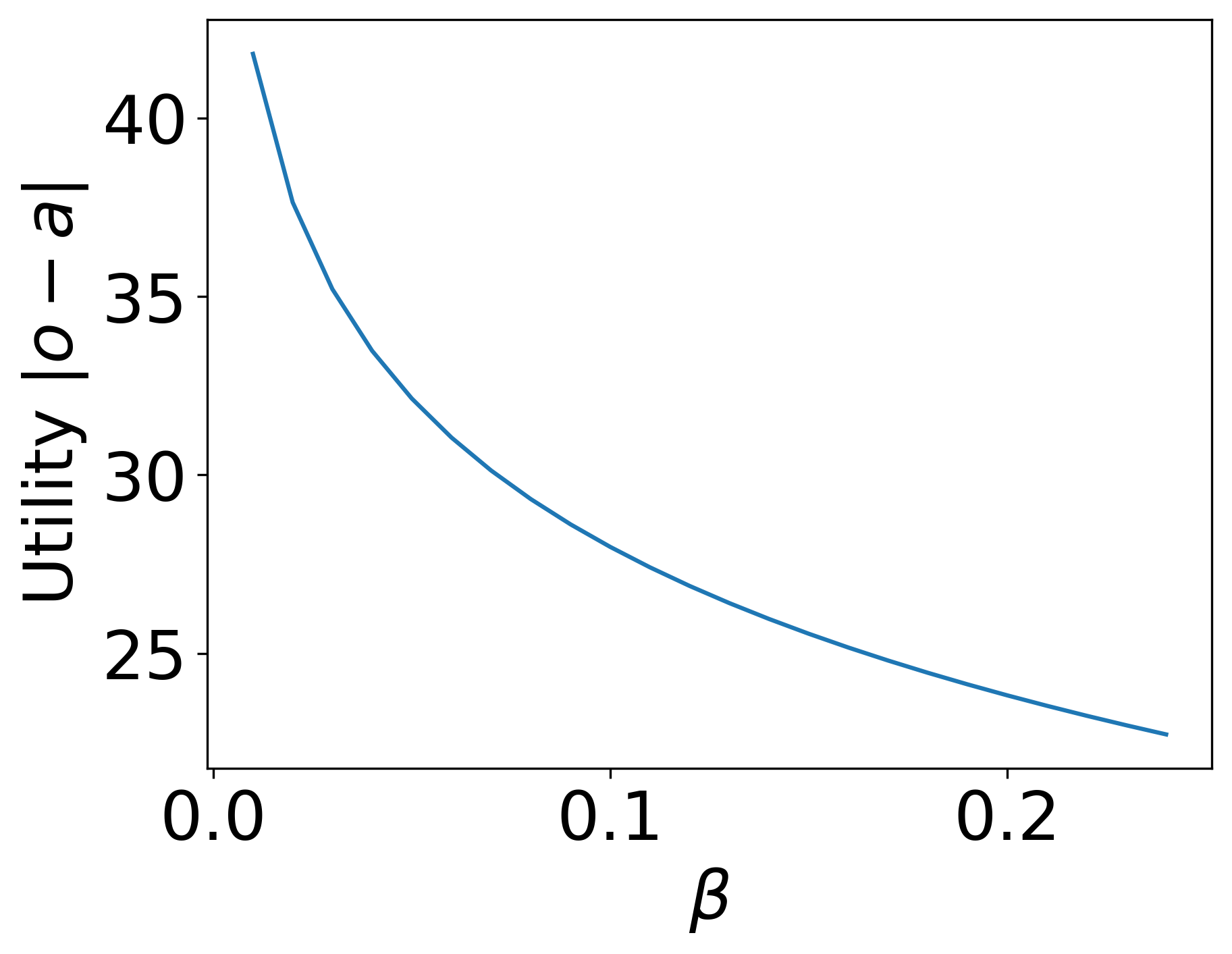}
    \caption{Trend of the utility analysis vs $\beta$}
    \label{fig:utility_vs_beta}
\end{figure}
\fi 

\end{document}